\documentclass[
oneside,
reqno,
10pt
]
{amsart}

\usepackage[utf8]{inputenc}
\usepackage{multirow} 
\usepackage{amsmath}
\usepackage{amsfonts}
\usepackage{amssymb}
\usepackage{tikz}
\usetikzlibrary{arrows,calc}
\usepackage{amsthm}

\usepackage[a4paper]{geometry}
\usepackage{color}
\usepackage{latexsym}
\usepackage[colorlinks=true,linkcolor=blue,citecolor=mauve, breaklinks=true]{hyperref}
\usetikzlibrary{decorations.fractals}
\usetikzlibrary{patterns}
\usepackage{caption}
\usepackage{diagbox}

\definecolor{mauve}{rgb}{0.58,0,0.82}

\tikzset{
>=stealth',
help lines/.style={dashed, thick},
axis/.style={<->},
important line/.style={thick},
connection/.style={thick, dotted},
}

\newcommand{\supp}{\mathrm{supp}\,}

\newcommand{\of}{\overline{f}}

\let\Re\relax
\DeclareMathOperator{\Re}{Re}

\def\R{\mathbb R}

\def\i{\mathrm i}

\usepackage{esint}

\let\phi\varphi
\let\epsilon\varepsilon
\newtheorem{thm}{Theorem}[section]
\newtheorem{prop}[thm]{Proposition}

\newtheorem{lem}[thm]{Lemma}
\newtheorem{cor}[thm]{Corollary}
\theoremstyle{definition}
\newtheorem{df}[thm]{Definition}

\newtheorem*{rem}{Remark}

\numberwithin{equation}{section}

\date{}

\author{Simon Barth $^{1,a)}$, Andreas Bitter $^{2,b)}$ and Semjon Vugalter $^{2,c)}$}

\title[]{\ \textbf{The absence of the Efimov effect
in systems \\ of one- and two-dimensional particles} \ \\ \ \\ }
\begin{document}
\maketitle

\section*{Affiliations}$ $\\
\small{
$^1$ Institute of Analysis, Dynamics and Modeling, University of Stuttgart, Germany\\
$^2$ Institute for Analysis, Karlsruhe Institute of Technology, Germany\\
$^{A)}$ {Email: simon.barth@mathematik.uni-stuttgart.de}\\
$^{B)}$ Email: andreas.bitter@kit.edu\\
$^{C)}$ Email: semjon.wugalter@kit.edu}

\normalsize

\begin{abstract}
We study virtual levels of $N$-particle Schr\"odinger operators and prove that if the particles are one-dimensional and $N\ge 3$, then virtual levels at the bottom of the essential spectrum correspond to eigenvalues. The same is true for two-dimensional particles if $N\ge 4$. These results are applied to prove the non-existence of the Efimov effect in systems of $N\ge 4$ one-dimensional or $N\ge 5$ two-dimensional particles.\\
\end{abstract}
\pagestyle{headings}
\renewcommand{\sectionmark}[1]{\markboth{#1}{}}

%\tableofcontents

\section{Introduction}
In recent years the Efimov effect has attracted large interest.  
This effect is named after the physicist V. Efimov and can be stated as follows: 
A system of three quantum particles in dimension three, interacting through attractive short-range potentials, has an infinite number of bound states if the subsystems do not have negative spectrum and at least two of them are resonant \cite{Efimov}, i.e. any arbitrarily small negative perturbation of the pair potential leads to a negative spectrum. In this situation we also say that the two-body Hamiltonian of the system has a virtual level.\\
 The Efimov effect is a surprising phenomenon, because although the pair interactions are short-range, the system of three particles behaves as a system of two particles with long-range potential. Another interesting feature is its universality. This means that the existence of the effect as well as the distribution of the eigenvalues do not depend on the shape of the potentials of interaction of particles. It is only important that they are short-range and resonant. Indeed, the counting function $N(z)$ of the eigenvalues of the three-body Hamiltonian below $z <0$ obeys the following asymptotics 
 \begin{equation}\label{eq: asymptotics Sobolev}
	\lim\limits_{z\rightarrow 0-} \frac{N(z)}{|\ln |z| |}=\mathcal{A}_0>0,
\end{equation}
where the constant $\mathcal{A}_0$ depends on the masses of the particles, but not on the potentials \cite{Efimov}.
\\
For a long time the Efimov effect was regarded by many as a theoretical peculiarity. After the theoretical discovery of the Efimov effect great efforts were made to verify it experimentally. However, it took more than 30 years before in 2002 it was found in an ultra-cold gas of caesium atoms \cite{Grimm}.  This experiment was a milestone and opened the way to many further experiments in different systems of ultra-cold atoms in many laboratories all over the world \cite{Ferlaino, Gross, Berninger}.\\
In addition, it lead to a resurgence of interest to the Efimov effect, see for example the review of P. Naidon and S. Endo \cite{Naidon2017}, which contains 400 references. Since then, many predictions of phenomena similar to the Efimov effect have been made \cite{NishidaFiveBosons, NishidaLiberating, NishidaSemisuper, NishidaSuper}. Some of these predictions focus on the question whether an Efimov-type effect can be found in $N$-particle systems consisting of one- or two-dimensional particles under the assumptions that the $(N-1)$-particle subsystems have a virtual level at the bottom of the essential spectrum. For example in \cite{NishidaSemisuper} and \cite{NishidaFiveBosons} it is predicted that such an effect occurs for systems of $N=4$ two-dimensional or $N=5$ one-dimensional particles if the interactions in subsystems of less than $N-1$ particles are absent. In the table below we give a list of systems of $N$ identical particles where the Efimov effect is expected in the physics literature.
\begin{table}[h]
\begin{tabular}[h]{ccccc}
\hline
\multirow{2}{*}{\textbf{System}}& \textbf{Art of}  &{\textbf{Resonant}} & \multirow{2}{*}{\textbf{Predicted in}} & \textbf{Does a mathematical }\\ & \textbf{Interactions}&\textbf{subsystems} & & \textbf{proof exist?}\\
\hline
\multirow{3}{*}{$5$ bosons, $d=1$} & four-body &  \multirow{3}{*}{four-body} & \multirow{3}{*}{\cite{NishidaFiveBosons}} & \multirow{3}{*}{No} \\
&  no two-body&  & &\\ &  no three-body&  & &\\ \hline
\multirow{2}{*}{$4$ bosons, $d=2$} & three-body &\multirow{2}{*}{three-body }& \multirow{2}{*}{\cite{NishidaSemisuper}} & \multirow{2}{*}{No} \\
&  no two-body&  & &\\ \hline
$3$ bosons, $d=3$ & two-body &two-body & \cite{Efimov}& Yes, \cite{YafaevEfimov}\\ 
$3$ fermions, $d=2$ & two-body & two-body & \cite{NishidaSuper}&Yes, \cite{GridnevSuper}\\
\hline\\
\end{tabular}
\caption{Systems for which an Efimov-type effect is expected by physicists}
\end{table}

Besides these systems, effects similar to the Efimov effect are expected for systems with mixed dimensions, i.e. where the particles move in the three-dimensional space, but some of them are confined in a lower-dimensional space. Such an effect is called confinement-induced Efimov effect, see for example \cite{NishidaLiberating}. From a mathematical point of view the question of existence or non-existence of Efimov-type effects in systems with mixed dimensions is completely open. 

The first mathematically rigorous proof of the Efimov effect for systems of three three-dimensional particles was given by D. R. Yafaev \cite{YafaevEfimov} using a system of symmetrized Faddeev equations combined with the low-energy asymptotics of the resolvents of the two-body Hamiltonians. In \cite{JafaevAbsence} he also showed that the Hamiltonian has only a finite number of bound states if at most one of the subsystems is resonant.
Later, A. V. Sobolev proved the asymptotics \eqref{eq: asymptotics Sobolev} for the eigenvalue counting function \cite{Sobolev}. In the years after the mathematical confirmation of the Efimov effect different techniques were developed and many other mathematical results related to this effect were obtained \cite{Semjon2, Ovchinnikov, Tamura, Tamura2, second, SemjonMagnetic}.
\\
In particular, it was proved in \cite{Semjon2} that the existence of the Efimov effect depends on the nature of the virtual levels in the subsystems. If the virtual level in the two-body subsystems correspond to eigenvalues, which for example is the case if the three-particle Hamiltonian is considered on certain symmetry subspaces of $L^2(\mathbb{R}^6)$, then the Efimov effect is absent.  \\
For a long time it has been expected that due to the same reason the Efimov effect does not exist for systems of $N\ge 4$ three-dimensional bosons. However, to prove that virtual levels in subsystems of $N-1$ particles correspond to eigenvalues and not to resonances was a very challenging problem, because the sum of the pair potentials does not decay in all directions at infinity, which makes it difficult to use Green's functions. This problem was first solved by D. Gridnev \cite{Gridnev2} and recently a proof with simpler methods and less restrictions on the potentials was given in \cite{second}. In addition, it was shown in \cite{second} that the Efimov effect can not occur in systems of $N\ge 4$ one- or two-dimensional spinless fermions.\\
For the case of three two-dimensional spinless fermions, which is not covered by \cite{second}, it was predicted in the physics literature that an effect similar to the Efimov effect is present, namely the so-called super Efimov effect \cite{NishidaSuper}. The first mathematical proof of this prediction was given by D. Gridnev \cite{GridnevSuper}. Recently, this result was improved by H. Tamura, where the conditions on the potentials \cite{TamuraSuper} were less restrictive.\\
Much less is mathematically known about the existence of the Efimov effect in systems of $N\ge 3$ one- or two-dimensional bosons or systems without symmetry restrictions. For such systems consisting of three one- or two-dimensional particles the absence of the Efimov effect was first proved by G. Zhislin and one of the authors of this paper in \cite{Semjon4} under very strong restrictions on the potentials.
Later, in \cite{Semjon3} these restrictions were relaxed, but unfortunately Lemma 1 in \cite{Semjon3} contains a mistake. We will correct this mistake in Section \ref{Section: No Efimov} at the end of this paper. For systems of $N\ge 4$ one- or two-dimensional bosons mathematical results are unknown. The main goal of our work is to fill this gap, at least partially. \\
Our main results are the following: For systems of $N\ge 3$ one-dimensional bosons or particles without symmetry restrictions with pair interactions we prove that the existence of virtual levels in the $(N-1)$-particle subsystems does not imply the infiniteness of the number of negative eigenvalues. For systems of $N$ two-dimensional particles we prove the same result except for $N=4$.
\\ The method of the proof is analogous to the proof in \cite{second}. We study the decay of solutions of the Schr\"odinger equation corresponding to a virtual level and show that these solutions are eigenfunctions. Then we use arguments similar to \cite{Semjon2}. To obtain the decay rate of the solutions we apply a modification of Agmon's method \cite{Ag}, developed in \cite{second}. This method requires estimates on the quadratic form of a multi-particle Schr\"odinger operator on functions supported far from the origin. In order to obtain these estimates we make a partition of unity of the configuration space according to decompositions of the original system into clusters with careful estimates of the localization error.
\\ On the technical level however this work is very different from \cite{second}. A crucial difference between one- or two-dimensional particles and $d$-dimensional particles with $d\geq 3$ is that in lower dimensions the common Hardy inequality does not hold. This manifests in particular in the fact that the one-particle Schr\"odinger operator $h=-\Delta+V$ in dimension one or two with a short-range potential $V\not\equiv 0$ has negative eigenvalues if $\int V(x)\, \mathrm{d}x\leq 0$. Consequently, if we know that $h$ does not have negative spectrum, we can immediately say that $\int V(x)\, \mathrm{d}x >0$. This simple observation plays an important role in our proof.
\\ On the other hand to localize regions in the configuration space in \cite{second} we used a special type of the cut-off functions constructed in \cite{Semjon2}. For this choice of the cut-off functions, due to Hardy's inequality, the localization error can be compensated by a small part of the kinetic energy. Since in one- and two-dimensional cases the Hardy inequalities are different, this construction can not be applied in lower dimensions. To overcome this obstacle, we develop in Section \ref{section: virtual levels multi-particle} an advanced way to construct the cut-off functions, which is better compatible with one- and two-dimensional variants of Hardy's inequality. \\
The paper is organized as follows. In Section \ref{Section Shortrange} we discuss virtual levels of one-body Schr\"odinger operators with short-range potentials in dimension one and two. We prove that virtual levels correspond to resonances and give an estimate for the decay rate of the corresponding solutions. This section is contained for completeness. Readers only interested in the multi-particle case can skip it and go immediately to Section \ref{section: virtual levels multi-particle}, where we extend the study of virtual levels to the multi-particle case. We prove that for systems of $N\ge 3$ one-dimensional or $N\ge 4$ two-dimensional particles virtual levels correspond to eigenvalues. We also derive lower bounds for the decay rates of the zero energy eigenfunctions. In Section \ref{section: d=2 N=3} we discuss systems of three two-dimensional particles, which is the only case where a virtual level might correspond to a resonance. We show that there exists a solution of the Schr\"odinger equation in the space $L^2(\mathbb{R}^4,(1+|x|)^{-\delta}),\ \delta>0$. Section \ref{Section No Efimov multi} is devoted to the absence of the Efimov effect for multi-particle systems in dimension one and two. In Section \ref{Section: No Efimov} we give the proof for the absence of the Efimov effect in systems of three one- or two-dimensional particles.
\section{Virtual levels of one-particle Schr\"odinger operators in dimension one and two}\label{Section Shortrange}
Although the main subject of this paper are virtual levels of multi-particle systems consisting of one- or two-dimensional particles and the non-existence of the Efimov effect in such systems, in this section we discuss virtual levels of one-particle Schr\"odinger operators in these dimensions. Some of the results of this section will be applied later to study the multi-particle case, others are given for a better understanding of one- and two-dimensional systems. 
\subsection{Notation and assumptions}
In this section we consider the one-particle Schr\"odinger operator
\begin{equation}\label{eq: definition operator h0}
	h =-\Delta +V
\end{equation}
in $L^2(\mathbb{R}^d)$ with $d=1$ or $d=2$. Within the whole section we assume that $V\neq 0$. Furthermore, we assume that $V$ is relatively form bounded with relative bound zero, i.e. for every $\varepsilon>0$ there exists a constant $C(\varepsilon)>0$, such that
\begin{equation}\label{cond: relatively form bounded}
	\langle |V|\psi, \psi\rangle \le \varepsilon \Vert \nabla \psi\Vert^2+C(\varepsilon) \Vert \psi\Vert^2
\end{equation}
holds for any $\psi\in H^1(\mathbb{R}^d)$. This condition is fulfilled if $V\in L^p_{\mathrm{loc}}(\mathbb{R}^d)+L^\infty(\mathbb{R}^d)$ with $p=1$ if $d=1$ and $p>1$ if $d=2,$ \cite{BrezisKato}. Due to the KLMN-theorem \cite[Theorem X.17]{reed} under this assumption the operator $h$ is self adjoint on $L^2(\mathbb{R}^d)$ with the associated quadratic form
\begin{equation}
Q[\psi]=\Vert \nabla \psi \Vert^2+\langle V\psi,\psi\rangle
\end{equation}
with form domain $H^1(\mathbb{R}^d)$. For any $\varepsilon \in (0,1)$ we define
\begin{equation}
h_{\varepsilon}=h+\varepsilon \Delta.
\end{equation}
For any self-adjoint operator $\mathcal A$ we denote by $\mathcal{S}(\mathcal A)$, $\mathcal{S}_{\mathrm{ess}}(\mathcal A)$ and $\mathcal{S}_{\mathrm{disc}}(\mathcal A)$ the spectrum, the essential spectrum and the discrete spectrum of $\mathcal A$, respectively. \\
Following \cite{birman} we introduce function spaces which will be important for our studies in this paper. For dimension $d\ge3$ the homogenous Sobolev space $\tilde{H}^1(\mathbb{R}^d)$ is defined as the completion of $C_0^\infty(\mathbb{R}^d)$-functions with respect to the norm
\begin{equation}\label{eq: norm H punkt}
	\Vert u\Vert_{\tilde{H}^1}=\left(\int_{\mathbb{R}^d} |\nabla u|^2\,\mathrm{d}x\right)^\frac{1}{2}.
\end{equation}
It follows from Hardy's inequality that for $d\ge 3$ a sequence of functions $u_n\in \tilde H^1(\mathbb{R}^d)$ with $\Vert u_n\Vert_{\tilde{H}^1}\rightarrow 0$ converges to zero in $L^2_{\mathrm{loc}}(\mathbb{R}^d)$. It is also known (see, for example \cite{birman}) that for $d=1$ and $d=2$ the completion of $C_0^\infty(\mathbb{R}^d)$ with respect to \eqref{eq: norm H punkt} does not lead to a function space, because constant functions are identified. In order to avoid this problem we add a local $L^2$ norm to the gradient semi-norm and define for $d=1$ or $d=2$
\begin{equation}\label{eq: norm tilde H d=2}
\Vert u\Vert_{\tilde{H}^1} = \left(\int_{\mathbb{R}^d} |\nabla u|^2\,\mathrm{d}x+\int_{\{|x|\le 1\}}|u|^2\,\mathrm{d}x\right)^{\frac{1}{2}}.
\end{equation}
Let $\tilde{H}^1(\mathbb{R}^d)$ be the completion of $C_0^\infty(\mathbb{R}^d)$ with respect to the norm \eqref{eq: norm tilde H d=2}, then
\begin{equation}
\tilde{H}^1(\mathbb{R}^d)=\left\{u \in L^1_{\mathrm{loc}}(\mathbb{R}^d), \nabla u \in L^2(\mathbb{R}^d)\right\}.
\end{equation}
In Appendix \ref{Appendix A} we collect some elementary properties of the space $\tilde H^1(\mathbb{R}^d)$, $d=1,2$, which we use in this paper. 
\subsection{Properties of virtual levels of one-particle Schr\"odinger operators with short-range potentials}
\begin{df}\label{df: Virtual level one particle}
Assume that the potential $V$ satisfies \eqref{cond: relatively form bounded}. We say that the operator $h$, defined in \eqref{eq: definition operator h0}, has a virtual level at zero if 
\begin{equation}
h \geq 0, \qquad \inf\mathcal{S}\left(h_\varepsilon\right)< 0 \qquad \text{and}\qquad \inf\mathcal{S}_{\mathrm{ess}}\left(h_\varepsilon\right)= 0
\end{equation}
holds for any sufficiently small $\varepsilon>0$. 
\end{df}
\begin{rem}
Note that the Laplace operator is critical in dimension one and two, i.e. for any $V\in L^1(\mathbb{R}^d)$ satisfying \eqref{cond: relatively form bounded} and $V\neq 0$ with $\int_{\mathbb{R}^d} V(x)\, \mathrm{d}x \leq 0$ the operator $h$ has at least one negative eigenvalue, see \cite{SimonThresold}. Consequently, for $d=1,2$ the condition $h \geq 0$ implies $\int_{-\infty}^\infty V(x)\, \mathrm{d}x>0$. On the other hand, the condition $\inf \mathcal{S}(h_\varepsilon)<0$ yields that $V(x)$ has a non-trivial negative part. 
\end{rem}
Let us briefly motivate our goals for this section. For the case $d=3$ it was shown that the one-particle Schr\"odinger operator $h\ge 0$ with short-range potential has a virtual level if and only if $h\psi=0$ has a solution in $\tilde{H}^1\left(\mathbb{R}^3\right)$. This solution does not belong to $L^2\left(\mathbb{R}^3\right)$ and decays as $|x|^{-1}$ as $|x|\rightarrow \infty$, see \cite{YafVL}. \\
Moreover, by applying Hardy's inequality one can see that in case $d\ge 3$ for short-range potentials the operator $h$ has a virtual level at zero if and only if $h\ge 0$ and for any $\varepsilon>0$ the operator $\tilde h=-\Delta+V-\varepsilon(1+|x|)^{-2}$ has a discrete eigenvalue below zero.\\
Our goal is to generalize these two statements to the cases $d=1,2$. This will be done in the following two theorems.
\begin{thm}[Solutions of the Schr\"odinger equation corresponding to virtual levels]\label{thm: abstraktes Theorem} Assume that $d=1$ or $d=2$ and that the potential $V$ satisfies {$V\neq 0$}, condition \eqref{cond: relatively form bounded} and
\begin{equation}\label{cond: decay at infinity}
	|V(x)|\le C\left(1+|x|\right)^{-2-\nu}, \qquad x\in \mathbb{R}^d, \quad |x|\ge A
\end{equation}
for constants $A,C,\nu>0$. If $h$ has a virtual level at zero, then the following assertions hold:
\begin{enumerate}
\item[\textbf{(i)}]
There exists a solution $\varphi_0 \in \tilde{H}^1 (\mathbb{R}^d),\ \varphi_0\neq 0$, of the equation $-\Delta \varphi_0+V\varphi_0=0$, i.e. for all $\psi \in \tilde{H}^1(\mathbb{R}^d)$
\begin{equation}\label{eq: weak solution phi0}
\langle \nabla\varphi_0,\nabla\psi \rangle + \langle V\varphi_0,\psi\rangle = 0.
\end{equation}
\item[\textbf{(ii)}]  Let $d=1$. Then for the functions $\varphi_0$ satisfying \eqref{eq: weak solution phi0} we have 
\begin{equation}\label{eq: estimate phi0 n=1 N=1}
	 (1+|x|)^{-\frac{1}{2}-\varepsilon}\varphi_0 \in L^2 (\mathbb{R}) \qquad \text{for any} \ \varepsilon>0.
\end{equation}
\item[\textbf{(iii)}] Let $d=2$.  Then for the functions $\varphi_0$ satisfying \eqref{eq: weak solution phi0} we have 
\begin{equation}\label{eq: estimate phi0 n=2 N=1}
	 (1+|x|)^{-1}\left(1+|\ln(|x|)|\right)^{-\frac{1}{2}-\varepsilon}\varphi_0 \in L^2 (\mathbb{R}^2) \qquad \text{for any}\ \varepsilon>0.
\end{equation}
\item[\textbf{(iv)}] If in addition the potential $V$ is relatively $-\Delta$-bounded, i.e. there exists a constant $C>0$, such that
\begin{equation}\label{cond: rel. operator bounded}
	\Vert V\psi\Vert^2 \le C\left( \Vert \Delta \psi\Vert^2 + \Vert \psi\Vert^2\right)
\end{equation}
holds for all functions $\psi\in H^2(\mathbb{R}^d)$, then there exists a constant $\delta_0>0$, such that for any function $\psi \in H^1(\mathbb{R}^d)$ satisfying $\langle \nabla\psi,\nabla \varphi_0\rangle =0$ 
\begin{equation}\label{eq: orthogonality subtract short-range}
\langle h \psi,\psi \rangle \geq \delta_0 \Vert \nabla \psi \Vert^2.
\end{equation}
\end{enumerate}
\end{thm}
\begin{rem}
\begin{enumerate}
\item[\textbf{(i)}]
Note that the left-hand side of \eqref{eq: weak solution phi0} is well-defined due to condition \eqref{cond: decay at infinity} and inequalities \eqref{eq: weidl inequality in 1} and \eqref{eq: weidl inequality in 2} in Appendix \ref{Appendix A}.
\item[\textbf{(ii)}] Similarly to Theorem 2.1 in \cite{second} we use the condition \eqref{cond: rel. operator bounded} on the potential $V$ only to be able to apply the unique continuation theorem. Without this condition we are not able to prove uniqueness of the solution $\varphi_0$ of the equation $-\Delta \varphi_0+V\varphi_0=0$ in $\tilde{H}^1(\mathbb{R}^d)$.
However, similarly to \cite{second} we can show that the subspace $\mathcal{M}\subset \tilde{H}^1(\mathbb{R}^d)$ of functions $\varphi$ satisfying \eqref{eq: weak solution phi0} is finite-dimensional and that for each $\psi\in \tilde{H}^1(\mathbb{R}^d)$, satisfying $\langle \nabla \varphi,\nabla\psi\rangle=0$ for all $\varphi\in \mathcal{M}$ holds \eqref{eq: orthogonality subtract short-range}.
\item[\textbf{(iii)}] 
Theorem \ref{thm: abstraktes Theorem} gives a lower bound on the decay rate of solutions of the Schr\"odinger equation corresponding to virtual levels. It is easy to see that if the potentials are compactly supported and $V(x)=V(|x|)$, then estimates \eqref{eq: estimate phi0 n=1 N=1} and \eqref{eq: estimate phi0 n=2 N=1} are almost sharp. It is also easy to see that the solution can not be an eigenfunction, it is a zero energy resonance. 
\end{enumerate}
\end{rem}
\begin{thm}[Necessary and sufficient condition for a virtual level]\label{lem no virtual level subtract d=1}
Let $d=1,2$. We assume that $V\neq 0$ satisfies \eqref{cond: relatively form bounded} and \eqref{cond: decay at infinity} and that $h \geq 0 $. Further, let $\mathcal{U}$ be a continuous, strictly negative potential satisfying for $|x|\ge A$ the condition 
\begin{align}\label{eq: stoerung durch u virtual level}
\begin{split}
  |\mathcal{U}(x)|\le C|x|^{-2}  \ \text{if}\ d=1\quad \text{and}\quad |\mathcal{U}(x)|\le C|x|^{-2}\ln^{-2}(|x|)  \ \text{if } d=2
\end{split}
\end{align}
for some $A,C>0$. 
Then $h$ has a virtual level at zero if and only if for any $\varepsilon>0$ we have
\begin{align}\label{eq: no virtual level subtract d=1}
\inf\mathcal{S}\left(h+\varepsilon \mathcal{U}\right)<0.
\end{align}
\end{thm}
\begin{rem}\begin{enumerate}
\item[\textbf{(i)}]
Note that in dimension $d\ge 3$ Hardy's inequality yields that $\inf \mathcal{S}(h+\varepsilon \mathcal{U})=0$ for sufficiently small $\varepsilon>0$ if $h$ does not have a virtual level. For dimension $d=1$ or $d=2$ it does not follow from Hardy's inequality. However, Theorem \ref{lem no virtual level subtract d=1} shows that it is still true.
\item[\textbf{(ii)}]Assume that $V\neq 0$ satisfies \eqref{cond: relatively form bounded} and \eqref{cond: decay at infinity}, $h \geq 0 $ and that $h$ does not have a virtual level at zero. Then for small $\varepsilon_0>0$ the operator $h_{\varepsilon_0}>0$ also does not have a virtual level and therefore Theorem \ref{lem no virtual level subtract d=1} can be applied to the operator $h_{\varepsilon_0}$.
Hence, there exists $\varepsilon_1>0$, such that
\begin{equation}\label{eq: equivalent definition vl rem}
(1-\varepsilon_0)\Vert \nabla\psi\Vert^2 +\langle V\psi,\psi \rangle + \varepsilon_1\langle	 \mathcal{U}\psi,\psi\rangle \geq 0
\end{equation}
holds for any function $\psi \in H^1(\mathbb{R}^d)$ with $\mathcal U$ defined as in Theorem \ref{lem no virtual level subtract d=1}. This modification of Theorem \ref{lem no virtual level subtract d=1} will be used in the case of multi-particle systems in the next sections.
\item[\textbf{(iii)}] Assume that $h\ge 0$ does not have a virtual level and that the potential $V\neq 0$ satisfies \eqref{cond: relatively form bounded} and \eqref{cond: decay at infinity}. Then by choosing $\mathcal{U}$ according to Theorem \ref{lem no virtual level subtract d=1} with $\mathcal{U}(x)=-1$ for $|x|\le 1$ we obtain from \eqref{eq: equivalent definition vl rem} that for any $\psi\in \tilde{H}^1(\mathbb{R}^d)$ 
\begin{equation}
	\Vert \psi\Vert^2_{\tilde{H}^1} \le \frac{1+\varepsilon_1-\varepsilon_0}{\varepsilon_1} \Vert\nabla \psi\Vert^2+\frac{1}{\varepsilon_1}\langle V\psi,\psi\rangle.
\end{equation}
\end{enumerate}
\end{rem}
Theorem \ref{lem no virtual level subtract d=1} will be proved in Appendix \ref{Appendix Virtual Levels}, where similar statements for multi-particle Schr\"odinger operators are established. We turn to the 
\begin{proof}[Proof of Theorem \ref{thm: abstraktes Theorem}]
 Since for any $\varepsilon>0$ we have $\inf\mathcal{S}_{\mathrm{disc}}(h_\varepsilon)<0$, we find a sequence of eigenfunctions $\psi_n \in H^1(\mathbb{R}^d)$, corresponding to eigenvalues $E_n<0$ of the operator $h_{n^{-1}}$, i.e.
\begin{equation}\label{1: eigenfunctions}
-\left( 1-n^{-1}\right)\Delta \psi_n+V\psi_n=E_n\psi_n.
\end{equation}
We normalize the functions $\psi_n$ by the condition $\Vert \psi_n\Vert_{\tilde{H}^1}=1$. Then there exists a subsequence, also denoted by $(\psi_n)_{n\in \mathbb{N}}$, which converges weakly in $\tilde{H}^1(\mathbb{R}^d)$ to a function $\varphi_0\in \tilde{H}^1(\mathbb{R}^d)$. At first, we prove that $\varphi_0$ is a solution of the equation $-\Delta\varphi_0+V\varphi_0=0$ in $\tilde{H}^1(\mathbb{R}^d)$ and that $\varphi_0\neq 0$. Indeed, we have the following
\begin{lem}\label{lem: existence of solution d=1 N=1}
Assume that $h$ has a virtual level at zero and that $V$ satisfies \eqref{cond: relatively form bounded} and \eqref{cond: decay at infinity}. Then the function $\varphi_0$ defined above is not zero and for any $\psi \in \tilde{H}^1(\mathbb{R}^d)$ 
\begin{equation}\label{1: stand.}
\langle  \nabla \varphi_0,\nabla \psi \rangle + \langle V\varphi_0,\psi\rangle = 0.
\end{equation}
\end{lem}
\begin{proof}[Proof of Lemma \ref{lem: existence of solution d=1 N=1}]
Since $\varphi_0$ is the weak limit of the sequence $(\psi_n)_{n\in \mathbb{N}}$ in $\tilde{H}^1(\mathbb{R}^d)$, we have $\varphi_0\in L^2_{\mathrm{loc}}(\mathbb{R}^d)$ and by Proposition \ref{prop: Homogenous Sobolev} \textbf{(iii)} $\psi_n\rightarrow \varphi_0$ strongly in $L^2_{\mathrm{loc}}(\mathbb{R}^d)$. First, we show that 
\begin{equation}\label{eq: local convergence V psin to Vphi0}
\int_{\{|x|\leq R\}} V(x)|\psi_n(x)|^2\, \mathrm{d}x \longrightarrow \int_{\{|x|\leq R\}} V(x)|\varphi_0(x)|^2\, \mathrm{d}x \quad \text{as} \quad n\rightarrow \infty
\end{equation}
for any fixed $R>0$. We write 
\begin{equation}\label{eq: decomposition Vpsin minus Vphi0}
	\langle V \psi_n ,\psi_n\rangle_{B(R)} - \langle V \varphi_0 ,\varphi_0\rangle_{B(R)}= \langle V (\psi_n-\varphi_0) ,\psi_n\rangle_{B(R)}+\langle V \varphi_0, (\psi_n-\varphi_0)\rangle_{B(R)},
\end{equation}
where $B(R)=\{x\in \mathbb{R}^d: |x|\le R\}$.
Let $\chi$ be a piecewise differentiable function satisfying $\chi(x)=1$ for $x\in B(R)$ and $\chi(x)=0$ if $x\notin B(R+1)$. Then we get by Cauchy Schwarz
\begin{align}\label{eq: decomposition Vpsin-Vphi0}
\begin{split}
	\langle |V| |\psi_n-\varphi_0| ,|\psi_n|\rangle_{B(R)} &\le \langle |V|^{\frac{1}{2}} |\psi_n-\varphi_0|\chi ,|V|^{\frac{1}{2}}|\psi_n|\chi\rangle\\
	& \le \left(\langle |V| |\psi_n-\varphi_0|\chi,|\psi_n-\varphi_0|\chi\rangle\right)^{\frac{1}{2}}\left(\langle |V| |\psi_n|\chi,|\psi_n|\chi\rangle\right)^{\frac{1}{2}}.
	\end{split}
\end{align}
We estimate the two factors on the r.h.s. of \eqref{eq: decomposition Vpsin-Vphi0} separately. By assumption \eqref{cond: relatively form bounded} we get
\begin{equation}\label{eq: estimate Vij psin-phi0}
	\langle |V| |\psi_n-\varphi_0|\chi,|\psi_n-\varphi_0|\chi\rangle\le \varepsilon \Vert \nabla_0\left(|\psi_n-\varphi_0|\chi\right)\Vert^2+C(\varepsilon) \Vert (\psi_n-\varphi_0)\chi\Vert^2.
\end{equation}
Due to $\Vert \nabla_0\psi_n\Vert\le 1,\ \Vert \nabla_0\varphi_0\Vert\le 1, \ 0\le \chi\le 1$ and $| \nabla_0 \chi|\le C$ for some $C>0$, the first term on the r.h.s. of \eqref{eq: estimate Vij psin-phi0} is arbitrarily small if $\varepsilon>0$ is small enough. The second term tends to zero as $n\rightarrow\infty$ because $\psi_n\rightarrow \varphi_0$ in $L^2_{\mathrm{loc}}(\mathbb{R}^d)$. Similarly, we can show that $\langle |V| |\psi_n|\chi,|\psi_n|\chi\rangle$ is bounded and therefore $\langle |V| |\psi_n-\varphi_0| ,|\psi_n|\rangle_{B(R)}$ tends to zero as $n\rightarrow \infty$. Analogously we get $\langle V (\psi_n-\varphi_0) ,\varphi_0\rangle_{B(R)}\rightarrow 0$ as $n\rightarrow \infty$. Hence, we get \eqref{eq: local convergence V psin to Vphi0}.\\
By taking $R>A$, condition \eqref{cond: decay at infinity} together with inequality \eqref{eq: weidl inequality in 1} for $d=1$ and \eqref{eq: weidl inequality in 2} for $d=2$, respectively, implies
\begin{align}\label{eq: one dimensional Hardy V relatively bounded}\begin{split}
\int_{\{|x|> R\}} |V(x)||\psi_n(x)|^2\, \mathrm{d}x &\leq C \int_{\{|x|> R\}} \frac{|\psi_n(x)|^2}{(1+|x|)^{2+\nu}}\, \mathrm{d}x 
\\&\leq \tilde C R^{-\frac{\nu}{2}}\Vert \psi_n\Vert_{\tilde{H^1}}^2 = \tilde CR^{-\frac{\nu}{2}}
\end{split}
\end{align}
for some constants $C,\tilde{C}>0$.
Since $\Vert \varphi_0 \Vert_{\tilde{H}^1} \leq 1$, by the same arguments we get 
\begin{equation}
\int_{\{|x|>R\}} |V(x)||\varphi_0(x)|^2\, \mathrm{d}x \leq \tilde CR^{-\frac{\nu}{2}},
\end{equation}
which together with \eqref{eq: local convergence V psin to Vphi0} implies that $\langle V\varphi_0, \varphi_0 \rangle$ is well-defined and 
\begin{equation}\label{11: (.)}
\langle V\psi_n,\psi_n\rangle \rightarrow \langle V\varphi_0,\varphi_0 \rangle \qquad \text{as}\ n \rightarrow \infty.
\end{equation}
Recall that
\begin{align}
\begin{split}
\langle V\psi_n,\psi_n \rangle &\leq -\left(1-n^{-1}\right) \Vert  \nabla \psi_n\Vert^2
\\ &=\left(1-n^{-1}\right)\left(-1+\int_{\{|x|\le 1\}}|\psi_n|^2\, \mathrm{d}x\right) \longrightarrow -1 + \int_{\{|x|\le 1\}} |\varphi_0|^2\, \mathrm{d}x.
\end{split}
\end{align}
Sending $n$ to infinity and using \eqref{11: (.)} yields
\begin{equation}
\langle V\varphi_0,\varphi_0\rangle \le -1 +\int_{\{|x|\le 1\}} |\varphi_0|^2\, \mathrm{d}x = -1 -\Vert \nabla \varphi_0\Vert^2+\Vert \varphi_0 \Vert_{\tilde{H}^1}^2.
\end{equation}
Since $\Vert \varphi_0 \Vert_{\tilde{H}^1}\le 1$ and the operator $h$ is non-negative, we get $\Vert \varphi_0 \Vert_{\tilde{H}^1}=1$ and
\begin{equation}
\Vert\nabla \varphi_0 \Vert^2+\langle V\varphi_0,\varphi_0 \rangle = 0.
\end{equation}
Standard arguments show that $\varphi_0$ satisfies \eqref{1: stand.} for any $\psi\in \tilde{H}^1(\mathbb{R}^d)$.
\end{proof}
Now we turn to the proof of statements \textbf{(ii)} and \textbf{(iii)} of Theorem \ref{thm: abstraktes Theorem}, i.e. the estimate of the weighted $L^2(\mathbb{R}^d)$ norm of $\varphi_0$. At first, we prove a weighted $L^2(\mathbb{R}^d)$-estimate for the functions $\psi_n$.
\begin{lem}\label{lem: estimate phi 0 for N=1 n=1}$ $ \
\hspace{-0.5cm} Assume that $h$ has a virtual level at zero and that $V$ satisfies \eqref{cond: relatively form bounded} and \eqref{cond: decay at infinity}. Let $(\psi_n)_{n\in \mathbb{N}}$ be a sequence of eigenfunctions corresponding to negative eigenvalues $E_n<0$ of the operator $h_{n^{-1}}$, normalized as $\Vert \psi_n\Vert_{\tilde{H}^1}=1$. Then the following assertions hold:
\begin{enumerate} 
\item[\textbf{(i)}] If $d=1$, then for any $ \alpha_0 <\frac{1}{2}$ there exists a $C>0$, such that for all $n\in \mathbb{N}$ we have
\begin{equation}\label{eq: estimate psin d=1 N=1}
\Vert \nabla \left(|x|^{\alpha_0}\psi_n\right)\Vert \le C \quad \text{and} \quad  \Vert (1+|x|)^{\alpha_0-1}\psi_n\Vert \le C.
\end{equation}
\item[\textbf{(ii)}] If $d=2$, then for any $\alpha_0<\frac{1}{2}$ there exists a $C>0$, such that for all $n\in \mathbb{N}$ we have
\begin{equation}
	\Vert \nabla\left(|\ln(|x|)|^{\alpha_0}\psi_n\right) \Vert\le C	\quad \text{and}\quad\Vert (1+|x|)^{-1}(1+|\ln(|x|)|)^{\alpha_0-1}\psi_n \Vert\le C.
\end{equation}
\end{enumerate}
\end{lem}
\begin{proof}[Proof of Lemma \ref{lem: estimate phi 0 for N=1 n=1}]
The proof is a modification of the proof of Lemma 2.4 in \cite{second}. At first, we prove the Lemma for the case $d=1$. For any $\varepsilon>0$ and $R>0$ we define the function
\begin{equation}\label{1: function agmon}
G_\varepsilon(x)=\frac{|x|^{\alpha_0}}{1+\varepsilon|x|^{\alpha_0}} \chi_R(x),
\end{equation}
where $\chi_R$ is a $C^1$-cutoff function with
\begin{equation}\label{11chi}
\chi_R(x)=\begin{cases}
0, & |x|\leq R,
\\ 1, & |x|\geq 2R.
\end{cases}
\end{equation}
We multiply the eigenvalue equation 
\begin{equation}\label{1: multiply and int by parts}
-(1-n^{-1})\Delta\psi_n+V\psi_n=E_n\psi_n
\end{equation}
by $G^2_\varepsilon\overline{\psi_n}$ and integrate by parts to obtain \begin{equation}\label{eq: eigenvalue equation G psi_n}
\left( 1-n^{-1}\right) \langle \nabla \psi_n, \nabla\left(G_\varepsilon^2\psi_n\right) \rangle+ \langle V\psi_n , G_\varepsilon^2\psi_n \rangle= E_n\Vert G_\varepsilon \psi_n \Vert^2<0.
\end{equation}
Since 
\begin{equation}
\Re \langle V\psi_n,G_\varepsilon^2\psi_n \rangle = \langle V\psi_n,G_\varepsilon^2\psi_n \rangle \qquad \text{and} \qquad \Re E_n \Vert G_\varepsilon \psi_n \Vert^2 = E_n\Vert G_\varepsilon \psi_n \Vert^2,
\end{equation}
we have
\begin{equation}
\Re \langle \nabla \psi_n,\nabla \left( G_\varepsilon^2 \psi_n \right)\rangle = \langle \nabla \psi_n,\nabla \left( G_\varepsilon^2 \psi_n \right)\rangle.
\end{equation}
Note that 
\begin{align}\label{eq: agmon Re}
\begin{split}
\Re \langle \nabla \psi_n , \nabla (G_\varepsilon^2 \psi_n) \rangle &= \Re \langle \nabla \psi_n, G_\varepsilon \psi_n \nabla G_\varepsilon \rangle  + \Re   \langle (\nabla \psi_n) G_\varepsilon , \nabla (G_\varepsilon \psi_n)\rangle
\\ &= \Re \langle \nabla(\psi_n G_\varepsilon), \psi_n \nabla G_\varepsilon \rangle - \Re \langle \psi_n\nabla G_\varepsilon, \psi_n \nabla G_\varepsilon \rangle  
\\ & \ \ \ +\Re \langle \nabla(\psi_n G_\varepsilon), \nabla (\psi_n G_\varepsilon) \rangle - \Re \langle \psi_n \nabla G_\varepsilon, \nabla (\psi_n G_\varepsilon) \rangle 
\\ &=\Re \langle \nabla(\psi_n G_\varepsilon), \nabla (\psi_n G_\varepsilon) \rangle - \Re \langle \psi_n \nabla G_\varepsilon, \psi_n \nabla G_\varepsilon \rangle.
\end{split}
\end{align}
Therefore, we obtain
\begin{equation}
\langle \nabla \psi_n, \nabla (G_\varepsilon^2 \psi_n) \rangle= \Vert \nabla(\psi_n G_\varepsilon)\Vert^2 - \Vert \psi_n \nabla G_\varepsilon\Vert^2.
\end{equation}
This together with \eqref{eq: eigenvalue equation G psi_n} yields
\begin{equation}\label{eq: agmon with psin}
\left(1-n^{-1} \right)\left(\Vert \nabla(\psi_n G_\varepsilon) \Vert^2 - \int|\psi_n|^2 |\nabla G_\varepsilon|^2\, \mathrm{d}x\right)+\int V|\psi_n G_\varepsilon|^2\, \mathrm{d}x<0.
\end{equation}
Now we estimate the function $|\nabla G_\varepsilon|$. For $|x|>2R$ we have
\begin{equation}\label{eq: estimate nabla G large arguments}
|\nabla G_\varepsilon| = \frac{\alpha_0 |x|^{\alpha_0-1}}{(1+\varepsilon|x|^{\alpha_0})^2}\leq \alpha_0|x|^{-1}|G_\varepsilon|.
\end{equation} 
For $|x| \in[R,2R]$ the function $|\nabla G_\varepsilon|$ is uniformly bounded in $\varepsilon$, which implies
\begin{equation}\label{eq: estimate nabla G}
\int_{\{R\leq |x|\leq 2R\}} |\nabla G_\varepsilon|^2|\psi_n|^2\, \mathrm{d}x \leq C_0\int_{\{R\leq |x|\leq 2R\}} |\psi_n|^2\, \mathrm{d}x ,
\end{equation}
for some $C_0>0$ which depends on $R$ only.
Now we use inequality \eqref{eq: weidl inequality in 1} to estimate the r.h.s. of \eqref{eq: estimate nabla G}. We get
\begin{align}
	\begin{split}
	\int_{\{R\leq |x|\leq 2R\}} |\psi_n|^2\, \mathrm{d}x \le (1+4R^2)\int_{\{R\leq |x|\leq 2R\}} \frac{|\psi_n|^2}{1+x^2}\, \mathrm{d}x \le C_H(1+4R^2) \Vert \psi_n\Vert^2_{\tilde{H}^1},
	\end{split}
\end{align}
where $C_H$ is a Hardy-type constant in \eqref{eq: weidl inequality in 1}.
This, together with \eqref{eq: estimate nabla G} and $\Vert \psi_n\Vert_{\tilde{H}^1}=1$ implies
\begin{equation}\label{eq: estimate nabla G with C(R)}
	\int_{\{R\leq |x|\leq 2R\}} |\nabla G_\varepsilon|^2|\psi_n|^2 \le C_1
\end{equation}
for some $C_1>0$ which is independent of $n\in \mathbb{N}$ and $\varepsilon>0$.  
Substituting \eqref{eq: estimate nabla G large arguments} and \eqref{eq: estimate nabla G with C(R)} into \eqref{eq: agmon with psin} we obtain
\begin{equation}\label{1: 39}
\left(1-n^{-1}\right)\Vert \nabla(\psi_n G_\varepsilon) \Vert^2 +\langle V G_\varepsilon \psi_n,G_\varepsilon \psi_n \rangle - \alpha_0^2 \int_{\{|x|>2R\}}\frac{|G_\varepsilon \psi_n|^2}{|x|^2}\, \mathrm{d}x\leq C_2,
\end{equation}
where $C_2>0$ does not depend on $n\in \mathbb{N}$ or $\varepsilon>0$. The function $G_\varepsilon \psi_n$ is supported outside the ball with radius $R$. Therefore, choosing $R>A$ we can use \eqref{cond: decay at infinity} and apply Hardy's inequality for the half-line, which yields
\begin{equation}\label{1: gamma}
(1-\gamma_0)\Vert \nabla (G_\varepsilon \psi_n)\Vert^2+\langle VG_\varepsilon \psi_n,G_\varepsilon \psi_n \rangle-\alpha_0^2 \langle |x|^{-2}G_\varepsilon \psi_n,G_\varepsilon \psi_n \rangle \geq 0
\end{equation}
for all $\alpha_0^2<\frac{1}{4}$ and $\gamma_0<(1-4\alpha_0^2)$. 
For $n>2\gamma_0^{-1} $ estimates \eqref{1: 39} and \eqref{1: gamma} imply
\begin{equation}\label{241}
\frac{\gamma_0}{2} \Vert \nabla(G_\varepsilon \psi_n) \Vert^2 \leq C_2.
\end{equation}
Taking the limit $\varepsilon\rightarrow 0$ yields $\Vert \nabla \left( |x|^{\alpha_0}\psi_n\right) \Vert \leq C$ for some $C>0$. \\
Applying Hardy's inequality for the half-line to the function $G_\varepsilon \psi_n$ and
taking the limit $\varepsilon\rightarrow 0$ implies 
\begin{equation}
 \Vert (1+|x|)^{\alpha_0-1}\psi_n\Vert \le C.
\end{equation}
 This completes the proof of Lemma \ref{lem: estimate phi 0 for N=1 n=1} for $d=1$. 
Now we assume $d=2$. For $\varepsilon>0$ and $0<\alpha_0<\frac{1}{2}$ let 
\begin{equation}
G_\varepsilon(x)=\frac{|\ln(|x|)|^{\alpha_0}}{1+\varepsilon|\ln(|x|)|^{\alpha_0}}\chi_R(x),
\end{equation}
where $\chi_R$ is a $C^1$-cutoff function with
\begin{equation}
\chi_R(x)=\begin{cases}
0, & |x|\leq R,
\\ 1, & |x|\geq 2R.
\end{cases}
\end{equation}   Due to \eqref{cond: decay at infinity} and Hardy's inequality in dimension two we get similarly to \eqref{1: gamma} that
\begin{equation}
(1-\gamma_0)\Vert \nabla (G_\varepsilon \psi_n)\Vert^2+\langle VG_\varepsilon \psi_n,G_\varepsilon \psi_n \rangle-\alpha_0^2 \langle |x|^{-2}\left(\ln |x|\right)^{-2} G_\varepsilon \psi_n,G_\varepsilon \psi_n \rangle \geq 0
\end{equation}
for all $\alpha_0^2<\frac{1}{4}$ and $\gamma_0 <( 1-4\alpha_0^2)$. 
Now the proof is a straightforward modification of the one-dimensional case.
\end{proof}
Statements \textbf{(ii)} and \textbf{(iii)} of Theorem \ref{thm: abstraktes Theorem} follow from the following
\begin{cor}\label{cor: estimate phi0} The weak limit $\varphi_0$ of the sequence $(\psi_n)_{n\in\mathbb{N}}$ has the following properties. 
\begin{enumerate}
\item[\textbf{(i)}]  If $d=1$, then
\begin{equation}\label{eq: estimate phi0 n=1 N=1}
	 (1+|x|)^{\alpha_0-1}\varphi_0 \in L^2 (\mathbb{R}) \qquad \text{for any} \ \alpha_0<\frac{1}{2}.
\end{equation}
\item[\textbf{(ii)}] If $d=2$, then 
\begin{equation}\label{eq: estimate phi0 n=2 N=1}
	 (1+|x|)^{-1}\left(1+\ln(|x|)\right)^{\alpha_0-1}\varphi_0 \in L^2 (\mathbb{R}^2) \qquad \text{for any}\ \alpha_0<\frac{1}{2}.
\end{equation}
\end{enumerate}
\end{cor}
\begin{proof}[Proof of Corollary \ref{cor: estimate phi0}]
Let $d=1$. Since $(\psi_n)_{n\in \mathbb{N}}$ converges to $\varphi_0$ in $L^2_{\mathrm{loc}}(\mathbb{R})$ and for any $\alpha_0<\frac{1}{2}$ we have the estimate $\Vert (1+|x|)^{\alpha_0-1}\psi_n\Vert \le C$ uniformly in $n\in \mathbb{N}$, for every $\alpha_0<\frac{1}{2}$ we get
\begin{equation}
	(1+|x|)^{\alpha_0-1}\psi_n\rightarrow (1+|x|)^{\alpha_0-1}\varphi_0 \quad \text{in } L^2(\mathbb{R}) \quad \text{as} \ n\rightarrow\infty.
\end{equation} 
The case $d=2$ follows analogously.
\end{proof}
To complete the proof of Theorem \ref{thm: abstraktes Theorem} it remains to prove statement \textbf{(iv)}. This is a straightforward generalization of Lemma 2.10 in \cite{second}, which is based on the the unique continuation theorem \cite[Theorem 2.1]{Schechter}. 
\end{proof}
\section{Virtual levels of systems of $N$ one- or two-dimensional particles}\label{section: virtual levels multi-particle}
In this section we introduce virtual levels of Schr\"odinger operators corresponding to systems consisting of $N$ one- or two-dimensional particles. We prove several results on the decay rate of solutions of the Schr\"odinger equation corresponding to virtual levels of multi-particle systems. The main result of this section is Theorem \ref{thm: Virtual level and Hardy constant}, where we give sufficient conditions in terms of a Hardy-type constant, such that virtual levels of multi-particle Schr\"odinger operators correspond to eigenvalues and prove an estimate for the decay rate of the corresponding eigenfunctions.
In Corollaries \ref{cor: d=2, N>3} and \ref{cor: d=1, N>3} and Theorem \ref{thm d=1 N=3} we discuss applications of Theorem \ref{thm: Virtual level and Hardy constant} to multi-particle systems.
\subsection{Notation and definitions for multi-particle systems}
We consider a system of $N \geq 3$ quantum particles in dimension $d=1$ or $d=2$ with masses $m_i>0$ and position vectors $x_i \in \R^d, \ i=1,\dots, N$. Such a system is described by the Hamiltonian
\begin{equation}\label{Definition Hamiltonian}
	H_N = - \sum\limits_{i=1}^N \frac{1}{m_i} \Delta_{x_i} + \sum\limits_{1\le i<j\le N} V_{ij}(x_{ij}),\quad x_{ij}=x_i-x_j
\end{equation}
acting on $L^2(\mathbb{R}^{dN})$. The potentials $V_{ij}$ describe the particle pair interactions and in the following we assume that they satisfy $V_{ij}\neq 0$ and the conditions \eqref{cond: relatively form bounded} and \eqref{cond: decay at infinity}.  
\subsubsection*{Separation of the center of mass of the system}
We will consider the operator $H_N$ in the center-of-mass frame. Following \cite{Sigalov}, we denote by $\langle \cdot, \cdot\rangle_m$ the scalar product on $\mathbb{R}^{dN}$ which is given by
\begin{equation}\label{eq: mass metric}
	\langle x,y\rangle_m = \sum\limits_{i=1}^N m_i \langle x_i ,y_i \rangle, \qquad \vert x\vert^2_m = \langle x,x\rangle_m, \qquad x, y\in \R^{dN}.
\end{equation}
Here, $\langle \cdot,\cdot\rangle$ is the standard scalar product on $\mathbb{R}^{d}$. Let $X$ be the space $\mathbb{R}^{dN}$ equipped with the scalar product $\langle\cdot,\cdot\rangle_m$ and let
\begin{equation}\label{X[C]}
	X_0 = \left\{x=(x_1,\ldots,x_N) \in X \; : \; \sum\limits_{i=1}^N m_i x_i =0 \right\}
\end{equation}
be the space of positions of the particles in the center of mass frame and $X_c =X\ominus X_0$ be the space of the center of mass position of the system.  We denote by $P_0$ and $P_c$ the orthogonal projections from $X$ on $X_0$ and $X_c$, respectively. \\
Furthermore, we introduce $-\Delta,\ -\Delta_0$ and $-\Delta_c$ as the Laplace-Beltrami operators on $L^2(X)$, $L^2(X_0)$ and $L^2(X_c)$, respectively. Then, corresponding to the decomposition $L^2(X)=L^2(X_0)\otimes L^2(X_c)$ we find
\begin{equation}
-\Delta =-\Delta_0\otimes \mathrm{Id} +\mathrm{Id}\otimes(-\Delta_c).
\end{equation}
Since for every $x\in X$ we have
\begin{equation}
(P_0 x)_i-(P_0x)_j=x_i-x_j,
\end{equation}
the potential $V(x)= \sum_{1\le i< j\le N} V_{ij}(x_{ij})$ satisfies
\begin{equation}
	V(x)=V(P_0x).
\end{equation}
Therefore, $H_N$ is unitarily equivalent to the operator
\begin{equation}\label{eq: decomposition operator H_N}
H\otimes \mathrm{Id}+\mathrm{Id}\otimes (-\Delta_c),
\end{equation}
where
\begin{equation}\label{H[C]}
H=-\Delta_0+V.
\end{equation}
In view of \eqref{eq: decomposition operator H_N} the center of mass of the system moves like a free particle and the operator $H$ corresponds to the relative motion of the system.
\subsubsection*{Clusters and Cluster Hamiltonians}
A cluster $C$ of the system is defined as a non-empty subset of $\{1,\dots, N\}$ and we denote by $|C|$ the number of particles contained in $C$. 
For $1<|C|<N$ we  define the space of the relative positions of the particles in the cluster $C$ by
\begin{equation}
X_0[C]=\{x\in X_0: x_i = 0 \ \text{if}\ i\not \in C\}.
\end{equation}
Let $-\Delta_0[C]$ be the the Laplace-Beltrami operator on $L^2(X_0[C])$ and
\begin{equation}
V[C]=\sum\limits_{i,j\in C,\ i < j}V_{ij}
\end{equation} 
the potential of the interactions between the particles in the cluster $C$. Then for $1<|C|<N$ the cluster Hamiltonian with reduced center of mass, acting on $L^2(X_0[C])$, is given by 
\begin{equation}
H[C]=-\Delta_0[C]+V[C]
\end{equation}
and describes the internal dynamics of the cluster $C$. For $C=\{1,\dots, N\}$ we have $X_0[C]=X_0$, so we set $H[C]=H$. For $|C|=1$ we have $X[C]=\{0\}$ and we set $H[C]=0$.\\
Let $P_0[C]$ be the orthogonal projection from $X_0$ to $X_0[C]$ and for $x\in X_0$ let
\begin{equation}
	q[C]=P_0[C]x. 
\end{equation}
\subsubsection*{Partitions of the system}
We say that $Z=(C_1,\dots,C_p)$ is a partition or cluster decomposition of the system of order $\vert Z\vert =p$ if and only if
 \begin{equation}
C_i\neq \emptyset,  \qquad C_i\cap C_j= \emptyset, \qquad \bigcup\limits_{j=1}^p C_j = \{1,\dots, N\}
\end{equation}
holds for all $i,j=1,\dots, p$ with $i\neq j$. We refer to $C\subset Z$ as a cluster of the partition $Z=(C_1,\dots,C_p)$ if $C=C_i$ for some $i=1,\dots,p$.
Let
\begin{equation}
	X_0(Z)=\bigoplus \limits_{C_k \subset Z} X_0[C_k], \qquad X_c(Z)= X_0\ominus X_0(Z).
\end{equation}
This gives rise to the decomposition 
\begin{equation}
L^2(X_0(Z))=\bigotimes\limits_{C_k \subset Z} L^2(X_0[C_k]).
\end{equation}
By abuse of notation we denote the operator
\begin{equation}
	\mathrm{Id}\otimes \dots \otimes \mathrm{Id}\otimes (-\Delta_0[C_k])\otimes \mathrm{Id}\otimes \dots \otimes \mathrm{Id} \quad \text{and} \quad \mathrm{Id}\otimes \dots \otimes \mathrm{Id}\otimes H[C_k]\otimes \mathrm{Id}\otimes \dots \otimes \mathrm{Id},
\end{equation}
acting on $L^2(X_0(Z))$, by $-\Delta_0[C_k]$ and $H[C_k]$, respectively. The cluster decomposition Hamiltonian of the partition $Z$ is defined by
\begin{align}
\begin{split}
	H(Z)= \sum\limits_{C_k \subset Z} H[C_k]
	\end{split}
\end{align}
and describes the joint internal dynamics of the clusters in $Z$. Let $-\Delta_0(Z)$ be the Laplace-Beltrami operator on $L^2(X_0(Z))$. Then 
\begin{equation}
	-\Delta_0(Z)=\sum\limits_{C_k\subset Z} -\Delta_0[C_k].
\end{equation}
We denote the potential of the inter-cluster interaction by
\begin{equation}
	I(Z)=V-\sum\limits_{C_k \subset Z} V[C_k].
\end{equation}
Then the Hamiltonian of the whole system can be written as
\begin{equation}
H=H(Z)\otimes \mathrm{Id}+\mathrm{Id}\otimes (-\Delta_c(Z)) +I(Z),
\end{equation}
where $-\Delta_c(Z)$ is the Laplace-Beltrami operator on $L^2(X_c(Z))$.
We introduce the projections $P_0(Z)$ and $P_c(Z)$ from $X_0$ on $X_0(Z)$ and $X_c(Z)$, respectively. For $x\in X_0$ let
\begin{equation}\label{eq: q and xi}
	q(Z)= P_0(Z)x, \qquad \xi(Z)=P_c(Z) x.
\end{equation}
To emphasize the dependence on $q(Z)$ and $\xi(Z)$ we will write
\begin{equation}
	-\Delta_{q(Z)}=-\Delta_0(Z) \quad \text{and} \quad -\Delta_{\xi(Z)}=-\Delta_c(Z)
\end{equation}
and 
\begin{equation}
	H=-\Delta_{q(Z)}-\Delta_{\xi(Z)}+V\quad \text{or}\quad H=H(Z)-\Delta_{\xi(Z)}+I(Z).
\end{equation}
Note that the $i$-th coordinates of $q(Z)$ and $\xi(Z)$ are vectors $q_i$ and $\xi_i$ given by
\begin{equation}\label{eq: q(z) and xi(Z)}
	q_i=x_i-x_{C_l}, \qquad \xi_i=x_{C_l}
\end{equation}
where $C_l$ is the cluster which contains the particle $i$. Here, 
\begin{equation}
	x_{C_l}=\frac{1}{\sum_{j\in C_l}m_j}\sum_{j\in C_l} m_jx_j
\end{equation}
is the center of mass of the cluster $C_l$. \\
For $\kappa>\kappa'>0$, $ R >0$ and partitions $Z$ with $1<|Z|<N$ we define the regions
\begin{align}\label{eq: cones}
	\begin{split}
		B(R)&=\left\{x\in X_0 \ :\ \vert x\vert_m \le R\right\}, \\
		K(Z,\kappa) &=\left\{x\in X_0 \ :\ \vert q\left({Z}\right)\vert_m \le \kappa \vert \xi\left(Z\right)\vert_m\right\}	,	\\
		K_R(Z,\kappa) &=\left\{x\in X_0 \ :\ \vert q\left({Z}\right)\vert_m \le \kappa \vert \xi\left(Z\right)\vert_m, \ \vert x\vert_m \ge R\right\},\\
		K_R(Z,\kappa',\kappa) &=K_R(Z,\kappa) \setminus K_R(Z,\kappa'). \\
	\end{split}
\end{align}
For the entire system $Z=\{1,\dots,N\}$ we set 
\begin{equation}
	K(Z,\kappa)=\{x\in X_0:|x|_m\le \kappa\}.
\end{equation}
We will use the regions defined in \eqref{eq: cones} to make a partition of unity of $X_0$ corresponding to different cluster decompositions of the $N$-particle system. Now we extend Definition \ref{df: Virtual level one particle} of a virtual level to the case of multi-particle systems but first we give two remarks which justify our assumptions.
\begin{enumerate}
\item[\textbf{(i)}] For three-particle systems with the essential spectrum starting at zero, the existence of resonances in two-particle subsystems may lead, to the appearence of an infinite series of negative eigenvalues accumulating logarithmically at zero, the so-called Efimov effect. These eigenvalues and the corresponding eigenfunctions have many interesting properties. One of the goals of our work is to study whether similar effects may occur in systems of $N$ one- or two-dimensional particles. Due to this specific interest we will consider only the case when $\mathcal{S}_{\mathrm{ess}}(H)=[0,\infty)$. By the HVZ theorem this yields $H[C]\geq 0$ for any cluster $C$ with $|C|<N$.
\item[\textbf{(ii)}] The assumption $H[C]\geq 0$ is a strong restriction on the potentials $V_{ij}$. Since we consider one-or two-dimensional particles only, this in particular implies $\int V_{ij} \, \mathrm{d}x >0$.
\end{enumerate}
\begin{df}\label{def: Vitrtual levels multiparticle}
Assume that the potentials $V_{ij}$ satisfy \eqref{cond: relatively form bounded} and \eqref{cond: decay at infinity}. Let $C\subseteq \{1,\dots,N\}$ be a cluster. We say that $H[C]$ has a virtual level at zero if $H[C]\ge 0$ and 
\begin{enumerate}
\item[\textbf{(i)}] there exists a constant $\varepsilon_0 \in (0,1)$, such that
\begin{equation}\label{eq: essential spectrum stable}
\inf\mathcal{S}_{\mathrm{ess}} \left(H[C]+\varepsilon_0\Delta_0[C]\right)=0,
\end{equation} 
\item[\textbf{(ii)}] for any $\varepsilon\in (0,1)$ we have
\begin{equation}\label{eq_ def vl (ii)}
\inf\mathcal{S}\left(H[C]+\varepsilon\Delta_0[C]\right)<0.
\end{equation}
\end{enumerate}
\end{df}
\begin{rem}
\begin{enumerate}
\item[\textbf{(i)}] Note that if \eqref{eq: essential spectrum stable} is fulfilled for some $\varepsilon_0>0$, then it also holds for all $0<\tilde{\varepsilon}_0<\varepsilon_0$.
\item[\textbf{(ii)}] Let $H[C]\geq 0$. Then condition \eqref{eq: essential spectrum stable} can not be fulfilled if there exists a subcluster $\tilde{C} \subset C$ with $1<|\tilde{C}|< |C|$ such that $H[\tilde{C}]$ has a virtual level. Indeed, in this case we have $\inf \mathcal{S}\left(H[\tilde C]+\varepsilon\Delta_0[\tilde C]\right) <0$ for any $\varepsilon \in (0,1)$ and according to the HVZ theorem \eqref{eq: essential spectrum stable} does not hold. 
\\On the other hand, if \eqref{eq: essential spectrum stable} does not hold for a cluster $C$ and any $\varepsilon_0\in (0,1)$, then due to the HVZ theorem there exists at least one subcluster $\tilde C$ of the cluster $C$ with $1<|\tilde C|<|C|$, such that for any $\varepsilon\in (0,1)$ we have
\begin{equation}\label{eq: inf spectrum subset negative}
\inf\mathcal{S}\left(H[\tilde C]+\varepsilon \Delta_0[\tilde C]\right)<0.
\end{equation}
Among these subclusters we choose one with the smallest number of particles and denote it by $C_0$. If $C_0$ has only two particles, then $\inf \mathcal{S}_{\mathrm{ess}}(H[C_0]+\varepsilon_0\Delta_0[C_0])=0$ and according to the definition of a virtual level, the operator $H[C_0]$ has a virtual. Let $|C_0|\geq 3$. Then, since $C_0$ is the smallest cluster for which \eqref{eq: inf spectrum subset negative} holds for any $\varepsilon\in(0,1)$, for any subcluster $ C'\subsetneq C_0$ with $|C'|>1$ inequality \eqref{eq: inf spectrum subset negative} can not hold for all $\varepsilon\in (0,1)$, i.e. for some $\varepsilon_1\in (0,1)$ we have
\begin{equation}\label{eq: 331}
\inf\mathcal{S}\left(H[{C'}]+\varepsilon_1\Delta_0[{C'}]\right)=0.
\end{equation}
Obviously, this is also true for all $0<\tilde{\varepsilon}_1<\varepsilon_1$. Since $C_0$ has only a finite number of subclusters, we can choose $\varepsilon_1>0$ in \eqref{eq: 331}, such that this inequality holds for all subclusters of $C_0$. Applying the HVZ theorem yields
\begin{equation}
	\inf \mathcal{S}_{\mathrm{ess}} \left(H[C_0]+\varepsilon_0\Delta_0[C_0]\right)=0
\end{equation}
for any $\varepsilon_0\in (0,\varepsilon_1)$. At the same time, $\inf\mathcal{S} \left(H[C_0]+\varepsilon\Delta_0[C_0]\right)<0$ for any $\varepsilon\in(0,1)$. Hence, $H[C_0]$ has a virtual level at zero.
\item[\textbf{(iii)}] Similarly to the case of one-particle Schr\"odinger operators we can give necessary and sufficient conditions for the operator $H$ to have a virtual level at zero in terms of perturbations of the operator with additional potentials. This result can be found in Appendix \ref{Appendix Virtual Levels}, Theorem \ref{thm: equivalent definition virtual level multiparticle}.
\end{enumerate}
\end{rem}

\subsection{Statements of our results on the decay rates of solutions corresponding to virtual levels}
Now we give our main results of this section, namely the existence of solutions of the Schr\"odinger equation in the presence of a virtual level and estimates of the decay rate of these solutions. For these estimates certain Hardy-type constants play an important role. Let
\begin{equation}
	\mathcal{M}=\left\{\psi\in C_0^1(X_0\setminus B(1))\,:\, \psi(x)=0 \text{ for } x_i=x_j \, , 1\le i,j\le N,\ i\neq j\right\}
\end{equation}
and let
\begin{equation}\label{eq: definition tilde C}
 \tilde{C}_H(X_0)= \inf\limits_{0\neq \psi \in \mathcal{M}} \frac{\Vert \nabla_0 \psi\Vert}{\Vert |x|_m^{-1}\psi\Vert}.
\end{equation}
\begin{rem}
If the particles are two-dimensional, the sets $\{x_i=x_j\}$ have co-dimension two and the set $\mathcal{M}$ is dense in $H^1\left(X_0\setminus B(1)\right)$. In this case the constant $\tilde{C}_H(X_0)$ coincides with the Hardy constant $C_H(X_0)=\frac{d(N-1)-2}{2}=N-2$ of the $d(N-1)$-dimensional space $X_0$, see for example inequality $(2.18)$ in \cite{birman}. However, for one-dimensional particles the sets $\{x_i=x_j\}$ are hyperplanes and the closure of $\mathcal{M}$ with respect to the $H^1(X_0)$ norm includes only functions with trace zero on $\{x_i=x_j\}$. Below we will see that in this case we have $\tilde{C}_H(X_0)\ge\frac{N-1}{2}$. 
\end{rem}
The main result of this section is the following
\begin{thm}\label{thm: Virtual level and Hardy constant}
Let $H$ be the Hamiltonian of a system of $N\ge 3$ $d$-dimensional particles with $d\in \{1,2\}$, where the potentials $V_{ij}\neq 0$ satisfy \eqref{cond: relatively form bounded} and \eqref{cond: decay at infinity}. Assume that $H$ has a virtual level at zero and for the constant $\tilde{C}_H(X_0)$ defined in \eqref{eq: definition tilde C} we have $\tilde{C}_H(X_0)>1$. Then 
\begin{enumerate}
\item[\textbf{(i)}] 
 zero is a simple eigenvalue of $H$ and for the corresponding eigenfunction $\varphi_0$ we have
\begin{equation}\label{eq: decay rate n=1}
\nabla_0\left(|x|_m^{\alpha}\varphi_0\right)\in L^2(X_0)\quad \text{and} \quad  (1+|x|_m)^{\alpha-1}\varphi_0 \in L^2(X_0)
\end{equation}
for any $0\le\alpha <\tilde{C}_H(X_0)$.
\item[\textbf{(ii)}] There exists a constant $\delta_0>0$, such that for any function $\psi \in H^1(X_0)$ satisfying $\langle \nabla_0 \varphi_0,\nabla_0 \psi\rangle =0$ 
\begin{equation}\label{111: distance}
(1-\delta_0)\Vert \nabla_0 \psi \Vert^2 + \langle V\psi,\psi \rangle \geq 0.
\end{equation}
\end{enumerate}
\end{thm}

\begin{cor}\label{cor: d=2, N>3}
 If $d=2$ and $N\ge 4$, then we have $\tilde{C}_H(X_0)= C_H(X_0)=N-2>1$. Therefore, Theorem \ref{thm: Virtual level and Hardy constant} can be applied. In particular, it shows that in this case the solution $\varphi_0$ of the Schr\"odinger equation corresponding to the virtual level is a non-degenerate eigenfunction satisfying
\begin{equation}
	(1+|x|_m)^{\alpha-1}\varphi_0 \in L^2(X_0)\qquad \text{for any } \alpha<N-2.
\end{equation}
\end{cor}

\begin{cor}\label{cor: d=1, N>3}
If $d=1$ and $N\ge 4$, each of the hyperplanes $\{x_i=x_j\}$ divides the space $X_0$ into two half-spaces. Taking one of these hyperplanes and using that the Hardy constant for the half-space is given by $\frac{N-1}{2}$ \cite[Proposition 4.1]{nazarov} we get $\tilde{C}_H(X_0)\ge \frac{N-1}{2}>1$. Hence, Theorem \ref{thm: Virtual level and Hardy constant} can be applied. This implies that zero is a simple eigenvalue of $H$ and the corresponding eigenfunction $\varphi_0$ satisfies
\begin{equation}
	(1+|x|_m)^{\alpha-1}\varphi_0 \in L^2(X_0) \qquad \text{for any } \alpha<\frac{N-1}{2}.
\end{equation}
\end{cor}

\begin{rem}
We can significantly improve the estimate from below for the constant $\tilde{C}_H(X_0)$ given in Corollary \ref{cor: d=1, N>3} by taking into account that the traces of functions in $\mathcal{M}$ are zero not only on one of the hyperplanes $\{x_i=x_j\}$, but on all of them. For example, if we have a system of $N=4$ identical particles, then there are six hyperplanes $\{x_i=x_j\}$ which cut the space $X_0$ into congruent sectors $S_i$. One can show that the hyperplanes are the nodal set of a homogeneous harmonic polynomial of degree six. Its restriction to the unit sphere is a spherical harmonic of degree six and an eigenfunction corresponding to the first eigenvalue of the Dirichlet-Laplacian on $S_i\cap \mathbb{S}^2$. This, together with \cite[Proposition 4.1]{nazarov} implies that in this case the constant $\tilde{C}_H(X_0)$ is given by $\tilde{C}_H(X_0)=\left(\frac{1}{4}+42\right)^\frac{1}{2}=\frac{13}{2}.$
\end{rem}
Note that the constant $\tilde{C}_H(X_0)$, which gives a lower bound on the decay rate of the eigenfunction $\varphi_0$, does not depend on the potentials. However, for one-dimensional particles it does depend on the ratios of the masses of the particles. In particular if $d=1, \, N=3$ we get the following 
\begin{thm}\label{thm d=1 N=3}
For a system of three one-dimensional particles with masses $m_1,\, m_2,\, m_3>0$ let
\begin{equation}
	 \theta_i=\arccos\left(\frac{\sqrt{m_j m_k}}{\sqrt{m_i+m_j}\sqrt{m_i+m_k}}\right).
\end{equation}
Then we have
\begin{equation}
\tilde{C}_H(X_0)=\frac{\pi}{\theta_0},
 \quad \text{where} \quad \theta_0 = \max\{\theta_i, \ i=1,2,3\}.
\end{equation}
\end{thm}
\begin{rem}
\begin{enumerate}
\item[\textbf{(i)}] It is easy to see that for $d=1, \, N=3$ we have $\frac{\pi}{3}\le \theta_0\le\frac{\pi}{2}.$ The constant $\tilde{C}_H(X_0)$ takes its maximal value  $\tilde{C}_H(X_0)=3$ for $\theta_0=\frac{\pi}{3}$, which corresponds to the case $m_1=m_2=m_3$. On the other hand, if one of the masses $m_i$ tends to infinity, then $\theta_0\rightarrow \frac{\pi}{2}$ and therefore $\tilde{C}_H(X_0)\rightarrow 2$. 
\item[\textbf{(ii)}] Corollary \ref{cor: d=2, N>3}, Corollary \ref{cor: d=1, N>3} and Theorem \ref{thm d=1 N=3} show that for all multi-particle systems consisting of one- or two-dimensional particles, except for the case $d=2$ and $N=3$, virtual levels correspond to eigenvalues. This fact will be used in Section \ref{Section No Efimov multi} to prove the absence of the Efimov effect in multi-particle systems in dimension one and two.
\item[\textbf{(iii)}] Note that if the dimension of the particles is $d\ge 3$, the eigenfunction $\varphi_0$ corresponding to a virtual level decays with the same rate as the fundamental solution of the Laplace operator \cite{second}, \cite{mini}. Theorem \ref{thm: Virtual level and Hardy constant} shows that for one-dimensional particles the decay rate is higher. 
\end{enumerate}
\end{rem}
 For $d=2, \ N=3$ we do not expect that solutions of the Schr\"odinger equation corresponding to a virtual level are eigenfunctions. We will discuss this case in Section \ref{section: d=2 N=3}.

\subsection{Proof of Theorem \ref{thm: Virtual level and Hardy constant} --  Auxiliary results}
To prove Theorem \ref{thm: Virtual level and Hardy constant} we need several auxiliary results.  The first one is a generalization of Theorem \ref{thm: abstraktes Theorem} to potentials which do not necessarily decay at infinity.
\begin{thm}\label{thm: abstraktes thm paper2}
Let $h=-\Delta+V$ acting on $L^2(\mathbb{R}^k), \ k\in \mathbb{N},$ where the potential $V$ satisfies \eqref{cond: relatively form bounded}. Suppose that $h$ has a virtual level at zero and that there exist constants $\alpha_0>1$, $b>0$ and $\gamma_0\in(0,1)$, such that for any function $\psi\in H^1(\mathbb{R}^k)$ with  $\supp (\psi) \subset \{x\in \mathbb{R}^k:\ |x|\geq b\}$ we have
\begin{equation}\label{eq: subtract alpha0}
\langle h\psi,\psi \rangle -\gamma_0 \Vert \nabla \psi\Vert^2-\alpha_0^2 \langle |x|^{-2}\psi,\psi \rangle \geq 0.
\end{equation}
Then zero is a simple eigenvalue of $h$ and the corresponding eigenfunction $\varphi_0$ satisfies
\begin{align}\label{eq: estimate weigthed norm abziehtheorem}
	\begin{split}
	& (1+|x|)^{\alpha-1}\varphi_0 \in L^2(\mathbb{R}^k) \qquad\qquad \qquad \qquad\ \ \,  \text{if } k\neq 2,\\	
	 \text{and} \quad &(1+|x|)^{\alpha-1}(1+|\ln(|x|)|)^{-1}\varphi_0 \in L^2(\mathbb{R}^k) \qquad \text{if } k=2
	\end{split}
\end{align}
for any $\alpha<\alpha_0$. Moreover, there exists a constant $\delta_0>0$, such that for any function $\psi \in H^1(\mathbb{R}^k)$ with $\langle \nabla\psi,\nabla \varphi_0\rangle =0$ 
\begin{equation}\label{eq: orthogonality subtract abziehtheorem}
\langle h \psi,\psi \rangle \geq \delta_0 \Vert \nabla \psi \Vert^2.
\end{equation}
\end{thm}
\begin{rem}
\begin{enumerate}
\item[\textbf{(i)}] By Lemma \ref{lem: Abstraktes Lemma endlich viele Eigenwerte} condition \eqref{eq: subtract alpha0} implies $\inf\mathcal{S}_{\mathrm{ess}}( h_\varepsilon)=0$ for sufficiently small $\varepsilon>0$. 
\item[\textbf{(ii)}]
Theorem \ref{thm: abstraktes thm paper2} is a generalization of Theorem 2.1 in \cite{second} to dimensions $k=1$ and $k=2$. Therefore, we only need to prove the theorem for these dimensions.
\item[\textbf{(iii)}] Note that Theorem \ref{thm: abstraktes thm paper2} does not require that the potential $V$ decays in all directions, which is the case if we consider multi-particle systems where $V$ is the sum of the pair-potentials.
\item[\textbf{(iv)}] For dimension $k=1$ or $k=2$ Theorem \ref{thm: abstraktes thm paper2} considers the case which is in some sense complementary to the one studied in Theorem \ref{thm: abstraktes Theorem}. In Theorem \ref{thm: abstraktes Theorem} we assumed that the potential $V$ decays fast at infinity. In Theorem \ref{thm: abstraktes thm paper2} we do not require any decay of the potential. Instead of this we need inequality \eqref{eq: subtract alpha0} for functions $\psi$ which are supported far away from the origin. This condition can not be fulfilled for $k=1$ and $k=2$ if $V$ decays fast at infinity. Moreover, under the conditions of Theorem \ref{thm: abstraktes thm paper2} virtual levels correspond to eigenvalues of $h$. In contrast to that, under the conditions of Theorem \ref{thm: abstraktes Theorem} they correspond to resonances. 
\end{enumerate}
\end{rem}

\subsection*{Proof of Theorem \ref{thm: abstraktes thm paper2} for dimensions $\boldsymbol{k=1}$ and $\boldsymbol{k=2}$}
Since $\inf\mathcal{S}(h_\varepsilon)<0$ for every $\varepsilon>0$, we find a sequence of eigenfunctions $\psi_n \in H^1(\mathbb{R}^k)$, corresponding to eigenvalues $E_n<0$ of the operator $h_{n^{-1}}$, i.e.
\begin{equation}\label{1: eigenfunctions}
-\left( 1-n^{-1}\right)\Delta \psi_n+V\psi_n=E_n\psi_n.
\end{equation}
We normalize the functions $\psi_n$ by the condition $\Vert \psi_n\Vert_{\tilde{H}^1}=1$. In the first step we show a uniform bound for the $L^2(\mathbb{R}^k)$ norm of the functions $\psi_n$. 
\begin{lem}\label{lem: Lemma uniform bound of psi n}
Assume that $h$ has a virtual level at zero, where $V$ satisfies \eqref{cond: relatively form bounded} and suppose that \eqref{eq: subtract alpha0} holds for some $\alpha_0>1$. Then there exists a constant $C>0$, such that for any eigenfunction $\psi_n\in H^1(\mathbb{R}^k)$ corresponding to a negative eigenvalue of the operator $h_{n^{-1}}$, normalized by $\Vert \psi_n \Vert_{\tilde{H}^1}=1$, we have 
\begin{equation}
\Vert \nabla (|x|^{\alpha_0}\psi_n) \Vert \leq C 
\end{equation} 
and
\begin{align}\label{eq: Uniform estimate psin}
	\begin{split}
	& \Vert (1+|x|)^{\alpha_0-1}\psi_n \Vert \le C \qquad\qquad \qquad \qquad \, \ \ \text{if } k=1\\		
	 \text{and} \qquad &\Vert(1+|x|)^{\alpha_0-1}(1+|\ln(|x|)|)^{-1}\psi_n\Vert \le C \qquad \text{if } k=2.
	\end{split}
\end{align}
\end{lem}
The proof of Lemma \ref{lem: Lemma uniform bound of psi n} is a straightforward generalization of the proof of Lemma \ref{lem: estimate phi 0 for N=1 n=1} in Section \ref{Section Shortrange} and Lemma 2.4 in \cite{second} and we omit it here.
\\
Since we have normalized the sequence $\left(\psi_n\right)_n$ as $\Vert \psi_n\Vert_{\tilde{H}^1}=1$, there exists a subsequence, also denoted by $(\psi_n)_{n\in \mathbb{N}}$, which converges weakly in $\tilde{H}^1(\mathbb{R}^k)$ to a function $\varphi_0\in \tilde{H}^1(\mathbb{R}^k)$.\\
Now we show that $\varphi_0$ is an eigenfunction of $h$ corresponding to the eigenvalue zero. This is done in the following
\begin{lem}\label{lem: phi 0 eigenfunction}
Assume that $h$ has a virtual level at zero, where $V$ satisfies \eqref{cond: relatively form bounded} and suppose that \eqref{eq: subtract alpha0} holds for some $\alpha_0>1$. Then the function $\varphi_0$ given above is an eigenfunction of the operator $h$ corresponding to the eigenvalue zero, satisfying $\Vert \varphi_0\Vert_{\tilde{H}^1}=1$. 
\end{lem}
\begin{proof}[Proof of Lemma \ref{lem: phi 0 eigenfunction}]
By Lemma \ref{lem: Lemma uniform bound of psi n} estimate \eqref{eq: Uniform estimate psin} holds for some $\alpha_0>1$. Hence, we get convergence of the subsequence $(\psi_n)_{n\in \mathbb{N}}$ in $L^2(\mathbb{R}^k)$ and the limit satisfies $\left(1+|x|\right)^{\alpha}\varphi_0\in L^2(\mathbb{R}^k)$ for any $\alpha<\alpha_0-1$. In particular, $\varphi_0\in H^1(\mathbb{R}^k)$.
\\Due to the semi-continuity of the norm we have $\Vert \varphi_0\Vert_{\tilde{H}^1}\le 1$. Since $(\psi_n)_{n\in \mathbb{N}}$ converges to $\varphi_0$ in $L^2(\mathbb{R}^k)$ and $V$ satisfies \eqref{cond: relatively form bounded}, we get 
\begin{equation}\label{eq: V psin to Vphi0}
 \langle V\psi_n,\psi_n\rangle \rightarrow \langle V\varphi_0,\varphi_0\rangle, \qquad (n\rightarrow \infty).
\end{equation} 
Now it follows analogously to the proof of Lemma \ref{lem: estimate phi 0 for N=1 n=1} that 
$ \varphi_0$ satisfies
\begin{equation}
	\Vert \nabla \varphi_0\Vert^2+\langle V\varphi_0,\varphi_0\rangle =0
\end{equation}
and $\Vert \varphi_0\Vert_{\tilde{H}^1}=1$. Recall that $\varphi_0\in H^1(\mathbb{R}^k)$ and therefore is an eigenfunction of $h$ corresponding to the eigenvalue zero. This completes the proof of Lemma \ref{lem: phi 0 eigenfunction}.
\end{proof}
It remains to prove \eqref{eq: orthogonality subtract abziehtheorem} and the uniqueness of $\varphi_0$. This is a straightforward modification of the proof of (2.9) in \cite{second}. The only difference is that we normalize the sequence of eigenfunctions by $\Vert \psi_n\Vert_{\tilde{H}^1}=1$ instead of $\Vert \nabla\psi_n\Vert=1.$ 
\subsection*{A new estimate of the localization error}
To prove Theorem \ref{thm: Virtual level and Hardy constant} we will follow the strategy of the proof of Theorem 4.4 in \cite{second}, where so-called geometric methods were used to get lower bounds for the quadratic form of the multi-particle Schr\"odinger operator. These methods include, for example, a separation of regions $K(Z,\kappa)$ corresponding to different partitions $Z$. A crucial element of these methods is a partition of unity of the configuration space, which requires an appropriate estimate of the localization error. Note that in fact the localization error is responsible for the existence or non-existence of the Efimov effect.\\
If the dimension of the particles is  $d\ge 3$, one can use a localization error estimate given in \cite[Lemma 5.1]{Semjon2}. This estimate shows that when we separate a cone $K(Z,\kappa)$, then the localization error can be estimated as $\varepsilon |q(Z)|^{-2}$ with arbitrarily small $\varepsilon>0$. Using Hardy's inequality this term can be controlled by a small part of the kinetic energy. However, if the particles are one- or two-dimensional, the estimate given in \cite{Semjon2} cannot be used, because Hardy's inequality fails in dimension one and two. Therefore, we need a significant improvement of the estimate of the localization error. This is done in the following
\begin{thm}\label{lem: localization with log}
Given $\varepsilon>0$ and $\kappa>0$, for each partition $Z$ with $|Z|\ge 2$ one can find a constant $0<\kappa'<\kappa$ and piecewise continuously differentiable functions $u_{Z}, v_{Z}  :  X_0\rightarrow \mathbb{R}$, such that
\begin{equation}\label{Localization functions}
u_{Z}^2+v_{Z}^2 =1, \qquad u_{Z}(x)=\begin{cases}
    1, & x\in K\left(Z,\kappa'\right),\\
    0, & x\notin K\left(Z,\kappa\right),
\end{cases}
\end{equation}
and
\begin{equation}\label{Localization error I}
    \vert \nabla_0 u_{Z} \vert^2 + \vert \nabla_0 v_{Z} \vert^2 < \varepsilon \left[\vert v_{Z} \vert^2\vert x\vert_m^{-2}+\vert u_{Z} \vert^2\vert q\vert_m^{-2}\ln^{-2}\left(\vert q\vert_m \vert \xi\vert_m^{-1}\right)\right]
\end{equation}
for $x\in  K\left(Z,\kappa',\kappa\right)$. Here $q=q(Z)$ and $\xi=\xi(Z)$.
\end{thm}
To prove Theorem \ref{lem: localization with log} we will use an auxiliary result for scalar functions, namely the following
\begin{lem}\label{lem: localization scalar}
For any $\varepsilon>0$ and $0<\beta<1$ one can find a constant $0<\alpha<\beta^2$ and a non-increasing function $u\in H^1(\alpha,\beta)\cap C([\alpha,\beta])$, such that $ u(\alpha) =1, \ u(\beta)=0 $ and
\begin{equation}
(u'(t))^2\le \varepsilon t^{-2} \ln^{-2}(t), \qquad \alpha \le t\le \beta.
\end{equation}
\end{lem}

\begin{proof}[Proof of Lemma \ref{lem: localization scalar}]
   Let $\varepsilon>0$ and $\beta\in(0,1)$ be fixed. For any $0<\gamma <1$ and $\alpha\in(0,\beta^2)$ let $u:[\alpha,\beta] \rightarrow \mathbb{R}$ be given by
    \begin{equation}
        u(t):= \begin{cases}
        \left\vert\ln(\alpha \beta^{-1})\right\vert^{-\gamma}\left\vert \ln(t  \beta^{-1})\right\vert^\gamma & \text{ if } \alpha \le t\le \beta^2 ,       \\
        \left\vert\ln(\alpha \beta^{-1})\right\vert^{-\gamma}\left\vert \ln \beta\right\vert^{\gamma-1} \left\vert \ln(t  \beta^{-1})\right\vert & \text{ if } \beta^2\le t\le \beta.\\
        \end{cases}
    \end{equation}
    Obviously, $u \in C([\alpha,\beta])\cap H^1(\alpha,\beta)$ with $u(\beta)=0$ and $u(\alpha)=1$.\\
    At first, we prove the claimed estimate for $(u'(t))^2$ for $\alpha\le t\le \beta^2$ by choosing the constant $\gamma>0$ sufficiently small. For $\alpha < t< \beta^2$ we have
    \begin{equation}
        \left(u'(t)\right)^2= \gamma^2\left\vert\ln(\alpha \beta^{-1})\right\vert^{-2\gamma}\left\vert \ln(t  \beta^{-1})\right\vert^{2(\gamma-1)}t^{-2}.
    \end{equation}
    Note that $\alpha\beta^{-1}<1$ and $t\beta^{-1}<1$ for $\alpha \le t \le \beta^2$ and therefore $\left\vert\ln(\alpha \beta^{-1})\right\vert \ge \left\vert\ln(t  \beta^{-1})\right\vert$, which yields
    \begin{equation}
        \left(u'(t)\right)^2 \le \gamma^2 \left\vert\ln(t  \beta^{-1})\right\vert^{-2} t^{-2}, \qquad \alpha < t< \beta^2.
    \end{equation}
    Furthermore, for $t\le\beta^2$ we have $\left\vert\ln(t  \beta^{-1})\right\vert \ge \vert \ln \sqrt{t}\vert = \frac{1}{2}\left\vert \ln t\right\vert.$ This implies
    \begin{equation}
        \left(u'(t)\right)^2 \le 4 \gamma^2 \vert \ln t\vert^{-2}t^{-2}, \qquad \alpha< t< \beta^2.
    \end{equation}
    Choosing $0<\gamma < \frac{\sqrt{\varepsilon}}{2}$ we get
    \begin{equation}
        \left(u'(t)\right)^2 \le \varepsilon \vert \ln t\vert^{-2}t^{-2}, \qquad \alpha< t < \beta^2.
    \end{equation}
    Now we estimate $(u'(t))^2$ for $\beta^2< t < \beta$. In this case we have
    \begin{equation}
        \left(u'(t)\right)^2=\left\vert\ln(\alpha \beta^{-1})\right\vert^{-2\gamma}\left\vert \ln \beta\right\vert^{2(\gamma-1)}t^{-2}.
    \end{equation}
    Since $\beta<1$, we have $\vert \ln  \beta^2\vert \ge \vert \ln t\vert$ for $\beta^2\le t \le \beta$ and therefore
    \begin{align}
    \begin{split}
        \left(u'(t)\right)^2&\le\left\vert\ln(\alpha \beta^{-1})\right\vert^{-2\gamma}\left\vert \ln \beta\right\vert^{2(\gamma-1)}\vert \ln  \beta^2\vert^2  \vert \ln t\vert^{-2}t^{-2}\\
        &=4 \left\vert\ln(\alpha \beta^{-1})\right\vert^{-2\gamma}\left\vert \ln \beta\right\vert^{2\gamma}\vert \ln t\vert^{-2}t^{-2} \le \varepsilon\vert \ln t\vert^{-2}t^{-2}
        \end{split}
    \end{align}
    if $\alpha$ is chosen small enough. This completes the proof of Lemma \ref{lem: localization scalar}
\end{proof}
Now we turn to the 
\begin{proof}[Proof of Theorem \ref{lem: localization with log}]
Let $Z$ be a partition with $|Z|\ge 2$ and let $\varepsilon>0$ and $0<\kappa<1$ be fixed. We construct functions $u_{Z},\,v_{Z}$ which satisfy the conditions of Theorem \ref{lem: localization with log}.\\
 Let $v_1\in H^1(\mathbb{R}_+)$ be a non-decreasing function with $v_1(t)=1$ for $t\ge \kappa$ and $0\le v_1(t)<1 $ for $t<\kappa$, such that $v_1'(t)(1-v_1^2)^{-\frac{1}{2}} \rightarrow 0$ as $t \nearrow \kappa$. For $x\in X_0, \ x=q+\xi$, let
    \begin{equation}
        v_{Z}(x) = v_1\left(\frac{\vert q\vert_m}{\vert\xi\vert_m}\right), \qquad u_{Z}(x) = \sqrt{1-v_{Z}^2(x)}.
    \end{equation}
 Then for $x\in K(Z,\kappa)$ we have
\begin{align}\begin{split}\label{eq: gradient u + gradient v}
        \vert \nabla_0 u_{Z} \vert^2 + \vert \nabla_0 v_{Z} \vert^2 &= \vert \nabla_0 v_{Z} \vert^2 \left(1-v^2_{Z}\right)^{-1}\\
        &=(v_1'(t))^2\left(1-v_1^2(t)\right)^{-1}\left(1+\vert q\vert^2_m\vert \xi\vert_m^{-2}\right)\vert \xi\vert_m^{-2},
        \end{split}
\end{align}
where $ t= \vert q\vert_m\vert \xi\vert_m^{-1}.$ For $x\in K(Z,\kappa)$ we have $\vert \xi \vert_m^{-2} \le \left(1+\kappa^2\right) \vert x\vert_m^{-2}$. This implies
\begin{equation}
\vert \nabla_0 u_{Z} \vert^2 + \vert \nabla_0 v_{Z} \vert^2 \le (v_1'(t))^2\left(1-v_1^2(t)\right)^{-1} (1+\kappa^2)^2\vert x\vert_m^{-2}.
\end{equation}
Since $v_1'(t)(1-v_1^2(t))^{-\frac{1}{2}} \rightarrow 0$ as $t \nearrow \kappa$, we can find $0<\kappa''<\kappa$ so close to $\kappa$ that 
\begin{equation}
(v_1'(t))^2\left(1-v_1^2(t)\right)^{-1} (1+\kappa^2)^2\le\varepsilon v_1^2(t), \qquad \kappa'' \le t < \kappa.
\end{equation}
This implies
\begin{equation}
    \vert \nabla_0 u_{Z} \vert^2 + \vert \nabla_0 v_{Z} \vert^2 \le \varepsilon v_{Z}^2 \vert x\vert_m^{-2}, \qquad x \in K(Z,\kappa) \setminus K(Z,\kappa'').
\end{equation}
Now we define $v_{Z}$ for $x\in K(Z,\kappa'')$. By Lemma \ref{lem: localization scalar}, for given $\tilde{\varepsilon}>0$ we find a constant $0<\kappa'<\kappa''$ and a non-decreasing function $v_2$, such that 
\begin{equation}
v_2(\kappa')=0,\ v_2(\kappa'')=v_1(\kappa'') \quad \text{and}\quad  \left(v_2'(t)\right)^2 \le \tilde{\varepsilon}  \vert t\vert^{-2}\ln^{-2} t \quad \text{for} \quad \kappa'<t<\kappa''.
\end{equation}
Let $v_2$ be such a function and for $x\in K(Z,\kappa''), \ x=q+\xi$, let
\begin{equation}
    v_{Z}(x) = v_2\left(\frac{\vert q\vert_m}{\vert \xi \vert_m}\right), \qquad u_{Z}(x) = \sqrt{1-v_{Z}^2(x)}.
\end{equation}
Then, similar to \eqref{eq: gradient u + gradient v} we have
\begin{equation}\label{eq: compute nabla u + nabla v}
    \left(\vert \nabla_0 u_{Z} \vert^2 + \vert \nabla_0 v_{Z} \vert^2\right) u_{Z}^{-2}
    = (v_2'(t))^2\left(1-v_2^2(t)\right)^{-1}u_{Z}^{-2}\left(1+\vert q\vert^2_m\vert \xi\vert_m^{-2}\right)\vert \xi\vert_m^{-2},
\end{equation}
where $t=\vert q\vert_m\vert \xi\vert_m^{-1}$.
Since $v_2$ is non-decreasing, we have $(1-v_2^2(t))^{-1}u_{Z}^{-2}\le \left(1-v_2^2(k'')\right)^{-2}$ for $t\le \kappa''$.
Substituting this estimate into \eqref{eq: compute nabla u + nabla v} we have
\begin{equation}
    \left(\vert \nabla_0 u_{Z} \vert^2 + \vert \nabla_0 v_{Z} \vert^2\right) u_{Z}^{-2} \le (v_2'(t))^2\left(1-v_2(k'')^2\right)^{-2} \left(1+(\kappa'')^2\right) \vert \xi\vert_m^{-2}.
\end{equation}
 Recall that $v_2(\kappa'')$ is close to one, but strictly less then one. Due to $\left({v}_2'(t)\right)^2 \le \tilde \varepsilon \vert t\vert^{-2} \ln^{-2} t$ we get
\begin{equation}
\left(\vert \nabla_0 u_{Z} \vert^2 + \vert \nabla_0 v_{Z} \vert^2\right) u_{Z}^{-2} \le  \tilde \varepsilon \vert t\vert^{-2} \ln^{-2} t \left(1-v_2(k'')^2\right)^{-2} \left(1+(\kappa'')^2\right) \vert \xi\vert_m^{-2}.
\end{equation}
Choosing $\tilde{\varepsilon}>0$ so small that $\tilde{\varepsilon} \left(1-v_2(k'')^2\right)^{-2} \left(1+(\kappa'')^2\right) <\varepsilon$ and using $t=\vert q\vert_m\vert \xi\vert_m^{-1}$ completes the proof of Theorem \ref{lem: localization with log}.
\end{proof} 
\subsection{Proof of Theorem \ref{thm: Virtual level and Hardy constant}}
Now we turn to the proof of Theorem \ref{thm: Virtual level and Hardy constant}. It is an application of Theorem \ref{thm: abstraktes thm paper2} and we use geometrical methods to prove that all conditions of the latter theorem are fulfilled. 
Since the pair potentials $V_{ij}$ are relatively form bounded, so is $V= \sum_{1\le i < j\le N} V_{ij}(x_{ij})$. Hence, we only need to show that condition \eqref{eq: subtract alpha0} is fulfilled for any $0\le\alpha<\tilde{C}_H(X_0)$. This is done in the following
\begin{lem}\label{lem abziehlemma N=4}
Let $d\in\{1,2\}$ and $N\ge 3$. Assume that the potentials $V_{ij}$ satisfy \eqref{cond: relatively form bounded} and \eqref{cond: decay at infinity}. Further, suppose that $H$ has a virtual level at zero. Then for any $0\le\alpha < \tilde{C}_H(X_0)$ there exist constants $\gamma_0,\, R>0$, such that for any function $\varphi\in  H^1(X_0)$ with $\supp (\varphi) \subset \{x\in X_0: |x|_m\geq R \}$ we have
\begin{equation}\label{eq: L ge 0}
L[\varphi]:=(1-\gamma_0) \Vert \nabla_0 \varphi \Vert^2+ \langle V\varphi,\varphi \rangle - \alpha^2 \Vert |x|_m^{-1} \varphi \Vert^2\ge 0.
\end{equation}
\end{lem}
In the proof of Lemma \ref{lem abziehlemma N=4} we use the following 
\begin{lem}\label{lem: estimate two-dimensional with subst}
Let $Z$ be a partition of the system, such that $\dim(X_0(Z))=2$. Furthermore, let $0<\kappa<1$. Then there exists $\varepsilon>0$, such that for any function $\psi \in H^1(R_0)$ with $\mathrm{supp}(\psi)\subset K_R(Z,\kappa)$ and any $0<\kappa'<\kappa$ we have
\begin{equation}
\Vert \nabla_{q} \psi \Vert^2 -\varepsilon\left\Vert  \vert q\vert^{-1}_m\ln^{-1}\left(\vert q\vert_m\vert \xi\vert^{-1}_m\right)\psi\right\Vert^2_{K_R(Z,\kappa',\kappa)} \ge 0.
\end{equation}
\end{lem}
\begin{proof}[Proof of Lemma \ref{lem: estimate two-dimensional with subst}]
We introduce the new variable $y=\frac{q}{|\xi|_m}$. Then we get
\begin{align}\label{eq: substitution for two dim Hardy}
    \begin{split}
    &\Vert \nabla_{q} \psi \Vert^2 -\varepsilon\left\Vert  \vert q\vert^{-1}_m\ln^{-1}\left(\vert q\vert_m\vert \xi\vert^{-1}_m\right)\psi\right\Vert^2_{K_R(Z,\kappa',\kappa)} \\
    & \ge \int \int_{\{\kappa'| \xi|_m \le |q|_m\le \kappa| \xi|_m\}}\left(|\nabla_{q} \psi|^2-\varepsilon\vert q\vert^{-2}_m\left|\ln^{-2}\left(\vert q\vert_m\vert \xi\vert^{-1}_m\right)\right| |\psi|^2\right)\, \mathrm{d}q\,\mathrm{d}\xi\\
    &=\int \frac{1}{|\xi|_m^2}\int_{\{\kappa'|  \le |y|_m\le \kappa\}} \left(|\nabla_{y} \tilde\psi(y,\xi)|^2-\varepsilon\vert y\vert^{-2}_m\left|\ln^{-2}\left(|y|_m\right)\right| |\tilde\psi(y,\xi)|^2\right)\, \mathrm{d}y\, \mathrm{d}\xi,
    \end{split}
\end{align}
{where $\tilde{\psi}(y,\xi)= \psi(y|\xi|_m,\xi)$.
Note that $\tilde\psi(y,\xi)=0$ for $|y|_m\ge\kappa$.} 
Due to $\kappa <1$ we have $(\ln|y|_m)^{-2} \le C(1+(\ln|y|_m)^2)^{-1}$ for some $C>0$ and $|y|_m\le \kappa$. Therefore, applying Corollary \ref{lem: two dimensional Hardy radially symmetric functions} to the function $\tilde\psi(y,\xi)$ for fixed $\xi$ shows that the r.h.s. of \eqref{eq: substitution for two dim Hardy} is non-negative for sufficiently small $\varepsilon>0$. This completes the proof of Lemma \ref{lem: estimate two-dimensional with subst}.
\end{proof}
Now we turn to the
\begin{proof}[Proof of Lemma \ref{lem abziehlemma N=4}]
The proof follows the idea of the proof of Theorem 4.4 in \cite{second}.  We make a partition of unity of the support of $\varphi$, separating regions $K(Z,\kappa)$ which correspond to different partitions $Z$ of the system into clusters. \\
 We start by estimating the functional $L[\varphi]$ in regions $K(Z,\kappa)$ corresponding to partitions $Z$ into two clusters. Let $\kappa_2\in(0,1)$ be so small that $K_R(Z,\kappa_2)$ and $K_R(Z',\kappa_2)$ for clusters $Z\neq Z'$ with $|Z|=|Z'|=2$ do not overlap. Such a constant $\kappa_2$ exists according to \cite{cones} (an English version can be found in  \cite[Theorem B.2]{second}). By Theorem \ref{lem: localization with log} we get
\begin{equation}
L[\varphi]\ge \sum_{Z:|Z|=2} L_2[\varphi u_{Z}] + L_2'[\mathcal{V}^{(2)}\varphi],
\end{equation}
where $\mathcal{V}^{(2)}=\sqrt{1-\sum_{Z:|Z|=2}u_{Z}^2}$ and the functionals $L_2, L_2':H^1(X_0)\rightarrow\mathbb{R}$ are given by
\begin{align}\label{eq: L2 and L2'}
\begin{split}
        L_2[\psi] &= (1-\gamma_0)\Vert \nabla_0 \psi \Vert^2 + \langle V\psi,\psi \rangle- \alpha^2\Vert |x|_m^{-1}\psi \Vert^2\\
        &\qquad \qquad \qquad \qquad  - {\varepsilon} \Vert |q(Z)|_m^{-1} \ln^{-1} \left(\vert q(Z)\vert_m\vert \xi(Z)\vert^{-1}_m\right)\psi\Vert^2_{K_R(Z,\kappa_2',\kappa_2)},\\
        L_2'[\psi]&=(1-\gamma_0)\Vert \nabla_0 \psi \Vert^2 + \langle V\psi,\psi \rangle- (\alpha^2+\varepsilon)\Vert |x|_m^{-1}\psi \Vert^2,
\end{split}
\end{align}
where $\varepsilon>0$ can be chosen arbitrarily small if $\kappa_2'>0$ is sufficiently small. Recall that the functions $u_{Z}$ are supported in the region $K(Z,\kappa_2)$, i.e where the two clusters in $Z$ are far away from each other. Note also 
that the terms $\varepsilon\Vert |x|_m^{-1}\psi \Vert^2$ and  ${\varepsilon} \Vert |q(Z)|_m^{-1} \ln^{-1} \left(\vert q(Z)\vert_m\vert \xi(Z)\vert^{-1}_m\right)\psi\Vert^2_{K_R(Z,\kappa_2',\kappa_2)}$ come from the estimate for the localization error given in Theorem \ref{lem: localization with log}.\\ 
Let $Z$ be an arbitrary partition into two clusters, $q=q(Z),\ \xi=\xi(Z)$ and $\psi =\varphi u_{Z}$. Our goal is to show that $L_2[\psi]\ge 0$.
We have
\begin{align}\label{eq: L2 ausgeschrieben}
\begin{split}
    L_{2}[\psi]=&\langle H(Z)\psi,\psi\rangle-\gamma_0 \Vert \nabla_{q} \psi \Vert^2  +\left(1-\gamma_0\right) \left\Vert\nabla_{\xi} \psi\right\Vert^2+\langle I(Z)\psi,\psi\rangle\\&   \qquad -\alpha^2\left\Vert \vert x\vert^{-1}_m\psi \right\Vert^2 -\varepsilon\left\Vert  \vert q\vert^{-1}_m\ln^{-1}\left(\vert q\vert_m\vert \xi\vert^{-1}_m\right)\psi\right\Vert^2_{K_R(Z,\kappa_2',\kappa_2)}.
\end{split}
\end{align}
First, we estimate the inter-cluster potential $I(Z)$ by
\begin{equation}
	|I(Z)(x)|\leq C|\xi|_m^{-2-\nu}\le \varepsilon |\xi|_m^{-2}
\end{equation}
for $x\in \mathrm{supp}(\psi)$ and sufficiently large $R>0$. Furthermore, on the support of $\psi$ we have $|q|_m \leq \kappa_2 |\xi|_m$ and therefore the Poincar\'{e}-Friedrich inequality \cite[Theorem 6.30]{adams} yields
\begin{equation}\label{1: friedrichs}
{\gamma_0}\Vert \nabla_q \psi \Vert^2 \geq \frac{\gamma_0}{4\kappa_2^2}\Vert |\xi|_m^{-1}\psi\Vert^2.
\end{equation}
By choosing $\kappa_2>0$ small enough this implies
\begin{equation}
	{\gamma_0}\Vert \nabla_q \psi \Vert^2+\langle I(Z)\psi,\psi\rangle-\alpha^2\left\Vert \vert x\vert^{-1}_m\psi \right\Vert^2\ge 0
\end{equation}
and therefore 
\begin{equation}\label{eq: estimate L2 only localisation}
	L_2[\psi]\ge \langle H(Z)\psi,\psi\rangle-2\gamma_0\Vert \nabla_q \psi \Vert^2-\varepsilon\left\Vert  \vert q\vert^{-1}_m\ln^{-1}\left(\vert q\vert_m\vert \xi\vert^{-1}_m\right)\psi\right\Vert^2_{K_R(Z,\kappa_2',\kappa_2)}.
\end{equation}
To estimate the r.h.s. of \eqref{eq: estimate L2 only localisation} we distinguish between several cases.\\
\textbf{(i)} If $\dim(X_0(Z))=1$, we have $d=1$ and $N=3$. Assume that $Z=(C_1,C_2)$ with $|C_1|=2$, then $H[C_2]=0$ and
\begin{equation}
	\langle H(Z)\psi,\psi\rangle=\langle H[C_1]\psi,\psi\rangle \quad \text{and}\quad \Vert \nabla_{q(Z)} \psi\Vert =\Vert\nabla_{q[C_1]} \psi\Vert.
\end{equation}
We estimate the last term on the r.h.s. of \eqref{eq: L2 ausgeschrieben} by
\begin{equation}\label{eq: estimate localisation error d=1 N=3}
\varepsilon\left\Vert  \vert q\vert^{-1}_m\ln^{-1}\left(\vert q\vert_m\vert \xi\vert^{-1}_m\right)\psi\right\Vert^2_{K_R(Z,\kappa_2',\kappa_2)}\le \varepsilon\Vert (1+|q|_m)^{-1}\psi\Vert^2_{K_R(Z,\kappa_2',\kappa_2)}
\end{equation}
for $\kappa_2>0$ small enough and $R>0$ sufficiently large. This yields 
\begin{equation}
	L_2[\psi]\ge \langle H[C_1]\psi,\psi\rangle-2\gamma_0\Vert \nabla_{q[C_1]}\psi\Vert^2- \varepsilon\Vert (1+|q|_m)^{-1}\psi\Vert^2_{K_R(Z,\kappa_2',\kappa_2)}.
\end{equation}
Since by the remark after Definition \ref{def: Vitrtual levels multiparticle} the operator $H[C_1]$ does not have a virtual level and $V[C_1]\neq 0$, we can use Theorem \ref{lem no virtual level subtract d=1} to conclude that $L_2[\psi]\ge 0$ for $\varepsilon>0$ and $\gamma_0>0$ small enough and $R>0$ sufficiently large.\\
\textbf{(ii)} If $\dim(X_0(Z))\ge 2$, we use again that for clusters $C$ with $1<|C|<N$ the operator $H[C]$ does not have a virtual level, which implies
\begin{equation}
	\langle H(Z)\psi,\psi\rangle -3\gamma_0\Vert \nabla_q \psi\Vert^2\ge 0
\end{equation}
for small $\gamma_0>0$ and therefore
\begin{equation}
	L_2[\psi]\ge \gamma_0\Vert \nabla_q \psi\Vert^2-\varepsilon\left\Vert  \vert q\vert^{-1}_m\ln^{-1}\left(\vert q\vert_m\vert \xi\vert^{-1}_m\right)\psi\right\Vert^2_{K_R(Z,\kappa_2',\kappa_2)}.
\end{equation}
  If $\dim(X_0(Z))=2$, we apply Lemma \ref{lem: estimate two-dimensional with subst} with $\kappa=\kappa_2$ and $\kappa'=\kappa_2'$ to conclude that $L_2[\psi]\ge 0$. If $\dim(X_0(Z))\ge 3$, we use \eqref{eq: estimate localisation error d=1 N=3} and Hardy's inequality in the form of \eqref{eq: Hardy exterior domain d>=3}.  This implies $L_2[\psi]\ge 0$.\\
Now we estimate $L_2'[\mathcal{V}^{(2)}\varphi]$. If $N=3$, the function $\mathcal{V}^{(2)}\varphi$ is supported in the region where all particles are separated, i.e. there exists a constant $c>0$, such that $|x_{ij}|\ge c|x|_m$ for all $i\neq j$.  This implies
\begin{equation}
	|V(x)|\le C |x|_m^{-2-\nu} \le \varepsilon|x|_m^{-2},
\end{equation}
 where $\varepsilon>0$ can be chosen arbitrarily small for sufficiently large $R>0$. Therefore,
\begin{equation}
	 L_2'[\mathcal{V}^{(2)}\varphi]\ge(1-\gamma_0)\Vert \nabla_0 (\mathcal{V}^{(2)}\varphi) \Vert^2 - (\alpha^2+2\varepsilon)\Vert |x|_m^{-1}\mathcal{V}^{(2)}\varphi \Vert^2.
\end{equation}
Since $\mathcal{V}^{(2)}\varphi$ can be approximated (in the norm of $H^1(X_0)$) by functions in  $\mathcal{M}$, we get
\begin{equation}\label{eq: estimate gradient by Hardy particles separated}
	\Vert \nabla_0 (\mathcal{V}^{(2)}\varphi) \Vert^2 \ge \left(\tilde{C}_H(X_0)\right)^2 \Vert |x|_m^{-1}\mathcal{V}^{(2)}\varphi \Vert^2.
\end{equation}
Due to $\alpha<\tilde{C}_H(X_0) $ we obtain $L_2'[\mathcal{V}^{(2)}\varphi]\ge 0$ for the case $ N=3$ by choosing $\gamma_0,\varepsilon>0$ small enough.\\
If $N\ge 4$, we make a partition of unity of the support of $\mathcal{V}^{(2)}\varphi$. Let $\kappa_3\in (0,1)$ be so small that $K(Z,\kappa_3)$ and $K(\tilde{Z},\kappa_3)$ do not overlap on the support of $\mathcal{V}^{(2)}\varphi$ for partitions $Z\neq \tilde{Z}$ with $|Z|=|\tilde{Z}|=3$. Such a constant $\kappa_3$ exists due to \cite{cones} (see \cite[Theorem B.2]{second} for an English version). Applying Theorem \ref{lem: localization with log} we get
\begin{equation}
    L_2'[\mathcal{V}^{(2)}\varphi]\ge \sum_{Z:|Z|=3} L_3[\mathcal{V}^{(2)}\varphi u_{Z}] + L_3'[\mathcal{V}^{(3)}\varphi],
\end{equation}
where $\mathcal{V}^{(3)}=\mathcal{V}^{(2)}\sqrt{1-\sum_{Z}u_{Z}^2}$ and the functionals $L_3, L_3':H^1(X_0)\rightarrow\mathbb{R}$ are given by
\begin{align}\begin{split}
        L_3[\psi] &= (1-\gamma_0)\Vert \nabla_0 \psi \Vert^2 + \langle V\psi,\psi \rangle- (\alpha^2+\varepsilon)^2\Vert |x|_m^{-1}\psi \Vert^2\\
        &\qquad \qquad \qquad \qquad  - {\varepsilon} \Vert |q(Z)|_m^{-1} \ln^{-1} \left(\vert q(Z)\vert_m\vert \xi(Z)\vert^{-1}_m\right)\psi\Vert^2_{K_R(Z,\kappa_3',\kappa_3)},\\
        L_3'[\psi]&=(1-\gamma_0)\Vert \nabla_0 \psi \Vert^2 + \langle V\psi,\psi \rangle- (\alpha^2+\varepsilon)\Vert |x|_m^{-1}\psi \Vert^2
\end{split}
\end{align}
for some $\varepsilon>0$ which can be chosen arbitrarily small. Let $Z$ be an arbitrary partition into three clusters.  Then by the same arguments as for partitions $Z$ with $|Z|=2$ we can prove $L_3[\mathcal{V}^{(2)}\varphi u_{Z}]\ge 0$.
If $N\ge 5$, we continue this process for all partitions $Z$ with $|Z|\le N-1$ and finally arrive at the point where it remains to estimate
the functional
\begin{equation}
	L'[\tilde \psi] :=(1-\gamma_0)\Vert \nabla_0 \tilde \psi \Vert^2 + \langle V\tilde \psi,\tilde \psi \rangle- (\alpha^2+\varepsilon)\Vert |x|_m^{-1}\tilde \psi \Vert^2 \ge 0
\end{equation}
for functions $\tilde\psi:=\mathcal{V}^{(N-1)}\varphi$ supported in the region where all particles are separated from each other, i.e. there exists a constant $c>0$, such that $|x_{ij}|\ge c|x|_m $ for $x\in \supp(\mathcal{V}^{(N-1)}\varphi).$ Therefore, we have
\begin{equation}
	 \vert V(x)\vert \le C(1+|x|_m)^{-2-\nu}\le \varepsilon (1+|x|_m)^{-2}
\end{equation}
on the support of $\mathcal{V}^{(N-1)}\varphi$ if $R>0$ is large enough. This implies
\begin{equation}
    L'[\mathcal V^{(N-1)}\varphi]\ge (1-\gamma_0)\Vert \nabla_0\big(\mathcal V^{(N-1)}\varphi\big) \Vert^2 -(\alpha^2+2\varepsilon)\Vert |x|_m^{-1}\mathcal V^{(N-1)}\varphi \Vert^2.
\end{equation}
Similarly to \eqref{eq: estimate gradient by Hardy particles separated} we have
\begin{equation}
\left\Vert \nabla_0 \left(\mathcal V^{(N-1)}\varphi\right)\right\Vert^2\ge \left(\tilde{C}_H(X_0)\right) ^2\left\Vert |x|_m^{-1} \mathcal V^{(N-1)}\varphi\right\Vert^2.
\end{equation}
Since $\alpha<(\tilde{C}_H(X_0))$, we can choose $\gamma_0,\, \varepsilon>0$ sufficiently small to obtain $L'[\mathcal V^{(N-1)}\varphi]\ge 0$. This completes the proof of Lemma \ref{lem abziehlemma N=4} and therefore the proof of Theorem \ref{thm: Virtual level and Hardy constant}.
\end{proof}
\subsection{Proof of Theorem \ref{thm d=1 N=3}}\label{subsection proof d=1 N=3}
Now we turn to the proof of Theorem \ref{thm d=1 N=3}. Since in the case of three one-dimensional particles the configuration space $X_0$ is two-dimensional, we are able to improve the geometric methods and therefore to derive the exact value of the constant $\tilde{C}_H(X_0)$. The proof of the theorem follows from the following lemma, where we collect some geometric properties of the space $X_0$.
\begin{lem}\label{lem: geometry of R0}
Let $d=1$ and $N=3$. Then the following statements hold.
\begin{enumerate} 
\item[\textbf{(i)}]
The lines $x_1=x_2,\ x_1=x_3$ and $x_2=x_3$ divide the space $X_0$ into six sectors $S_1, S_2, \dots, S_6$ with angles $\theta_1=\theta_4,\, \theta_2=\theta_5$ and $\theta_3=\theta_6$. The angles $\theta_i, \, i=1,2,3$ are given by
\begin{equation}\label{eq: angles theta i}
    \theta_i=\arccos\left(\frac{\sqrt{m_j m_k}}{\sqrt{m_i+m_j}\sqrt{m_i+m_k}}\right).
\end{equation}
\item[\textbf{(ii)}] Let 
$\psi\in H_0^1(S_i)$. Then we have
\begin{equation}\label{eq: Hardy cone}
	\Vert \nabla_0 \psi\Vert \ge \frac{\pi}{\theta_i}\Vert |x|_m^{-1}\psi\Vert.
\end{equation}
\end{enumerate}
\end{lem}
\begin{proof}[Proof of Lemma \ref{lem: geometry of R0}]
The half lines $x_1=x_2\ge 0$, $x_1=x_3\ge 0$ and $x_2=x_3\le 0$ in $X_0$ are spanned by the vectors 
\begin{align}
\begin{split}
	u_{12} &=\left(1,1,-\frac{m_1+m_2}{m_3}\right)^\top, \quad 	u_{13} =\left(1,-\frac{m_1+m_3}{m_2},1\right)^\top\\
	\text{and} \quad u_{23} &=\left(\frac{m_2+m_3}{m_1},-1,-1\right)^\top, \quad \text{respectively}.
\end{split}
\end{align}
Let $S_1$ be the sector between the half-lines $x_1=x_2\le 0$ and $x_1=x_3\ge 0$, $S_2$ the sector between the half-lines $x_1=x_2\ge 0$ and $x_2=x_3\le 0$ and $S_3$ the sector between the half-lines $x_2\le x_3\ge 0$ and $x_1=x_3\ge 0$. Here, we always choose the one sector with angle $0<\theta_i<\pi$, see Figure \ref{fig: sectors}. To illustrate the situation we choose an orthogonal basis $\{v_1,v_2\}$ of $X_0$ with $v_1=u_{12}$ and $v_2=\left(m_2,-m_1,0  \right)^\top $.
\begin{figure}[h]
\begin{center}
\begin{tikzpicture}[scale=0.8]
     \draw[->] (2,0) -- (3,2)node[above right]{$x_2=x_3\le 0$} ;
     \draw[->] (0,0) -- (4,0)node[right]{$x_1=x_2\ge 0$};
     \draw[->] (2,0)--(0,2) node[above]{$x_1=x_3\ge 0$};
      \draw[densely dotted, ->] (2,0) -- (2,3)node[above]{$v_2$} ;
          \node at (3.5,0.9) {$S_2$};
          \node at (1.7,1.5) {$S_3$};
            \node at (0.5,0.7) {$S_1$};
            \draw  (1,0) arc (180:135:1 cm);
             \node at (1.4,0.25) {$\theta_1$};
              \draw  (3,0) arc (0:60:1 cm);
             \node at (2.6,0.3) {$\theta_2$};
               \draw  (1.29,0.71) arc (135:60:1 cm);
             \node at (1.9,0.5) {$\theta_3$};
\end{tikzpicture}
\caption{The sectors $S_1,\ S_2,\ S_3$}\label{fig: sectors}
\end{center}
\end{figure}
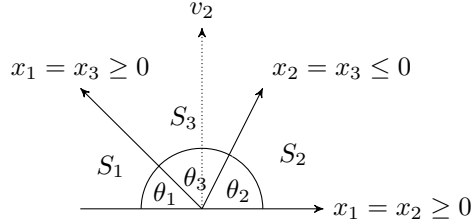$ $ \\
Let $S_4, \ S_5$ and $S_6$ be the sectors which we get by reflecting the sectors $S_1, S_2$ and $S_3$ at the origin. Obviously, $\theta_i=\theta_{i+3},\ i=1,2,3.$ 
Since $\langle - u_{12}, u_{13}\rangle_m > 0$, we have $\theta_1 \in \left(0,\frac{\pi}{2}\right)$ and analogously we see that $\theta_2,\theta_3\in \left(0,\frac{\pi}{2}\right).$ The angle $\theta_1$ can be computed by the formula
\begin{equation}
	\cos(\theta_1)= \frac{\langle -u_{12}, u_{13}\rangle_m}{|u_{12}|_m |u_{13}|_m} = \frac{\sqrt{m_2 m_3}}{\sqrt{m_1+m_2}\sqrt{m_1+m_3}}.
\end{equation}
Similarly we can see that the angles $\theta_2$ and $\theta_3$ also satisfy \eqref{eq: angles theta i}.
This completes the proof of statement \textbf{(i)} of Lemma \ref{lem: geometry of R0}.\\
Now we turn to the proof of the Hardy-type inequality \eqref{eq: Hardy cone} for the sectors $S_i$.
According to \cite[Proposition 4.1]{nazarov} functions $v\in H^1(\mathbb{R}^2)$ supported in a sector $S\subset \mathbb{R}^2$ satisfy
    \begin{equation}\label{eq: Hardy in cone}
        \Vert \nabla v \Vert \ge  \left(\Lambda(G)\right)^{\frac{1}{2}}\Vert \vert x\vert^{-1} v\Vert,
    \end{equation}
    were $\Lambda(G)$ is the first eigenvalue of the Dirichlet problem for the Laplace-Beltrami operator in $G=S \cap \mathbb{S}^1$. In dimension two $G$ can be identified with the interval $(0,\theta)$ where $\theta$ is the angle of $S$. The Dirichlet eigenvalues of the Laplacian on an interval of length $l>0$ are given by $\lambda_k=\left(\frac{k\pi}{l}\right)^2$. Therefore, we have $\Lambda(G) = \left(\frac{\pi}{\theta}\right)^2$, which implies that for any function $v\in H^1(\mathbb{R}^2)$ supported in $S$ we have
    \begin{equation}\label{eq: Hardy in cone 2D}
    \Vert \nabla v \Vert \ge  \frac{\pi}{\theta}\Vert \vert x\vert^{-1} v\Vert.
\end{equation}
This completes the proof of Lemma \ref{lem: geometry of R0} and therefore of Theorem \ref{thm d=1 N=3}.
\end{proof}
\section{Virtual levels of systems of three two-dimensional particles}\label{section: d=2 N=3}
In this section we consider systems of three two-dimensional particles. This is the only case of multi-particle systems in lower dimensions where we have $\tilde{C}_H(X)=1$, which leaves a possibility for virtual levels to correspond to resonances and not to eigenvalues. We give the following
\begin{thm}[Virtual levels of systems of three two-dimensional particles]\label{thm d=2 N=3}
Let $H$ be the Hamiltonian of a system of three two-dimensional particles. Assume that the potentials $V_{ij}\neq 0$ satisfy \eqref{cond: relatively form bounded} and \eqref{cond: decay at infinity} and that $H$ has a virtual level at zero. Then there exists a function $\varphi_0\in \tilde{H}^1(X_0), \ \varphi_0\neq 0$, satisfying 
\begin{equation}
	\Vert \nabla_0\varphi_0\Vert^2+\langle V\varphi_0,\varphi_0\rangle =0
\end{equation}
and 
\begin{equation}\label{eq: estimate phi for d=2 N=3}
	\left(1+|x|_m\right)^{-\alpha}\varphi_0 \in L^2(X_0)\qquad \text{for any } \alpha>0.
	\end{equation}
\end{thm}

\begin{proof}
To prove Theorem \ref{thm d=2 N=3} we take a sequence $(\psi_n)_{n\in \mathbb{N}}$ of eigenfunctions corresponding to eigenvalues $E_n<0$ of the operator $H+n^{-1}\Delta_0$, i.e.
\begin{equation}\label{1: eigenfunctions}
-\left( 1-n^{-1}\right)\Delta_0 \psi_n+V\psi_n=E_n\psi_n.
\end{equation}
We normalize the functions $\psi_n$ by $\Vert \nabla_0\psi_n\Vert=1$. Then there exists a subsequence of $(\psi_n)_{n\in\mathbb{N}}$, also denoted by $(\psi_n)_{n\in\mathbb{N}}$, which converges weakly in $\tilde{H}^1(X_0)$ to a function $\varphi_0\in \tilde{H}^1(X_0)$. Due to the Rellich-Kondrachov theorem we have convergence of $\psi_n$ to $\varphi_0$ in $L_{\mathrm{loc}}^2(X_0)$.\\
At first, we show that $\varphi_0\neq 0$ and establish the decay property \eqref{eq: estimate phi for d=2 N=3} of the function $\varphi_0$. 
Due to Lemma \ref{lem abziehlemma N=4} there exist constants $\gamma_0>0$ and $R>0$, such that for every function $\psi\in H^1(X_0)$ supported in the region $\{|x|_m\ge R\}$
\begin{equation}
	(1-\gamma_0)\Vert \nabla_0 \psi\Vert^2 +\langle V\psi,\psi\rangle \ge 0. 
\end{equation}
Applying Lemma 2.3 in \cite{second} we see that the weak limit $\varphi_0\in \tilde{H}^1(X_0)$ of the sequence $(\psi_n)_{n\in \mathbb{N}}$ of eigenfunctions normalized by $\Vert \nabla_0\psi_n\Vert=1$ is not  zero.
\\ In the next step we show that $\varphi_0$ satisfies the estimate \eqref{eq: estimate phi for d=2 N=3} on the decay rate. To do this we first give the following estimate for a weighted $L^2$ norm of the functions $\psi_n$.
\begin{lem}\label{lem: L2 bound d=2 N=3}
Let $H$ be the Hamiltonian of a system of three two-dimensional particles. Assume that the potentials $V_{ij}$ satisfy \eqref{cond: relatively form bounded} and \eqref{cond: decay at infinity} and that $H$ has a virtual level at zero.
Then, for any $0\le\alpha<1$ there exists a constant $C>0$, such that for all $n\in \mathbb{N}$ we have
\begin{equation}\label{eq: estimate norm psi_n d=2, N=3}
	\Vert \nabla_0\left(|x|_m^{\alpha}\psi_n\right)\Vert \le C \qquad \text{and} \qquad \Vert \left(1+|x|_m\right)^{\alpha-1}\psi_n\Vert \le C.
\end{equation}
\end{lem}
\begin{proof}
The proof is a straightforward modification of the proof of Lemma 2.4 in \cite{second}.
\end{proof}
 By Lemma \ref{lem: L2 bound d=2 N=3} we get convergence of $(\psi_n)_{n\in\mathbb{N}}$ to $\varphi_0$ in $L^2(X_0,(1+|x|_m)^{-\alpha}\mathrm{d}x)$ for any $\alpha>0$. This shows that the function $\varphi_0$ satisfies \eqref{eq: estimate phi for d=2 N=3}.\\
Our next goal is to prove that 
\begin{equation}
	\Vert \nabla_0 \varphi_0\Vert^2+\langle V\varphi_0,\varphi_0\rangle =0.
\end{equation}
Note that first we have to prove that $\langle V\varphi_0,\varphi_0\rangle$ is well defined. Since we do not know whether $\phi_0$ is square-integrable, we can not use the arguments of Lemma \ref{lem: phi 0 eigenfunction}. We prove that $\langle V_{ij}\varphi_0,\varphi_0\rangle$ is well-defined for each pair of particles $\beta=(i,j)$. By Corollary \ref{cor: Estimate Vu^2} 
to do this, it is sufficient to show that
\begin{equation}\label{eq: integral small q finite}
\int \int_{\{|q_{\beta}|_m\le 1\}} |\varphi_0|^2\,\mathrm{d}q_{\beta} \,\mathrm{d}\xi_{\beta} <\infty .
\end{equation}
In other words, it is enough to prove that the restriction of the function $\varphi_0$ to cylindrical regions $\{|q_{\beta}|_m\le 1\}$, $\beta\in\{(1,2),(1,3),(2,3)\}$, is square-integrable. Here and in the following for $\beta=(i,j)$ we denote by $q_{\beta}, \xi_{\beta}$ the variables $q[C],\,\xi[C]$, where $C=\{i,j\}$. \\
 To prove \eqref{eq: integral small q finite} we need to make several steps.  Let $(\psi_n)_{n\in \mathbb{N}}$ be the sequence of eigenfunctions of the operator $H+n^{-1}\Delta_0$, normalized by $\Vert \nabla_0\psi_n\Vert=1$. Furthermore, let $\chi_1:\mathbb{R}_+\rightarrow [0,1]$ be a function with $\chi_1\in C^1(\mathbb{R}_+)$ and $(1-\chi_1^2)^{\frac{1}{2}}\in C^1(\mathbb{R}_+)$, satisfying
\begin{equation}
  \chi_1(t)=0, \, 0\le t\le 1,\qquad \chi_1(t)=1, \, t\ge 2.
  \end{equation}
For $b>0$ let $\chi(x)=\chi_1\left(\frac{|x|_m}{b}\right).$
 The first step to prove that $\langle V_{ij}\varphi_0,\varphi_0\rangle$ is well-defined is the following 
\begin{lem}\label{lem: d=2,N=3 estimate estimate gradient outside}
Let $\psi_n$ and $\chi$ be defined as above. Then, for any $\varepsilon>0$ we can find $b>0$ and $n_0\in \mathbb{N}$, such that for all $n>n_0$ we have
\begin{equation}
	\textbf{(i)} \quad \Vert \nabla_0 \left(\chi \psi_n\right) \Vert <\varepsilon, \qquad \textbf{(ii)}\quad  \langle V_{ij}\chi \psi_n,\chi \psi_n\rangle < \varepsilon, \quad i,j\in \{1,2,3\}.
\end{equation}
\end{lem}
\begin{proof}[Proof of Lemma \ref{lem: d=2,N=3 estimate estimate gradient outside}]
For $\psi\in H^1(X_0)$ let 
\begin{equation}
L[\psi]=\Vert \nabla_0\psi\Vert^2+\langle	 V\psi,\psi\rangle.
\end{equation}
Then, by definition of the functions $\psi_n$ we have $L[\psi_n]\le \frac{1}{n}\Vert \nabla_0\psi_n\Vert^2= \frac{1}{n}$. On the other hand, by the IMS localisation formula we get
\begin{align}\label{eq: IMS d=2, N=3}
	L[\psi_n]=L[(1-\chi^2)^{\frac{1}{2}}\psi_n]+L[\chi\psi_n]-\int_{X_0} \left(|\nabla_0\chi|^2+|\nabla_0(1-\chi^2)^{\frac{1}{2}}|^2\right)|\psi_n|^2\,\mathrm{d}x.
\end{align}
We estimate the terms on the r.h.s. of \eqref{eq: IMS d=2, N=3} separately. Due to $H\ge 0$ the first term is non-negative. Since $\chi$ is supported in the region $\{|x|_m\ge b\}$, by Lemma \ref{lem abziehlemma N=4} with $\alpha=0$ we get 
\begin{equation}\label{eq: estimate L chi psi}
	L[\chi\psi_n]\ge \gamma_0\Vert \nabla_0\left(\chi\psi_n\right)\Vert^2
\end{equation}
for some $\gamma_0>0$ if $b>0$ is large enough. Now we estimate the last term on the r.h.s. of \eqref{eq: IMS d=2, N=3}. Note that $\nabla_0 \chi$ and $\nabla_0 \left(1-\chi^2\right)^{\frac{1}{2}}$ are supported in the region $\{b\le |x|\le 2b\}$ and satisfy 
\begin{equation}\label{eq: estimate cutoff gradient}
	| \nabla_0 \chi|^2+|\nabla_0  \left(1-\chi^2\right)^{\frac{1}{2}}|^2\le \frac{C}{b^2}
\end{equation}
for some $C>0$ which does not depend on $b$. This, together with the estimate \eqref{eq: estimate norm psi_n d=2, N=3} on the decay rate of $\psi_n$ we get, uniformly in $n\in \mathbb{N}$,
\begin{equation}
	\int_{X_0} \left(|\nabla_0\chi|^2+|\nabla_0(1-\chi^2)^{\frac{1}{2}}|^2\right)|\psi_n|^2\,\mathrm{d}x \le 4C \int_{\{|x|_m\ge b\}}\frac{|\psi_n|^2}{|x|_m^2}\,\mathrm{d}x\le \varepsilon_1(b)
\end{equation}
for some $\varepsilon_1(b)$ with $\varepsilon_1(b)\rightarrow 0$ as $b\rightarrow \infty$. Combining this with \eqref{eq: IMS d=2, N=3} and \eqref{eq: estimate L chi psi} we obtain 
\begin{equation}\label{eq: estimate L psin gradient chi psi}
L[\psi_n]\ge \gamma_0\Vert \nabla_0\left(\chi \psi_n\right)\Vert^2-\varepsilon_1(b).
\end{equation}
Since $L[\psi_n]\le \frac{1}{n}$, it follows from \eqref{eq: estimate L psin gradient chi psi} that for fixed $\varepsilon>0$ we can choose $n_0\in \mathbb{N}$ and $b>0$ large enough, such that $\Vert \nabla_0\left(\chi \psi_n\right)\Vert^2 \le \varepsilon$ holds uniformly for $n\ge n_0$. This completes the proof of statement \textbf{(i)} of the Lemma. \\
Now we turn to the proof of assertion \textbf{(ii)}. At first, we note that for any pair $(i_0,j_0)$ we have 
\begin{equation}\label{eq: representation V d=2, N=3}
	\langle V_{i_0j_0} \chi \psi_n,\chi\psi_n\rangle=L[\chi\psi_n]-\Vert \nabla_0\left(\chi\psi_n\right)\Vert^2- \sum\limits_{(i,j)\neq (i_0,j_0)} \langle V_{ij}\chi\psi_n,\chi \psi_n\rangle,
\end{equation} 
i.e. $\langle V_{i_0j_0} \chi \psi_n,\chi\psi_n\rangle$ can be estimated by estimating the r.h.s. of \eqref{eq: representation V d=2, N=3}. For the first term we get by \eqref{eq: IMS d=2, N=3} and \eqref{eq: estimate cutoff gradient}
\begin{equation}
	L[\chi\psi_n]\le L[\psi_n]+C\int_{\{|x|_m\ge b\}} \frac{|\psi_n|^2}{|x|_m^2}\,\mathrm{d}x.
\end{equation}
Now, by using $L[\psi_n]\le \frac{1}{n}$ and the estimate \eqref{eq: estimate norm psi_n d=2, N=3} for the functions $\psi_n$ we obtain
\begin{equation}
	L[\chi\psi_n] \le \frac{1}{n}+\varepsilon_2(b),
\end{equation}
where $\varepsilon_2(b)\rightarrow 0$ as $b\rightarrow \infty$. 
Substituting this in \eqref{eq: representation V d=2, N=3} we get
\begin{equation}\label{eq: estimate V d=2, N=3}
	\langle V_{i_0j_0} \chi \psi_n,\chi\psi_n\rangle \le \frac{1}{n}+\varepsilon_2(b) -\sum\limits_{(i,j)\neq (i_0,j_0)} \langle V_{ij}\chi\psi_n,\chi \psi_n\rangle.
\end{equation}
Now we estimate the last term on the r.h.s. of \eqref{eq: estimate V d=2, N=3}. Since the Hamiltonians of the clusters consisting of two particles do not have negative spectrum, we have
\begin{equation}
\langle V_{ij}\chi\psi_n,\chi \psi_n\rangle\ge -\Vert  \nabla_0\left(\chi\psi_n\right)\Vert^2\ge -\varepsilon,
\end{equation}
where according to statement \textbf{(i)} of the Lemma the constant $\varepsilon>0$ can be chosen arbitrarily small if $b>0$ and $n\in \mathbb{N}$ are sufficiently large.
Inserting this in \eqref{eq: estimate V d=2, N=3} we get
\begin{equation}
	\langle V_{i_0j_0} \chi \psi_n,\chi\psi_n\rangle \le \frac{1}{n}+\varepsilon_2(b)+2\varepsilon,
\end{equation}
which completes the proof of Lemma \ref{lem: d=2,N=3 estimate estimate gradient outside}. 
\end{proof}
Now we turn to the proof of the well-definedness of $\langle V_{ij}\varphi_0,\varphi_0\rangle$. Recall that we need to show that
\begin{equation}
	\int\int_{\{|q_{\beta}|_m\le 1\}} |\varphi_0|^2\,\mathrm{d}q_{\beta}\,\mathrm{d}\xi_{\beta} <\infty.
\end{equation}
Since the cluster Hamiltonians for non-trivial clusters do not have virtual levels and $V_{ij}\neq 0$, by the remark \textbf{(iii)} after Theorem \ref{lem no virtual level subtract d=1} we get
\begin{equation}\label{eq: integral kleine q d=2, N=3}
	\int\int_{\{|q_{\beta}|_m\le 1\}} |\chi\psi_n|^2\,\mathrm{d}q_{\beta}\,\mathrm{d}\xi_{\beta}\le C_1 \Vert \nabla_{q_{\beta}}(\chi\psi_n)\Vert^2+C_2\langle V_{ij}\chi\psi_n,\chi\psi_n\rangle
\end{equation}
for some constants $C_1,C_2>0$ and $\beta=(i,j)$. Now by Lemma \ref{lem: d=2,N=3 estimate estimate gradient outside} we see that the r.h.s. of \eqref{eq: integral kleine q d=2, N=3} can be done arbitrarily small if the constant $b>0$ in the definition of the function $\chi$ and $n\in\mathbb{N}$ are sufficiently large. Hence, for any $\varepsilon>0$ we find $b>0$, such that
\begin{equation}
	\int\int_{\{|q_{\beta}|_m\le 1\}} |\chi\psi_n|^2\,\mathrm{d}q_{\beta}\,\mathrm{d}\xi_{\beta}\le \varepsilon.
\end{equation}
Recall that for $|\xi_{\beta}|_m>2b$ we have $\chi(x)=1$ and therefore
\begin{equation}\label{eq: integral xi gross q klein}
\int_{\{|\xi_{\beta}|_m\ge 2b\}}\int_{\{|q_{\beta}|_m\le 1\}} |\psi_n(x)|^2 \,\mathrm{d}x= \int_{\{|\xi_{\beta}|_m\ge 2b\}}\int_{\{|q_{\beta}|_m\le 1\}} |\chi\psi_n(x)|^2 \,\mathrm{d}x\le \varepsilon
\end{equation}
for $b>0$ and $n\in \mathbb{N}$ large enough. Furthermore, we have $\psi_n\rightarrow \varphi_0$ in $L^2_{\mathrm{loc}}(X_0)$. Therefore, we get
\begin{equation}
\int_{\{|\xi_\beta|_m\le 2b\}} \int_{\{|q_{\beta}|_m\le 1\}} |\psi_n|^2\,\mathrm{d}q_{\beta}\,\mathrm{d}\xi_{\beta} \rightarrow\int_{\{|\xi_\beta|_m\le 2b\}} \int_{\{|q_{\beta}|_m\le 1\}} |\varphi_0|^2\,\mathrm{d}q_{\beta}\,\mathrm{d}\xi_{\beta}.
\end{equation}
This, together with \eqref{eq: integral xi gross q klein} shows that the integral 
\begin{equation}
\int \int_{\{|q_{\beta}|_m\le 1\}} |\varphi_0|^2\,\mathrm{d}q_{\beta}\,\mathrm{d}\xi_{\beta}
\end{equation}
is bounded and thus $\langle V_{ij} \varphi_0,\varphi_0\rangle$ is well-defined. \\
Now we show that $\langle V_{ij} \psi_n,\psi_n\rangle \rightarrow \langle V_{ij}\varphi_0,\varphi_0\rangle$ as $n\rightarrow \infty$.
At first, we consider the integral
\begin{equation}
	\int_{\{|\xi_{\beta}|_m\ge 2b\}} \int|V_{ij}| |\psi_n|^2\,\mathrm{d}q_{\beta}\,\mathrm{d}\xi_{\beta}
\end{equation}
and prove that it can be done arbitrarily small if $b>0$ and $n\in \mathbb{N}$ are large enough. By Corollary \ref{cor: Estimate Vu^2} we have
\begin{align}\label{eq: hardy ungleichung potential angewandt}
\begin{split}
&\int_{\{|\xi_{\beta}|_m\ge 2b\}} \int |V_{ij}| |\psi_n|^2\,\mathrm{d}q_{\beta}\,\mathrm{d}\xi_{\beta}\,\\
&\qquad \qquad \le C\int_{\{|\xi_{\beta}|_m\ge 2b\}}\left(\int |\nabla_{q_{\beta}}\psi_n|^2\,\mathrm{d}q_{\beta}+\int_{\{|q_{\beta}|_m\le 1\}} |\psi_n|^2\,\mathrm{d}q_{\beta}\right)\,\mathrm{d}\xi_{\beta}.
\end{split}
\end{align}
Note that by Lemma \ref{lem: d=2,N=3 estimate estimate gradient outside} we get for arbitrary $\varepsilon>0$
\begin{equation}\label{eq: apply lemma gradient outside}
	\int_{\{|\xi_{\beta}|_m\ge 2b\}}\int |\nabla_{q_{\beta}} \psi_n|^2\,\mathrm{d}q_{\beta}\,\mathrm{d} \xi_{\beta} = \int_{\{|\xi_{\beta}|_m\ge 2b\}}\int |\nabla_{q_{\beta}} (\chi\psi_n)|^2\,\mathrm{d}q_{\beta}\,\mathrm{d} \xi_{\beta} \le \varepsilon
\end{equation}
if $b>0$ and $n\in \mathbb{N}$ are large enough.
Substituting this inequality and inequality \eqref{eq: integral xi gross q klein} in \eqref{eq: hardy ungleichung potential angewandt} yields 
\begin{equation}\label{eq: estimate Vpsin large xi}
	\int_{\{|\xi_{\beta}|_m\ge 2b\}}\int |V_{ij}| |\psi_n|^2\,\mathrm{d}q_{\beta}\,\mathrm{d}\xi_{\beta}\le 2\varepsilon.
\end{equation}
Due to 
\begin{equation}
	\int |V_{ij}| |\varphi_0|^2\,\mathrm{d}q_{\beta}\,\mathrm{d}\xi_{\beta} <\infty
\end{equation}
we also obtain
\begin{equation}\label{eq: estimate Vphi0 large xi}
	\int_{\{|\xi_{\beta}|_m\ge 2b\}}\int |V_{ij}| |\varphi_0|^2\,\mathrm{d}q_{\beta}\,\mathrm{d}\xi_{\beta}\le \varepsilon
\end{equation}
for $b>0$ large enough.
Now we consider the region $\{|\xi_{\beta}|_m\le 2b\}$. Due to the decay property \eqref{cond: decay at infinity} of the potentials $V_{ij}$ and the estimates \eqref{eq: estimate norm psi_n d=2, N=3} and $\eqref{eq: estimate phi for d=2 N=3}$ for the functions $\psi_n$ and $\varphi_0$ we get
\begin{equation}\label{eq: estimate Vpsin large q}
	\int_{\{|\xi_{\beta}|_m\le 2b\}}\int_{\{|q_{\beta}|_m\ge b_1\}}|V_{ij}| |\psi_n|^2\,\mathrm{d}q_{\beta}\,\mathrm{d}\xi_{\beta}<\varepsilon
\end{equation}
and 
\begin{equation}\label{eq: estimate Vphi0 large q}
\int_{\{|\xi_{\beta}|_m\le 2b\}}\int_{\{|q_{\beta}|_m\ge b_1\}}|V_{ij}| |\varphi_0|^2\,\mathrm{d}q_{\beta}\,\mathrm{d}\xi_{\beta}<\varepsilon
\end{equation}
where $\varepsilon>0$ can be chosen arbitrarily small if $b_1>0$ is large enough and estimate \eqref{eq: estimate Vpsin large q} holds uniformly in $n\in \mathbb{N}$.\\
Estimates \eqref{eq: estimate Vpsin large xi} - \eqref{eq: estimate Vphi0 large q} show that to prove convergence $\langle V_{ij}\psi_n,\psi_n\rangle\rightarrow \langle V_{ij}\varphi_0,\varphi_0\rangle$ it suffices to show that $\langle V_{ij}\psi_n,\psi_n\rangle_\Omega\rightarrow \langle V_{ij}\varphi_0,\varphi_0\rangle_\Omega$ for the compact set
\begin{equation}
\Omega:=\{x\in X_0\,:\,|q_{\beta}|_m\le b_1,\ |\xi_{\beta}|_m\le 2b\}.
\end{equation}
We write
\begin{equation}\label{eq: decomposition Vpsin minus Vphi0}
	\langle V_{ij} \psi_n ,\psi_n\rangle_\Omega - \langle V_{ij} \varphi_0 ,\varphi_0\rangle_\Omega= \langle V_{ij} (\psi_n-\varphi_0) ,\psi_n\rangle_\Omega+\langle V_{ij} \varphi_0, (\psi_n-\varphi_0)\rangle_\Omega.
\end{equation}
Since $\psi_n$ converges to $\varphi_0$ in $L^2_{\mathrm{loc}}(X_0)$, $\Vert \nabla_{q_{\beta}}\psi_n\Vert \le 1$, $\Vert \nabla_{q_{\beta}}\varphi_0\Vert \le 1$ and the potential $V_{ij}$ satisfies \eqref{cond: relatively form bounded},  both summands on the r.h.s. of \eqref{eq: decomposition Vpsin minus Vphi0} tend to zero as $n\rightarrow \infty$.  Combining this with the estimates \eqref{eq: estimate Vpsin large xi} - \eqref{eq: estimate Vphi0 large q} we conclude $\langle V_{ij}\psi_n,\psi_n\rangle \rightarrow \langle V_{ij}\varphi_0,\varphi_0\rangle$ for every pair $(i,j)$ of particles and therefore $\langle V\psi_n,\psi_n\rangle \rightarrow \langle V\varphi_0,\varphi_0\rangle$ as $n\rightarrow\infty$.\\
Since by definition of the functions $\psi_n$ 
\begin{equation}
\langle V\psi_n,\psi_n\rangle \le -\left(1-n^{-1}\right),
\end{equation}
we get $\langle V\varphi_0,\varphi_0\rangle \le -1$. On the other hand, $H\ge 0$ and $\Vert \nabla_0\varphi_0\Vert \le 1$. This shows 
\begin{equation}
	\Vert \nabla_0\varphi_0 \Vert^2+\langle V\varphi_0,\varphi_0\rangle =0,
\end{equation}
which completes the proof of Theorem \ref{thm d=2 N=3}.
\end{proof}
\section{Absence of the Efimov effect in multi-particle systems consisting of one- or two-dimensional particles}\label{Section No Efimov multi}
In this section we prove that the Efimov effect does not occur in systems of $N\ge 4$ one-dimensional or $N\ge 5$ two-dimensional particles.
The absence of the Efimov effect for such systems is mainly caused by the fact that in these cases virtual levels of the cluster Hamiltonians $H[C]$ with $|C|=N-1$ correspond to eigenvalues, as we have shown in Section \ref{section: virtual levels multi-particle}. We follow the strategy of the proof of Theorem 5.1 in \cite{second}, which itself is based on ideas of \cite{Semjon2}. However, on a technical level the proof in this section is slightly different from those in \cite{second} and \cite{Semjon2} because Hardy's inequality, which plays an important role in \cite{second} and \cite{Semjon2}, is different in lower dimensions. The main result of this section is the following
\begin{thm}\label{thm: No Efimov}
Let $d=1$ and $N\ge 4$ or $d=2$ and $N\ge 5$. Suppose that every pair potential $V_{ij}\neq 0$ satisfies \eqref{cond: decay at infinity} and is operator bounded with respect to $-\Delta$ with relative bound zero, i.e. for any $\varepsilon>0$ there exists a constant $C(\varepsilon)>0$, such that
\begin{equation}
\Vert V_{ij} \psi\Vert^2\le \varepsilon \Vert \Delta \psi\Vert^2+C(\varepsilon)\Vert \psi\Vert^2,\qquad \psi\in H^2(\mathbb{R}^d).
\end{equation}
Furthermore, assume that $H[C]\ge 0$ for all clusters $C$ with $|C|=N-1$ and there exists $\varepsilon\in (0,1)$, such that
\begin{equation}\label{eq: no virtual level}
	\mathcal{S}_{\mathrm{ess}}\left(-(1-\varepsilon)\Delta_0[C]+V[C]\right)=[0,\infty).
\end{equation}
Then the discrete spectrum of $H$ is finite.
\end{thm}
\begin{rem}
We emphasize that in Theorem \ref{thm: No Efimov} the cluster Hamiltonian $H[C]$ with $|C|=N-1$ may have a virtual level at zero. For clusters $C'$ with $1<|C'|<N-1$ however, the Hamiltonian $H[C']$ are not allowed to have a virtual level, which is a consequence of \eqref{eq: no virtual level} and the HVZ theorem. 
\end{rem}
\begin{proof}[Proof of Theorem \ref{thm: No Efimov}]
For $\varepsilon>0$ we define the functional $L:H^1(X_0)\rightarrow \mathbb{R}$ as
\begin{equation}
    L[\varphi]:= \langle H\varphi,\varphi\rangle -\varepsilon \Vert |x|_m^{-2}\varphi\Vert^2
\end{equation}
and prove that $L[\varphi]\ge 0$ for any function $\varphi\in H^1(X_0)$ with $\supp(\varphi)\subset \{\vert x\vert_m\ge R\}$ if $R>0$ is large enough and $\varepsilon>0$ is small enough. This implies finiteness of the discrete spectrum of $H$, see Lemma \ref{lem: Abstraktes Lemma endlich viele Eigenwerte} in Appendix \ref{App: Technical Lemmas} (see also \cite{Zhislin1}).
\\
We fix a constant $\kappa>0$, such that $K_R(Z,\kappa)\cap K_R(Z',\kappa)=\emptyset$ for all partitions $Z\neq Z'$ with $|Z|=|Z'|=2$. By applying Theorem \ref{lem: localization with log} we get
\begin{align}
    L[\varphi] \ge \sum\limits_{Z:|Z|=2} L_2[\varphi u_{Z}] + L_2'[\varphi \mathcal{V}],
\end{align}
where $\mathcal{V}=\sqrt{1-\sum_{Z:|Z|=2} u^2_{Z}}$ and the functionals $L_2$ and $L'$ are defined by
\begin{align}
\begin{split}
        L_2[\psi] &:= \langle H\psi,\psi\rangle - \varepsilon \Vert |x|_m^{-2}\psi\Vert^2\\
        &\qquad- \varepsilon_1 \left\Vert \vert q(Z)\vert_m^{-1} |\ln( |q(Z)|_m|\xi(Z)|_m^{-1} )|^{-1}\psi \right\Vert^2_{K_R(Z,\kappa',\kappa)},\\
        L_2' [\psi]&:=\langle H\psi, \psi\rangle - (\varepsilon+\varepsilon_1)  \Vert |x|_m^{-2}\psi\Vert^2.\label{1: L'}
        \end{split}
\end{align}
Here, the constants $\kappa>0$ and $\varepsilon_1>0$ can be chosen arbitrarily small and $\kappa'\in (0,\kappa)$ depends on $\kappa$ and $\varepsilon_1$ only. For the sake of simplicity we omit the index $Z$ in the following computations and write $q$ and $\xi$ instead of $q(Z)$ and $\xi(Z)$, respectively. At first, we prove $L_2[\varphi u_{Z}]\ge 0$. We distinguish between the following two types of partitions $Z=(C_1,C_2)$:
\begin{enumerate}
\item[\textbf{(i)}] $|C_1|<N-1$ and $|C_2|<N-1$,
\item[\textbf{(ii)}] $|C_1|=N-1$ or $|C_2|=N-1$.
\end{enumerate}
In the first case the operators $H[C_1]$ and $H[C_2]$ do not have virtual levels, which implies that there exists a constant $\mu_0>0$, such that
\begin{equation}
\langle H(Z)\psi ,\psi\rangle \geq \mu_0 \Vert \nabla_q \psi \Vert^2
\end{equation}
holds for any function $\psi \in H^1(X_0)$. 
Repeating the arguments which were used in the proof of Lemma \ref{lem abziehlemma N=4} we get $L_2[\varphi u_{Z}]\ge 0$.\\
We turn to case \textbf{(ii)}, where the Hamiltonian $H[C_1]$ or $H[C_2]$ may have a virtual level. Suppose that $|C_1|=N-1$ and that $H[C_1]$ has a virtual level. According to Theorem \ref{thm: Virtual level and Hardy constant}, Corollaries \ref{cor: d=2, N>3} and \ref{cor: d=1, N>3} and Theorem \ref{thm d=1 N=3} zero is a simple eigenvalue of the operator $H[C_1]$. Let $\varphi_0$ be the corresponding eigenfunction normalized by $\Vert \varphi_0 \Vert=1$.  Let
\begin{equation}\label{eq: definition f,g}
    f(\xi) := \Vert \nabla_q \varphi_0\Vert^{-2}\langle \nabla_q \left(\varphi u_{Z}(\cdot, \xi)\right),\nabla_q \varphi_0 \rangle_{{L^2(X_0(Z))}} 
\end{equation}
and 
\begin{equation}\label{eq: definition g}
 g(q,\xi) :=\varphi u_{Z}(q,\xi) -f(\xi)\varphi_0(q).
\end{equation}
Then we have
\begin{equation}\label{orthogonality of phi 0 f and g}
\varphi u_Z= f\varphi_0+g \quad \text{and}\quad \langle \nabla_{q}  g(\cdot,\xi), \nabla_{q}\varphi_0\rangle_{L^2(X_0(Z))} = 0
\end{equation}
for almost every $\xi$. For $|\xi|_m \le\frac{R}{2}$ we have $f(\xi)=0$ and $g(q,\xi) = 0$, because $\varphi u_{Z}=0$ for $|x|\le R$. We write
\begin{align}\label{eq: L2 decomposition}
\begin{split}
    L_2[\varphi u_{Z}] &= \langle H[C_1]\ g,g\rangle  +\langle H[C_1] \varphi_0f,\varphi_0f\rangle + 2\Re\langle g,H[C_1]\varphi_0f\rangle\\& \qquad + \Vert \nabla_{\xi} \left(\varphi u_{Z}\right)\Vert^2 
    +\langle I(Z)\varphi  u_{Z},\varphi  u_{Z}\rangle - \varepsilon  \Vert |x|_m^{-2}\varphi u_{Z}\Vert^2\\& \qquad \qquad - \varepsilon_1 \Vert \vert q\vert_m^{-1} |\ln(|q|_m|\xi|_m^{-1})|^{-1}\varphi  u_{Z}\Vert^2_{K_R(Z,\kappa',\kappa)}.
     \end{split}
\end{align}
Due to $H[C_1]\varphi_0=0$ the second term and the third term on the r.h.s. of \eqref{eq: L2 decomposition} are zero.
Now we estimate the term $\langle I(Z)\varphi  u_{Z},\varphi  u_{Z}\rangle$. For fixed $\varepsilon_2>0$ we get
\begin{equation}
|I(Z)(x)| \le C|\xi|_m^{-2-\nu}\le \frac{\varepsilon_2}{4} |\ln(|\xi|_m)|^{-2}|\xi|_m^{-2}
\end{equation}
for $x\in K_R(Z,\kappa)$ if $R>0$ is large enough. Since $\varphi u_{Z}(q,\xi)=0$ for $|\xi|_m\le \frac{R}{2}$, we can apply the one- or two- dimensional Hardy inequality in the $\xi$-variable to obtain
\begin{equation}\label{eq: Hardy for I(Z2)}
 |	\langle I(Z)\varphi  u_{Z},\varphi  u_{Z}\rangle|\le \frac{\varepsilon_2}{4} \Vert |\ln(|\xi|_m)|^{-1}|\xi|_m^{-1}\varphi u_{Z}\Vert^2\le \varepsilon_2 \Vert \nabla_\xi (\varphi u_{Z})\Vert^2.
\end{equation}
This, together with \eqref{eq: L2 decomposition} implies
\begin{align}\label{eq: absense efimov L2 estimate |x|^-2}
\begin{split}
    L_2[\varphi u_{Z}] \geq & \langle H[C_1] g,g\rangle  +(1-\varepsilon_2) \Vert \nabla_{\xi}\left( \varphi u_{Z}\right)\Vert^2
    - \varepsilon  \Vert |x|_m^{-2}\varphi u_{Z}\Vert^2\\ 
    &\ \qquad   \qquad \ \ \ \ \ \ \qquad- \varepsilon_1 \Vert |\ln(|q|_m|\xi|_m^{-1})|^{-1}\vert q\vert_m^{-1} \varphi  u_{Z}\Vert^2_{K_R(Z,\kappa',\kappa)}.
\end{split}
\end{align}
Since 
\begin{equation}
	\left\Vert |x|_m^{-2}\varphi u_{Z}\right\Vert^2\le \left\Vert |\xi|_m^{-1} (\ln^{-1}|\xi|_m)\varphi u_{Z}\right\Vert^2
\end{equation}
for $ |x|_m>1$ and we have $|\xi|_m\ge \frac{R}{2}$ on the support of $\varphi u_{Z}$, we get
\begin{equation}
	4\varepsilon \Vert \nabla_\xi \left(\varphi u_{Z}\right)\Vert^2-\varepsilon\Vert |x|_m^{-2}\varphi u_{Z}\Vert^2\ge 0.
\end{equation}
Substituting this inequality into \eqref{eq: absense efimov L2 estimate |x|^-2} yields
\begin{align}\label{eq: L2 after 1d Hardy}
\begin{split}
	  L_2[\varphi u_{Z}] &\geq  \langle H[C_1] g,g\rangle  +(1-\varepsilon_3) \Vert \nabla_{\xi}\left( \varphi u_{Z}\right)\Vert^2 \\&\qquad -\varepsilon_1 \Vert |\ln(|q|_m|\xi|_m^{-1})|^{-1} \vert q\vert_m^{-1} \varphi  u_{Z}\Vert^2_{K_R(Z,\kappa',\kappa)},
\end{split}
\end{align}
where $\varepsilon_3=\varepsilon_2+4\varepsilon.$ Now we estimate the term
\begin{equation}\label{eq: H0g,g - loc error}
\langle H[C_1] g,g\rangle -\varepsilon_1 \Vert |\ln(|q|_m|\xi|_m^{-1})|^{-1} \vert q\vert_m^{-1} \varphi  u_{Z}\Vert^2_{K_R(Z,\kappa',\kappa)}.
\end{equation}
This is done in the following
\begin{lem}\label{lem: estimate H0g g}
Let $1<\alpha<\tilde{C}_H(X_0)$ and let $C_1$ be a cluster with $|C_1|=N-1$ and the functions $f,g$ be defined by \eqref{eq: definition f,g} and \eqref{eq: definition g}. Then for $\varepsilon_1>0$ small enough and $R>0$ sufficiently large
\begin{align}
\begin{split}
\langle H[C_1] g,g\rangle  -\varepsilon_1 \Vert |\ln(|q|_m|\xi|_m^{-1})|^{-1} \vert q\vert_m^{-1} \varphi  u_{Z}&\Vert^2_{K_R(Z,\kappa',\kappa)}\\
&\ge - \int_{\{|\xi|_m\ge \frac{R}{2}\}} |\xi|_m^{-2\alpha}|f(\xi)|^2\, \mathrm{d}\xi.
\end{split}
\end{align}
\end{lem}
\begin{proof}[Proof of Lemma \ref{lem: estimate H0g g}]
Due to Theorem \ref{thm: Virtual level and Hardy constant} the orthogonality in \eqref{orthogonality of phi 0 f and g} implies
\begin{equation}\label{eq: estimate Hg,g}
      \langle H[C_1] g,g\rangle \ge \delta_0 \Vert\nabla_{q} g\Vert^2
 \end{equation}
for some $\delta_0>0$. Therefore,
\begin{align}\label{eq: estimate delta gradient g - loc error}
\begin{split}
\langle H[C_1] g,g\rangle &-\varepsilon_1 \Vert |\ln(|q|_m|\xi|_m^{-1})|^{-1} \vert q\vert_m^{-1} \varphi  u_{Z}\Vert^2_{K_R(Z,\kappa',\kappa)}\\
& \ge \delta_0 \Vert\nabla_{q} g\Vert^2-\varepsilon_1 \Vert |\ln(|q|_m|\xi|_m^{-1})|^{-1} \vert q\vert_m^{-1} \varphi  u_{Z}\Vert^2_{K_R(Z,\kappa',\kappa)}.
\end{split}
\end{align}
Since $\varphi u_{Z} = \varphi_0 f + g$, we have 
\begin{equation}
|\nabla_q(\varphi u_{Z})|^2= |\nabla_q (\varphi_0f+g)|^2\le 2|\nabla_q \varphi_0f|^2+2|\nabla_q g|^2,
\end{equation}
which yields
\begin{equation}\label{eq: estimate gradient of g}
\Vert \nabla_q g\Vert^2_{K_R(Z,\kappa',\kappa)}\ge \frac{1}{2}\Vert \nabla_q(\varphi u_{Z})\Vert^2_{K_R(Z,\kappa',\kappa)}-\Vert \nabla_q \varphi_0 f\Vert^2_{K_R(Z,\kappa',\kappa)}.
\end{equation}
Since $\varphi u_{Z}=0$ for $|q|_m=\kappa|\xi|_m$, we get similarly as in the proof of Lemma \ref{lem abziehlemma N=4}  that
\begin{equation}\label{eq: Hardy phi u}
\frac{\delta_0}{2}\Vert \nabla_q(\varphi u_{Z})\Vert^2_{K_R(Z,\kappa',\kappa)}-\varepsilon_1 \Vert |\ln(|q|_m|\xi|_m^{-1})|^{-1} \vert q\vert_m^{-1} \varphi  u_{Z}\Vert^2_{K_R(Z,\kappa',\kappa)} \ge 0
\end{equation}
if $\varepsilon_1>0$ is small enough.
Combining this inequality with \eqref{eq: estimate gradient of g} and \eqref{eq: estimate delta gradient g - loc error} yields
\begin{equation}\label{eq: Hgg nabla q phi0}
	\langle H[C_1]g,g\rangle -\varepsilon_1 \Vert |\ln(|q|_m|\xi|_m^{-1})|^{-1} \vert q\vert_m^{-1} \varphi  u_{Z}\Vert^2_{K_R(Z,\kappa',\kappa)}\ge-\delta_0 \Vert \nabla_q \varphi_0 f\Vert^2_{K_R(Z,\kappa',\kappa)}.
\end{equation}
Now we estimate the term $ \Vert \nabla_q \varphi_0 f\Vert^2_{K_R(Z,\kappa',\kappa)}.$
By Theorem \ref{thm: Virtual level and Hardy constant} we have
\begin{equation}
\left|\nabla_q\left(|q|_m^{\alpha}\varphi_0\right)\right|\in L^2(X_0(Z)) \quad \text{and}\quad  (1+|q|_m)^{\alpha-1}\varphi_0 \in L^2(X_0(Z))
\end{equation}
for any $0\le\alpha<\tilde{C}_H(X_0)$. 
This implies
\begin{equation}\label{eq: |q|gradient in L2}
	|q|_m^{\alpha} \left|\nabla_q \varphi_0\right|\in L^2(X_0(Z)).
\end{equation}
Due to $|\xi|_m\ge \frac{R}{2}$ for $x\in K_R(Z,\kappa)$ we get
\begin{align}
	\notag \int_{K_R(Z,\kappa',\kappa)} |\nabla_q \varphi_0 f|^2\,\mathrm{d}x &= \int_{\{|\xi|_m\ge \frac{R}{2}\}}  |f(\xi)|^2\int_{\kappa'|\xi|_m}^{\kappa|\xi|_m} |\nabla_q \varphi_0|^2\,\mathrm{d}q\, \mathrm{d}\xi\\
	& = \int_{\{|\xi|_m\ge \frac{R}{2}\}}  |f(\xi)|^2\int_{\kappa'|\xi|_m}^{\kappa|\xi|_m} |q|_m^{-2\alpha} |q|_m^{2\alpha} |\nabla_q \varphi_0|^2\,\mathrm{d}q\, \mathrm{d}\xi\\ \notag
	&\le (\kappa')^{-2\alpha}\int_{\{|\xi|_m\ge \frac{R}{2}\}}  |\xi|_m^{-2\alpha}|f(\xi)|^2\int_{\kappa'|\xi|_m}^{\kappa|\xi|_m} |q|_m^{2\alpha} |\nabla_q \varphi_0|^2\,\mathrm{d}q\,\mathrm{d}\xi,
\end{align}
where in the last inequality we used $|q|_m\ge \kappa'|\xi|_m$. Since $|q|_m^{\alpha}|\nabla_q\varphi_0|\in L^2(X_0(Z))$, we have
\begin{equation}
\int_{\kappa'|\xi|_m}^{\kappa|\xi|_m} |q|_m^{2\alpha} |\nabla_q \varphi_0|^2\,\mathrm{d}q \le (\kappa')^{2\alpha}\delta_0^{-1}
\end{equation}
for $|\xi|_m \ge \frac{R}{2}$ if $R>0$ is sufficiently large. This yields
\begin{equation}
-\delta_0 \Vert \nabla_q \varphi_0 f\Vert^2_{K_R(Z,\kappa',\kappa)}\ge - \int_{\{|\xi|_m\ge \frac{R}{2}\}} |\xi|_m^{-2\alpha}|f(\xi)|^2\, \mathrm{d}\xi,
\end{equation}
which completes the proof of Lemma \ref{lem: estimate H0g g}.
\end{proof}
We continue to estimate the functional $L_2[\varphi u_{Z}]$. Combining \eqref{eq: L2 after 1d Hardy} with Lemma \ref{lem: estimate H0g g} we get
\begin{equation}\label{eq: L2 after Lemma}
L_2[\varphi u_{Z}] \ge (1-\varepsilon_3)\Vert \nabla_\xi(\varphi u_{Z})\Vert^2 - \varepsilon_1\int_{\{|\xi|_m\ge \frac{R}{2}\}} |\xi|_m^{-2\alpha}|f(\xi)|^2\, \mathrm{d}\xi.
\end{equation}
In the next step we estimate the term $\Vert \nabla_\xi(\varphi u_{Z})\Vert^2$. This is done in the following
\begin{lem}\label{lem: estimate nabla_xi phi no efimov}
Let $\delta>0$. There exists a constant $\omega>0$ which depends on $\Vert \varphi_0\Vert$, $\Vert \nabla_q\varphi_0\Vert$ and $\Vert\Delta_q\varphi_0\Vert$ only, such that
\begin{equation}\label{eq: estimate nabla_xi phi no efimov}
	\Vert \nabla_{\xi}\left( \varphi u_{Z}\right)\Vert^2 \ge {\omega}\left(\Vert|\xi|_m^{-1-\delta}\varphi_0f \Vert^2+\Vert|\xi|_m^{-1-\delta}g \Vert^2\right).
\end{equation}
\end{lem}
\begin{rem}
For the case $\dim X_c(Z)=3$, a statement similar to Lemma \ref{lem: estimate nabla_xi phi no efimov} was proved in \cite{Semjon2}. In the proof of Lemma \ref{lem: estimate nabla_xi phi no efimov} we follow the ideas of this work.
\end{rem}
\begin{proof}[Proof of Lemma \ref{lem: estimate nabla_xi phi no efimov}]
Since $\varphi u_{Z}(q,\xi)=0$ for $|\xi|_m\le \frac{R}{2}$, we can apply the one- or two-dimensional Hardy inequality in the space $X_c(Z)$ to the function $\varphi u_{Z}(q,\cdot)$ for fixed $q$. This implies
\begin{align}\label{eq: estimate of gradient xi}
\begin{split}
	\Vert \nabla_{\xi}\left(\varphi u_{Z}\right)\Vert^2&\ge \frac{1}{4}\Vert |\xi|_m^{-1-\delta} \varphi u_{Z}\Vert^2 = \frac{1}{4} \Vert |\xi|_m^{-1-\delta}\varphi_0f+|\xi|_m^{-1-\delta}g\Vert^2\\
	&\ge \frac{1}{4}\left( \Vert|\xi|_m^{-1-\delta}\varphi_0f \Vert^2+\Vert|\xi|_m^{-1-\delta}g \Vert^2-2|\langle |\xi|_m^{-1-\delta}\varphi_0f ,|\xi|_m^{-1-\delta}g \rangle |\right).
\end{split}
\end{align}
Since $\langle \nabla_q \varphi_0 , \nabla_q g\rangle_{L^2(X_0(Z))}=0$, we have
\begin{equation}
	\langle \nabla_q |\xi|_m^{-1-\delta}\varphi_0 f, \nabla_q |\xi|_m^{-1-\delta}g\rangle=0
\end{equation}
and by Lemma 5.3 in \cite{Semjon2} we can find a constant $\omega>0$ which depends on $\Vert \varphi_0\Vert$, $\Vert \nabla_q\varphi_0\Vert$ and $\Vert\Delta_q\varphi_0\Vert$ only, such that 
\begin{equation}
	 \left|\langle |\xi|_m^{-1-\delta}\varphi_0f ,|\xi|_m^{-1-\delta}g \rangle\right| \le \frac{1}{2}\left(1-4\omega\right)\left(\Vert|\xi|_m^{-1-\delta}\varphi_0f \Vert^2+\Vert|\xi|_m^{-1-\delta}g \Vert^2\right).
\end{equation}
Substituting this inequality in \eqref{eq: estimate of gradient xi} yields
\begin{equation}
	\Vert \nabla_{\xi}\left( \varphi u_{Z}\right)\Vert^2 \ge \omega\left(\Vert|\xi|_m^{-1-\delta}\varphi_0f \Vert^2+\Vert|\xi|_m^{-1-\delta}g \Vert^2\right),
\end{equation}
which completes the proof of Lemma \ref{lem: estimate nabla_xi phi no efimov}.
\end{proof}
Combining \eqref{eq: L2 after Lemma} with \eqref{eq: estimate nabla_xi phi no efimov} and using $\Vert \varphi_0\Vert=1$ we get
 \begin{equation}
L_2[\varphi u_{Z}] \ge (1-\varepsilon_3){\omega}\int_{\{|\xi|_m\ge \frac{R}{2}\}}|\xi|_m^{-2-2\delta}|f(\xi)|^2\,\mathrm{d}\xi - \varepsilon_1\int_{\{|\xi|_m\ge \frac{R}{2}\}} |\xi|_m^{-2\alpha}|f(\xi)|^2\, \mathrm{d}\xi.
\end{equation}
Choosing $\delta<\alpha-1$ and $\varepsilon_1, \varepsilon_3>0$ small enough yields $L_2[\varphi u_{Z}]\ge 0$.
\\
To complete the proof of Theorem \ref{thm: No Efimov} it remains to show $L_2'[\varphi\mathcal{V}] \geq 0$ for every $\varphi\in H^1(X_0)$ with $\supp(\varphi)\subset\{x\in X_0\,:\,|x|_m\ge R\}$, where $L_2'$ is the functional defined in \eqref{1: L'}. 
Note that for all partitions $Z=(C_1,\dots,C_p)$ with $p=3,4,\dots , N-1$ the Hamiltonians $H[C_i]$ do not have a virtual level if $|C_i|>1$. Hence, we can estimate the functional $L_2'[\mathcal{V} \varphi]\ge 0$ in cones corresponding to partitions $Z$ with $|Z|\ge 3$ in the same way as in the proof of Lemma \ref{lem abziehlemma N=4}. In the region which remains after separation of the cones corresponding to all partitions $Z$ with $|Z|\le N-1$ we have $|V_{ij}(x_{ij})|\le |x|_m^{-2-\nu}$ for all $i\neq j$. Applying Hardy's inequality in the space $X_0$ completes the proof. 
\end{proof}
\section{Absence of the Efimov effect in systems \\of three one- or two-dimensional particles}\label{Section: No Efimov}
Now we prove that the Efimov effect is absent for systems of three one- or two-dimensional particles. This was first proved in \cite{Semjon4} under restrictive conditions on the pair potentials. There, the potentials had to be compactly supported or short-range and negative at infinity. Later, in \cite{Semjon3} the restrictions on the potentials were relaxed. Unfortunately, Lemma 1 in \cite{Semjon3} contains a mistake. Below we follow the ideas of \cite{Semjon3} and correct this mistake. We give the proof for both, the one- and the two-dimensional case.

\subsection{Systems of three one-dimensional particles}

\begin{thm}[Absence of the Efimov effect for systems of three one-dimensional particles]\label{thm: No Efimov N=3}
Let $H$ be the Hamiltonian corresponding to a system of $N=3$ one-dimensional particles. Suppose that each pair potential $V_{ij}$ satisfies \eqref{cond: relatively form bounded} and \eqref{cond: decay at infinity} and that $H[C]\ge 0$ for any cluster $C$ with $|C|=2$. Then the discrete spectrum of $H$ is finite.
\end{thm}
In the proof of Theorem \ref{thm: No Efimov N=3} we will use the following lemmas.
\begin{lem}\label{lem: Lemma 1 in diemension 1}
Consider the Schr\"odinger operator $h=-\Delta +V$ in $L^2(\mathbb{R})$, such that $h\geq 0$ and the potential $V$ satisfies \eqref{cond: relatively form bounded} and \eqref{cond: decay at infinity}. Then there exists a constant $C>0$, such that for any $b_0>A$ and any function $\psi \in H^1(\mathbb{R})$ 
\begin{equation}\label{APP: Lemma2}
J[\psi,b_0]:=\int_{-b_0}^{b_0} \left(|\psi'(t)|^2+V(t)|\psi(t)|^2\right) \, \mathrm{d}t \geq - Cb_0^{-1-\nu} \left(|\psi(b_0)|^2+|\psi(-b_0)|^2\right).
\end{equation}
Here, $\nu$ and $A$ are the constants given by \eqref{cond: decay at infinity}.
\end{lem}
\begin{proof}[Proof of Lemma \ref{lem: Lemma 1 in diemension 1}] Let $\psi \in H^1(\mathbb{R})$ and $b_0>A$. For $n\geq 2$ we define the function $\psi_n$ as 
$\psi_n(t)=\psi(t) $ for $-b_0 \leq t\leq b_0$, $\psi_n(t) = 0$ for $   t< -nb_0$ and for $t> nb_0 $,  $\psi_n(t)=\psi(-b_0)\frac{nb_0+t}{b_0(n-1)}$ for $ -nb_0<t<-b_0$ and $\psi_n(t)= \psi(b_0)\frac{nb_0-t}{b_0(n-1)}$ for $ b_0<t<nb_0.$ 
Since $\psi$ and $\psi_n$ coincide for $-b_0\le t\le b_0$, we have
\begin{align}\label{eq: estimate quadratic form one dimensional}
\begin{split}
 \langle h \psi_n,\psi_n\rangle &\leq \int_{-b_0}^{b_0}\left(|\psi'(t)|^2+V(t)|\psi(t)|^2\right) \, \mathrm{d}t +\int\limits_{-nb_0}^{-b_0}\left(|\psi_n'(t)|^2+|V(t)| |\psi_n(t)|^2\right)\mathrm{d}t 
\\ &\ \qquad  + \int\limits_{b_0}^{nb_0}\left(|\psi_n'(t)|^2+|V(t)||\psi_n(t)|^2\right) \, \mathrm{d}t.
\end{split}
 \end{align}
At first, we estimate the two last integrals of the r.h.s of \eqref{eq: estimate quadratic form one dimensional}. Since $\psi_n'(t)=\frac{\psi(-b_0)}{b_0(n-1)}$ for $t\in (-nb_0,-b_0)$ and $\psi_n'(t)=-\frac{\psi(b_0)}{b_0(n-1)}$ for $t\in (b_0,nb_0)$, we get
\begin{equation}\label{eq: estimate psi' by eps}
	\int\limits_{-nb_0}^{-b_0}|\psi_n'(t)|^2\,\mathrm{d}t+ \int\limits_{b_0}^{nb_0}|\psi_n'(t)|^2\,\mathrm{d}t <\varepsilon,
\end{equation}
where $\varepsilon>0$ can be chosen arbitrarily small if $n$ is large enough. Moreover, $0\le\frac{nb_0+t}{b_0(n-1)}\le 1$ for $t\in (-nb_0,-b_0)$.  
 This implies
 \begin{equation}
 \int\limits_{-nb_0}^{-b_0}|V(t)| |\psi_n(t)|^2\mathrm{d}t \le |\psi(-b_0)|^2\int \limits_{-nb_0}^{-b_0}|V(t)|\mathrm{d}t
 \end{equation}
 and analogously we get
\begin{equation}
 \int\limits_{b_0}^{nb_0}|V(t)| |\psi_n(t)|^2\mathrm{d}t \le |\psi(b_0)|^2\int \limits_{b_0}^{nb_0}|V(t)|\mathrm{d}t.
 \end{equation} 
 This, together with \eqref{eq: estimate quadratic form one dimensional} and \eqref{eq: estimate psi' by eps} yields
\begin{align}\label{eq: estimate quadratic form one dimensional 2}
\begin{split}
 \langle h \psi_n,\psi_n\rangle & \le J[\psi,b_0]+|\psi(-b_0)|^2\int\limits_{-nb_0}^{-b_0}|V(t)|\mathrm{d}t + |\psi(b_0)|^2\int\limits_{b_0}^{nb_0}|V(t)| \, \mathrm{d}t+\varepsilon.
\end{split}
 \end{align}
Now we estimate the integrals on the r.h.s of \eqref{eq: estimate quadratic form one dimensional 2}.
For any $0<\delta<\nu$ we have $-1-\nu+\delta<-1$. Since $b_0\ge A$, we get by \eqref{cond: decay at infinity}
\begin{equation}\label{eq: estimate integral v 1}
\int_{b_0}^{nb_0}|V(t)|\, \mathrm{d}t \leq c b_0^{-1-\delta} \int_{A}^{\infty}|t|^{-1-\nu+\delta}\, \mathrm{d}t \leq c_1 b_0^{-1-\delta}
\end{equation}
for some constants $c,c_1>0$. Analogously we have
\begin{equation}\label{eq: estimate integral v 2}
\int_{-nb_0}^{-b_0}|V(t)|\, \mathrm{d}t \leq c_1 b_0^{-1-\delta}.
\end{equation}
 Due to $h \geq 0$ we conclude from \eqref{eq: estimate quadratic form one dimensional 2}, \eqref{eq: estimate integral v 1} and \eqref{eq: estimate integral v 2} that
\begin{equation*}
J[\psi,b_0] \geq - cb_0^{-1-\delta} \left(|\psi(b_0)|^2+|\psi(-b_0)|^2\right)-\varepsilon.
\end{equation*}
Since $\varepsilon>0$ can be chosen arbitrarily small, this completes the proof.
\end{proof}

\begin{lem}\label{lem: integration in xi}
Let $C_0>0$. Then for any sufficiently large $b>0$ and for any $\psi\in H^1(\mathbb{R})$ 
\begin{equation}
\int_b^\infty \left(|\psi'(t)|^2-C_0 t^{-2-\nu}|\psi(t)|^2\right)\,\mathrm{d}t\ge -{2C_0b^{-1-\nu}} |\psi(b)|^2.
\end{equation}
\end{lem}
\begin{proof}[Proof of Lemma \ref{lem: integration in xi}]
Let $\psi \in H^1(\mathbb{R})$ and $\tilde{\psi}(t)= \psi(t)-\psi(b).$ Then $\tilde{\psi}'(t)=\psi'(t)$ and we have
\begin{align}\label{eq: subtract psi(b) to apply hardy}
\begin{split}
\int_b^\infty \left(|\psi'(t)|^2-C_0 t^{-2-\nu}|\psi(t)|^2\right)\,\mathrm{d}t &\ge \int_b^\infty \left(|\tilde \psi'(t)|^2-2C_0 t^{-2-\nu}|\tilde\psi(t)|^2\right)\,\mathrm{d}t\\
& \quad -2C_0 \int_b^\infty t^{-2-\nu} |{\psi}(b)|^2\,\mathrm{d}t.
\end{split}
\end{align}
Since $\tilde{\psi}(b)=0$, we can use the one-dimensional Hardy inequality, which for sufficiently large $b>0$ yields
\begin{equation}
	\int_b^\infty \left(|\tilde \psi'(t)|^2-2C_0 t^{-2-\nu}|\tilde\psi(t)|^2\right)\,\mathrm{d}t\ge 0.
\end{equation}
This, together with \eqref{eq: subtract psi(b) to apply hardy} implies
\begin{equation}\label{eq: inequality for integration in xi}
\int_b^\infty \left(|\psi'(t)|^2-C_0 t^{-2-\nu}|\psi(t)|^2\right)\,\mathrm{d}t \ge -2C_0|{\psi}(b)|^2 \int_b^\infty t^{-2-\nu} \,\mathrm{d}t.
\end{equation}	
Computing the integral on the r.h.s. of \eqref{eq: inequality for integration in xi} completes the proof.
\end{proof}

\begin{lem}\label{lem: Lemma 2 in diemension 1}
Let $b_2>b_1$. Then for any $\psi \in H^1(\mathbb{R})$ 
\begin{align}
|\psi(b_i)|^2 \leq 2(b_2-b_1)^{-1} \int_{b_1}^{b_2}|\psi(x)|^2\, \mathrm{d}x+2(b_2-b_1) \int_{b_1}^{b_2}|\psi'(x)|^2\, \mathrm{d}x, \quad i=1,2.
\end{align}
\end{lem}
\begin{rem}
Lemma \ref{lem: Lemma 2 in diemension 1} is the one-dimensional analogue of Lemma 2 in \cite{Semjon3}.
\end{rem}
\begin{proof}[Proof of Lemma \ref{lem: Lemma 2 in diemension 1}]
For $x\in(b_1,b_2)$ we write
\begin{equation}
\psi(x)=\int_{b_1}^x \psi'(t)\, \mathrm{d}t+\psi(b_1).
\end{equation}
Therefore, we have
\begin{equation}\label{eq: estimate psi(d1)}
|\psi(b_1)|^2 \leq 2|\psi(x)|^2+2\left(\int_{b_1}^{b_2} |\psi'(x)|\, \mathrm{d}x\right)^2, \quad x\in(b_1,b_2).
\end{equation}
Applying the Cauchy-Schwarz inequality to the integral on the r.h.s. of \eqref{eq: estimate psi(d1)} yields
\begin{equation}\label{eq: cauchy schwarz}
|\psi(b_1)|^2 \leq 2|\psi(x)|^2 + 2(b_2-b_1) \int_{b_1}^{b_2} |\psi'(x)|^2\, \mathrm{d}x, \quad x\in (b_1,b_2).
\end{equation}
Integrating both sides of \eqref{eq: cauchy schwarz} over $(b_1,b_2)$ and dividing by $(b_2-b_1)$ implies
\begin{equation}
|\psi(b_1)|^2 \leq 2(b_2-b_1)^{-1}\int_{b_1}^{b_2}|\psi(x)|^2\, \mathrm{d}x + 2(b_2-b_1) \int_{b_1}^{b_2} |\psi'(x)|^2\, \mathrm{d}x.
\end{equation}
Similarly we can prove the statement for $b_2$.
\end{proof}
\begin{proof}[Proof of Theorem \ref{thm: No Efimov N=3}] 
As in the proof of Theorem \ref{thm: No Efimov} we show that 
\begin{equation}
    L[\varphi]:= \int\left(|\nabla_0\varphi|^2+V|\varphi|^2-\varepsilon |x|_m^{-4}|\varphi|^2\right)\,\mathrm{d}x\ge 0
\end{equation}
holds for all functions $\varphi\in H^1(X_0)$ with $\supp(\varphi)\subset \{\vert x\vert_m\ge R\}$ if $\varepsilon>0$ is small enough and $R>0$ is sufficiently large. Let $Z=(C_1,C_2)$ be a partition into two clusters with $|C_1|=2$. First, we estimate the part of the quadratic form $L$ corresponding to the cone $K(Z,\kappa)$, where $\kappa>0$ is so small that cones $K(Z,\kappa)$ and $K(Z',\kappa)$ corresponding to different partitions do not overlap. Denote the particles in $C_1$ by $i$ and $j$ and the third particle by $k$.  In the following we will need subtle geometric arguments and therefore we introduce a basis of $X_0$ and work with the corresponding coordinates. Recall that $\dim(X_0(Z))=1$ and $\dim(X_c(Z))=1$. Choosing a vector $u_1\in X_0(Z)$ and a vector $u_2\in X_c(Z)$, both normalized with respect to the norm $|u_i|_m=1$, we get an orthonormal basis of $X_0$. Denote by $\tilde{q}$ and $\tilde{\xi} $ the coefficients corresponding to the basis $\{u_1,u_2\}$. Then we have $|q|_m=|\tilde{q}|$, $|\xi|_m=|\tilde{\xi}|$ and we can represent $K_R(Z,\kappa)$ as 
\begin{equation}
K_R(Z,\kappa)=\left\{(\tilde q,\tilde \xi)\in \mathbb{R}^2 \,:\, |\tilde{q}|\le \kappa |\tilde \xi|,\ |\tilde{q}|^2+|\tilde{\xi}|^2\ge R^2\right\}
\end{equation}
and $\varphi=\varphi(\tilde{q},\tilde{\xi})$ as a function of $\tilde{q}$ and $\tilde{\xi}$. We have
 \begin{align}\label{eq: quadratic form no Efimov N=3}
 \begin{split}
 	&\int_{K_R(Z,\kappa)}\left(|\nabla_0\varphi|^2+V|\varphi|^2-\varepsilon |x|_m^{-4}|\varphi|^2\right)\,\mathrm{d}x  =\int_{K_R(Z,\kappa)}\left(|\partial_{\tilde q}\varphi|^2+V_{ij}|\varphi|^2\right)\,\mathrm{d}x \\
 	&\qquad\qquad+\int_{K_R(Z,\kappa)}\left(|\partial_{\tilde \xi}\varphi|^2+(V_{ik}+V_{jk})|\varphi|^2-\varepsilon|x|_m^{-4}|\varphi|^2\right)\,\mathrm{d}x
 	\end{split}
 \end{align}
 and estimate the two integrals on the r.h.s of \eqref{eq: quadratic form no Efimov N=3} separately. By choosing $\kappa>0$ small enough we have $|\tilde \xi|\ge \frac{R}{2}$ and therefore
 \begin{equation}\label{eq: no efimov N=3 integral over q}
 	\int_{K_R(Z,\kappa)}\!\left(|\partial_{\tilde q}\varphi|^2+V_{ij}|\varphi|^2\right)\,\mathrm{d}x = \int_{\{|\tilde \xi|\ge\frac{R}{2}\}} \int_{\{|\tilde q|\le  \kappa|\tilde \xi|\}}\left(|\partial_{\tilde q}\varphi|^2+V_{ij}|\varphi|^2\right)\,\mathrm{d}\tilde q\, \mathrm{d}\tilde \xi.
 \end{equation}
 Applying Lemma \ref{lem: Lemma 1 in diemension 1} to the integral over $\tilde q$ in \eqref{eq: no efimov N=3 integral over q} with $b_0=\kappa|\tilde \xi|$ we get
 \begin{align}\label{eq: no efimov N=3 estimate q part}
 \begin{split}
 	&\int_{\{|\tilde \xi|\ge\frac{R}{2}\}} \int_{\{|\tilde q|\le  \kappa|\tilde \xi|\}}\left(|\partial_{\tilde q}\varphi|^2+V_{ij}|\varphi|^2\right)\,\mathrm{d}\tilde q\, \mathrm{d}\tilde \xi\\
 	& \quad \ge -C \int_{\{|\tilde \xi|\ge\frac{R}{2}\}} |\tilde \xi|^{-1-\nu} \left(|\varphi(\kappa \tilde \xi,\tilde \xi)|^2+|\varphi(-\kappa \tilde \xi,\tilde \xi)|^2\right)\,\mathrm{d}\tilde \xi.
 	\end{split}
 \end{align}
 Now we estimate the second integral on the right hand side of \eqref{eq: quadratic form no Efimov N=3}. Note that for $R>0$ sufficiently large we have 
 \begin{equation}
 	|V_{ik}(x_{ik})|+|V_{jk}(x_{jk})|\le c|\tilde \xi|^{-2-\nu}
 \end{equation}
 for some $c>0$. This, together with $|x|_m^{-1}\le |\tilde \xi|^{-1}$ implies
 \begin{align}\label{eq: integral cone with xi}
 \begin{split}
& \int_{K_R(Z,\kappa)}\left(|\partial_{\tilde \xi}\varphi|^2+(V_{ik}+V_{jk})|\varphi|^2-|x|_m^{-4}|\varphi|^2\right)\,\mathrm{d}x\\ &  \qquad \qquad\qquad\qquad\qquad\ge \int_{K_R(Z,\kappa)}\left(|\partial_{\tilde \xi}\varphi|^2-C|\tilde \xi|^{-2-\nu}|\varphi|^2\right)\,\mathrm{d}x
 \end{split}
 \end{align}
 for some $C>0$, where without loss of generality we assumed that $\nu<2$. 
{To estimate the integral on the r.h.s. of \eqref{eq: integral cone with xi} we first integrate over the variable $\tilde \xi$ for fixed $\tilde q$. Let us describe the domain of integration first. The integral over $\tilde \xi$ is from the boundary of $K_R(Z,\kappa)$ to infinity. Note that if $|\tilde q|$ is small, then the boundary of $K_R(Z,\kappa)$ is given by an arc with radius $R$, see Figure 2. 
 \begin{figure}[h!]
   \begin{tikzpicture}[scale=1]
      \path [fill, pattern=north west lines ,opacity=.6]
      (1.8477,0.765) -- (4.619,1.9125)--(4.619,-1.9125) -- (1.8477,-0.765);
              \path [fill, pattern=north west lines ,opacity=.6]
      (-1.8477,0.765) -- (-4.619,1.9125)--(-4.619,-1.9125) -- (-1.8477,-0.765);     
      \path [fill=white] (0,0) circle(2cm);
      \path [fill,blue,opacity=.1] (0,0) circle(2cm);
      \draw[->] (-5,0) -- (5.5,0) node[below] {$\tilde \xi$};
      \draw[->] (0,-2.5) -- (0,3) node[left] {$\tilde q$};
      \draw[blue] (-0.5,-0.5) node {\scriptsize  $B(R)$};
      \draw[thick] (3,0.5) node {\scriptsize  \textbf{$K_R(Z,\kappa_2)$}};
      \draw[blue] (0,0) circle (2cm);
      \draw[thick,gray] (1.8477,0.765) -- (4.619,1.9125);
      \draw[thick,gray] (1.8477,-0.765) -- (4.619,-1.9125);
       \draw[thick,gray] (1.8477,0.765) -- (4.619,1.9125);
      \draw[dashed,thick,gray,->] (4.1574,1.722) -- (5.17356,2.142) node [right] {$\tilde q = \kappa \tilde \xi$};
      \draw[dashed,thick,gray,->] (4.1574,-1.722) -- (5.17356,-2.142) node [right] {$\tilde q = -\kappa \tilde \xi$};
      \draw (1,0.1)--(1,-0.1) node[below]{$\frac{R}{2}$};
      \draw (-0.1,0.765) -- (0.1,0.765) ;
      \node at (-0.4,0.765) {$\eta$};
      \draw[thick,gray] (-1.8477,-0.765) -- (-4.619,-1.9125);
       \draw[thick,gray] (-1.8477,0.765) -- (-4.619,1.9125);  
 \draw[dashed,thick,gray]    (0,0.765) --  (1.8477,0.765);
    \end{tikzpicture}
     \caption{The cone $K_R(Z,\kappa)$}
\end{figure}
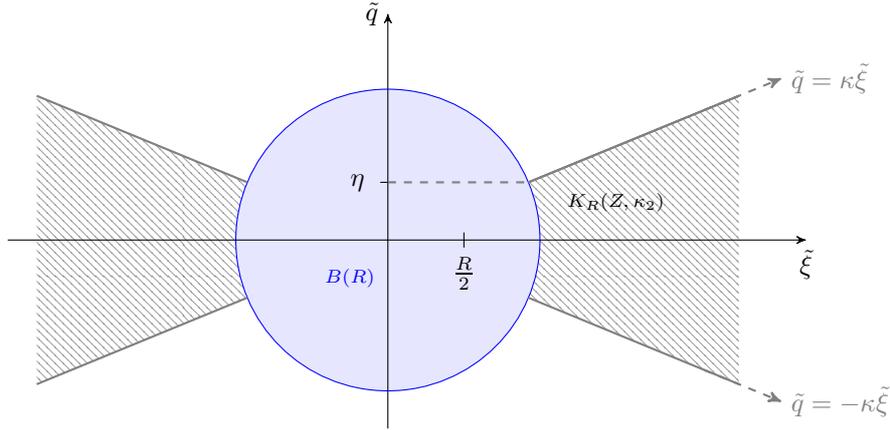
$ $\\
By definition, the function $\varphi$ vanishes on this arc. For large values of $|\tilde q|$ the boundary of $K_R(Z,\kappa)$ is given by the straight lines $|\tilde q|=\kappa|\tilde \xi|$.  {Let $x=(\tilde q,\tilde\xi)$ be a point of intersection of the ball $B(R)$ with the set $\{|\tilde q| =\kappa |\tilde \xi|\}$. Then $|\tilde q|=\left(1+\kappa^{-2}\right)^{-\frac{1}{2}}R=:\eta$}} and we have
 \begin{align}\label{eq: no efimov N=3 integral in xi}
 \begin{split}
 	&\int_{K_R(Z,\kappa)}\!\!\left(|\partial_{\tilde \xi}\varphi|^2-c|\tilde \xi|^{-2-\nu}|\varphi|^2\right)\,\mathrm{d}x \\
 	&\qquad = \int_{\{|\tilde q| < \eta\}}\int_{\left\{|\tilde \xi|\ge  \sqrt{R^2-|\tilde q|^2}\right\}}\!\!\left(|\partial_{\tilde \xi}\varphi|^2-c|\tilde \xi|^{-2-\nu}|\varphi|^2\right)\,\mathrm{d}\tilde \xi \,\mathrm{d}\tilde q\\
 	&\qquad \qquad +  \int_{\{|\tilde q|\ge \eta\}}\int_{\left\{|\tilde \xi|\ge \kappa^{-1}|\tilde q|\right\}}\!\!\left(|\partial_{\tilde \xi}\varphi|^2-c|\tilde \xi|^{-2-\nu}|\varphi|^2\right)\,\mathrm{d}\tilde \xi\,\mathrm{d}\tilde q.
 	\end{split}
 \end{align}
Since $\varphi(\tilde q,\tilde \xi)=0$ for $|\tilde \xi| \le \sqrt{R^2-|\tilde q|^2}$, the one-dimensional Hardy inequality implies
\begin{equation}
\int_{\{|\tilde q| < \eta\}}\int_{\left\{|\tilde \xi|\ge \sqrt{R^2-|\tilde q|^2}\right\}}\!\!\left(|\partial_{\tilde \xi}\varphi|^2-c|\tilde \xi|^{-2-\nu}|\varphi|^2\right)\,\mathrm{d}\tilde \xi\,\mathrm{d}\tilde q\ge 0
\end{equation}
if $R>0$ is large enough. To estimate the second integral on the r.h.s of \eqref{eq: no efimov N=3 integral in xi} we apply Lemma \ref{lem: integration in xi} with $b=\kappa^{-1}|\tilde{q}|$, which yields
\begin{align}\label{eq: no efimov N=3 estimate xi part}
\begin{split}
&\int_{\{|\tilde q|\ge \eta\}}\int_{\left\{|\tilde \xi|\ge \kappa^{-1}|\tilde q|\right\}}\!\!\left(|\partial_{\tilde \xi}\varphi|^2-c|\tilde \xi|^{-2-\nu}|\varphi(x)|^2\right)\,\mathrm{d}\tilde \xi\,\mathrm{d}\tilde q\\
&\qquad\ge- C \int_{\{|\tilde q|\ge\eta\}} |\tilde q|^{-1-\nu} \left(|\varphi(\tilde q,\kappa^{-1}|\tilde q|)|^2+|\varphi(\tilde q,-\kappa^{-1}|\tilde q|)|^2\right)\,\mathrm{d}\tilde q
\end{split}
\end{align}
for some $C>0$.
Combining \eqref{eq: no efimov N=3 estimate q part} and \eqref{eq: no efimov N=3 estimate xi part} with \eqref{eq: quadratic form no Efimov N=3} we get
 \begin{align}\label{eq: quadratic form no Efimov N=3 estimate cone}
 \begin{split}
 	&\int_{K_R(Z,\kappa)}\left(|\nabla_0\varphi|^2+V|\varphi|^2-\varepsilon |x|_m^{-4}|\varphi|^2\right)\,\mathrm{d}x  \\
 	&\qquad \ge -C \int_{\{|\tilde \xi|\ge\frac{R}{2}\}} |\tilde \xi|^{-1-\nu} \left(|\varphi(\kappa|\tilde \xi|,\tilde \xi)|^2+|\varphi(-\kappa|\tilde \xi|,\tilde \xi)|^2\right)\,\mathrm{d}\tilde \xi \\
 	&\qquad\qquad  - C \int_{\{|\tilde q|\ge\eta\}} |\tilde q|^{-1-\nu} \left(|\varphi(\tilde q,\kappa^{-1}|\tilde q|)|^2+|\varphi(\tilde q,-\kappa^{-1}|\tilde q|)|^2\right)\,\mathrm{d}\tilde q.
\end{split} 	
 \end{align}
Note that the integrals on the r.h.s of \eqref{eq: quadratic form no Efimov N=3 estimate cone} are in fact integrals of the function $\varphi$ over the edges of the cone $K(Z,\kappa)$. We introduce polar coordinates $(\rho,\theta)$ in the space $X_0$. Let $\theta_0=\arctan(\kappa)\in\left(0,\frac{\pi}{2}\right)$. At first, we consider the integral over $\tilde\xi$ in \eqref{eq: quadratic form no Efimov N=3 estimate cone} and integrate initially over the set where $\tilde{\xi}>0$. We have
\begin{equation}\label{eq: no Efimov N=3 integral xi}
\int_{\{\tilde \xi \ge\frac{R}{2}\}} \tilde \xi^{-1-\nu} |\varphi(\kappa \tilde \xi,\tilde \xi)|^2\,\mathrm{d}\tilde \xi =c \int_R^\infty \rho^{-1-\nu} |\varphi(\rho,\theta_0)|^2\,\mathrm{d}\rho
\end{equation} 
and 
\begin{equation}
	\int_{\{\tilde \xi \ge\frac{R}{2}\}} \tilde \xi^{-1-\nu} |\varphi(-\kappa \tilde \xi,\tilde \xi)|^2\,\mathrm{d}\tilde \xi =c \int_R^\infty \rho^{-1-\nu} |\varphi(\rho,-\theta_0)|^2\,\mathrm{d}\rho,
\end{equation}
where $c$ is a constant depending on $\theta_0$ only. Together with the analogous integrals for $\tilde\xi<-\frac{R}{2}$ we get
\begin{align}
&\int_{\{|\tilde \xi|\ge\frac{R}{2}\}} |\tilde \xi|^{-1-\nu} \left(|\varphi(\kappa|\tilde \xi|,\tilde \xi)|^2+|\varphi(-\kappa|\tilde \xi|,\tilde \xi)|^2\right)\,\mathrm{d}\tilde \xi\\ \notag
&\quad = c\int_{R}^\infty\rho^{-1-\nu} \left(|\varphi(\rho,\theta_0)|^2 +|\varphi(\rho,\pi-\theta_0)|^2 +|\varphi(\rho,\pi+\theta_0)|^2 +|\varphi(\rho,-\theta_0)|^2 \right) \,\mathrm{d}\rho.
\end{align}
Similarly we can represent the second integral on the r.h.s. of \eqref{eq: quadratic form no Efimov N=3 estimate cone}. Hence, we arrive at
\begin{align}\label{eq: no efimov N=3 eindimensionales integral}
&\int_{K_R(Z,\kappa)}\left(|\nabla_0\varphi|^2+V|\varphi|^2-\varepsilon |x|_m^{-4}|\varphi|^2\right)\,\mathrm{d}x  \\ \notag
&\quad \ge -C\int_{R}^\infty\rho^{-1-\nu} \left(|\varphi(\rho,\theta_0)|^2 +|\varphi(\rho,\pi-\theta_0)|^2 +|\varphi(\rho,\pi+\theta_0)|^2 +|\varphi(\rho,-\theta_0)|^2 \right) \,\mathrm{d}\rho
\end{align}
for some $C>0$. To estimate the r.h.s. of \eqref{eq: no efimov N=3 eindimensionales integral} let us estimate exemplarily the integral
\begin{equation}\label{eq: integral boundary}
\int_R^\infty \rho^{-1-\nu}|\varphi(\rho,\theta_0)|^2 \,\mathrm{d}\rho.
\end{equation}
 We choose $\kappa'>\kappa$ such that $K(Z,\kappa')$ and $K(Z',\kappa')$ do not overlap for any pair of two-cluster partitions $Z\neq Z'$ and denote $\theta_1 = \arctan(\kappa')\in \left(0,\frac{\pi}{2}\right)$. Applying Lemma \ref{lem: Lemma 2 in diemension 1} to the function $\varphi(\rho,\cdot)$ for fixed $\rho$ and with $b_1=\theta_0$ and $b_2=\theta_1$ we get
\begin{equation}\label{eq: no Efimov N=3 applying Lemma 2 d=1}
	|\varphi(\rho,\theta_0)|^2 \le C(\theta_0,\theta_1)\int_{\theta_0}^{\theta_1} \left(|\varphi(\rho,\theta)|^2 +|\partial_\theta \varphi(\rho,\theta)|^2\right)\,\mathrm{d}\theta
\end{equation}
for some $C(\theta_0,\theta_1)>0$. Substituting inequality \eqref{eq: no Efimov N=3 applying Lemma 2 d=1} into
\eqref{eq: integral boundary} we get
\begin{align}\label{eq: no Efimov N=3 estimate integral rho}
\begin{split}
& \int_R^\infty \rho^{-1-\nu}|\varphi(\rho,\theta_0)|^2 \,\mathrm{d}\rho\\
&\qquad\le C(\theta_0,\theta_1)\int_R^\infty\int_{\theta_0}^{\theta_1} \rho^{-1-\nu}\left(|\varphi(\rho,\theta)|^2 +|\partial_\theta \varphi(\rho,\theta)|^2\right)\,\mathrm{d}\theta\,\mathrm{d}\rho\\
&\qquad = C(\theta_0,\theta_1)\int_{\theta_0}^{\theta_1}\int_R^\infty \rho^{-1-\nu}\left(|\varphi(\rho,\theta)|^2 +|\partial_\theta \varphi(\rho,\theta)|^2\right)\,\mathrm{d}\rho\,\mathrm{d}\theta.
\end{split}
\end{align}
Applying inequality \eqref{eq: two dimensional hardy - scala variant} for fixed $\theta$ yields
\begin{equation}
C(\theta_0,\theta_1)\int_R^\infty\rho^{-1-\nu}|\varphi(\rho,\theta)|^2 \,\mathrm{d}\rho\le \varepsilon \int_R^\infty |\partial_\rho \varphi(\rho,\theta)|^2\rho\,\mathrm{d}\rho
\end{equation}
where $\varepsilon>0$ can be chosen arbitrarily small if $R>0$ is large enough. Substituting this inequality into \eqref{eq: no Efimov N=3 estimate integral rho} and using 
\begin{equation}
\left|\frac{\partial \varphi}{\partial \rho}\right|^2 + \frac{1}{\rho^2}\left|\frac{\partial \varphi}{\partial\theta}\right|^2\le |\nabla_0 \varphi|^2
\end{equation}
we obtain
\begin{equation}\label{eq: estimate integral over rho by integral over x d=1 N=3}
 \int_R^\infty \rho^{-1-\nu}|\varphi(\rho,\theta_0)|^2 \,\mathrm{d}\rho\le \varepsilon\int_{K_R(Z,\kappa,\kappa')} |\nabla_0\varphi|^2\,\mathrm{d}x
\end{equation}
for sufficiently large $R>0$. We can estimate the other integrals on the r.h.s. of \eqref{eq: no efimov N=3 eindimensionales integral} by the same arguments.
Therefore, we obtain
\begin{align}\label{eq: no efimov N=3 inequality KR(Z2)}
\begin{split}
&\int_{K_R(Z,\kappa)}\left(|\nabla_0\varphi|^2+V|\varphi|^2-\varepsilon |x|_m^{-4}|\varphi|^2\right)\,\mathrm{d}x \ge - 4\varepsilon\int_{K_R(Z,\kappa,\kappa')} |\nabla_0\varphi|^2\,\mathrm{d}x.
\end{split}
\end{align}
Summing inequality \eqref{eq: no efimov N=3 inequality KR(Z2)} over all partitions $Z$ with $|Z|=2$, inserting the resulting inequality into the definition of $L$ and using that the cones $K_R(Z,\kappa')$ and $K_R(Z',\kappa')$ do not intersect we get
\begin{equation}
    L[\varphi]\ge \int_{K_R^c(\kappa)}\left((1-4\varepsilon)|\nabla_0\varphi|^2+V|\varphi|^2-\varepsilon |x|_m^{-4}|\varphi|^2\right)\,\mathrm{d}x,
\end{equation}
where $K_R^c(\kappa)=X_0\setminus\left(B(R)\cup\bigcup_{Z:|Z|=2}K(Z,\kappa)\right)$. Note that $K_R^c(\kappa)$ is the region in $X_0$, where all particles are far away from each other. Therefore, we can estimate $|V(x)|\le C|x|_m^{-2-\nu}$. Moreover, we can assume $\nu<2$ and therefore $|x|_m^{-4}\le |x|_m^{-2-\nu}$. Hence, we get
\begin{equation}
  L[\varphi]\ge \int_{K_R^c(\kappa)}\left((1-4\varepsilon)|\nabla_0\varphi|^2-(C+\varepsilon) |x|_m^{-2-\nu}|\varphi|^2\right)\,\mathrm{d}x.
\end{equation}
Using polar coordinates $(\rho,\theta)$ and $|\nabla_0\varphi|\ge \left|\frac{\partial\varphi}{\partial \rho}\right|$ we find
\begin{equation}
L[\varphi]\geq \int_{\omega \in I}\int_{R}^\infty\left((1-4\varepsilon) \left|\partial_\rho \varphi \right|^2-(C+\varepsilon)\rho^{-2-\nu}|\varphi|^2 \right)\rho \, \mathrm{d}\rho \mathrm{d}\theta,
\end{equation}
where $I\subset [0,2\pi]$ is the set of angles corresponding to the region $K_R^c(\kappa)$. Now since $\varphi(\rho,\theta)=0$ for $\rho\leq R$, we can apply inequality \eqref{eq: two dimensional hardy - scala variant} to the function $u(\rho)=\varphi(\rho,\theta)$ for fixed $\theta \in I$. Choosing $R$ sufficiently large completes the proof of Theorem \ref{thm: No Efimov N=3}. 
\end{proof}

\subsection{Systems of three two-dimensional particles}
Now we turn to systems of three two-dimensional bosons or three non-identical two-dimensional particles. For systems of three spinless fermions in dimension two there exists the so-called super Efimov effect, see \cite{GridnevSuper}. We prove that for systems without symmetry restrictions such an effect can not occur. Our result is the following
\begin{thm}[Absence of the Efimov effect for systems of three two-dimensional particles]\label{thm: No Efimov n=2 N=3}
Let $H$ be the Hamiltonian corresponding to a system of $N=3$ two-dimensional particles. Assume that $H[C]\ge 0$ for any cluster $C$ with $|C|=2$ and that the pair potentials $V_{ij}$ satisfy \eqref{cond: relatively form bounded} and \eqref{cond: decay at infinity} and they are radially symmetric, i.e. $V_{ij}(x_{ij})=V_{ij}(|x_{ij}|)$. Then the discrete spectrum of $H$ is finite.
\end{thm}
First, we give some auxiliary Lemmas which are analogous to Lemma \ref{lem: Lemma 1 in diemension 1} and Lemma \ref{lem: integration in xi} for $d=2$. 
\begin{lem}\label{lem: lemma for main result 1 ind dim 2}
Let $d=2$ and consider the operator $h=-\Delta+V$ acting in $L^2\left(\mathbb{R}^2\right)$, where $h\ge 0$ and $V(x)=V(|x|)$ satisfies \eqref{cond: relatively form bounded} and \eqref{cond: decay at infinity}. Then there exists a constant $c>0$, such that  for any $b_0>A$ and for any function $\psi \in H^1(\mathbb{R}^2)$ we have
\begin{equation}
J[\psi,b_0]:= \int_{\left\{|x|\leq b_0\right\}}\left( |\nabla \psi(x)|^2+V(x)|\psi(x)|^2 \right)\, \mathrm{d}x \geq - cb_0^{-\nu} \int_0^{2\pi} |\psi(b_0,\theta)|^2\,\mathrm{d}\theta.
\end{equation}
\end{lem}
\begin{rem}
Lemma \ref{lem: lemma for main result 1 ind dim 2} does not hold if we restrict the operator $h$ to anti-symmetric functions. This is the reason why our proof of Theorem \ref{thm: No Efimov n=2 N=3} does not work for a fermionic system, where the super Efimov effect is known to exist. 
\end{rem}
\begin{proof}[Proof of Lemma \ref{lem: lemma for main result 1 ind dim 2}]
Let $\psi \in H^1(\mathbb{R}^2)$ and $b_0>A$. We introduce polar coordinates $x=(\rho,\omega)$ and write $\psi(x)=\sum\limits_{n=-\infty}^\infty \psi_n(x)$ with $\psi_n(x)= R_n(\rho)e^{\mathrm{i} n\omega}$. For $k\in \mathbb{N}, \ k\ge 2,$ let
\begin{equation}
R_n^k(\rho):=
\begin{cases}
R_n(\rho) & \text{if }\rho \leq b_0,
\\ R_n(b_0) \ln \left(k b_0 \rho^{-1}\right) \left(\ln k\right)^{-1} & \text{if } b_0 <\rho \leq kb_0,
\\ 0 &\text{if } \rho >k b_0.
\end{cases}
\end{equation}
We set $\psi_n^k: \mathbb{R}^2 \rightarrow \mathbb{C}, \ \psi_n^k(x)= R_n^k(|x|)\mathrm{e}^{\i n\omega}$. Then we have $J[\psi_n^k,b_0] =J[\psi_n,b_0]$ and therefore
\begin{equation}
    J[\psi,b_0] = \sum\limits_{n=-\infty}^\infty J[\psi_n^k,b_0] \quad \text{for any} \ k\ge 2.
\end{equation}
Now we estimate $J[\psi_n^k,b_0]$ for fixed $k,n\in \mathbb{N}$ with $k\ge 2$.
Due to
\begin{equation}
\vert \nabla\psi\vert^2 =\vert \frac{\partial\psi}{\partial \rho}\vert^2+\frac{1}{\rho^2}\vert \frac{\partial\psi}{\partial \omega}\vert^2\ge\vert \frac{\partial \psi}{\partial \rho}\vert^2
\end{equation}
 and $V(x)=V(|x|)$ we can estimate
\begin{equation}\label{eq: estimate integral of d rho R + V d=2 N=3}
    %\int_{\left\{|x| \le b_0\right\}} \left( |\nabla \psi_n^k(x)|^2+V(x)|\psi_n^k(x)|^2 \right)\, \mathrm{d}x
J[\psi_n^k,b_0]
    \ge 2\pi \int\limits_{0}^{b_0} \left(|\partial_\rho R_n^{k}(\rho)|^2 +V(\rho)|R_n^{k}(\rho)|^2\right)\,\rho\mathrm{d} \rho.
\end{equation}
Using $\langle h \tilde{\psi}_n^k,\tilde{\psi}_n^k\rangle \ge 0$ for the radial function $\tilde{\psi}_n^k(x)=R_n^k(|x|)$ and \eqref{eq: estimate integral of d rho R + V d=2 N=3} yields
\begin{equation}\label{eq: estimate Jpsink}
    J[\psi_n^k ,b_0]\ge -2\pi\int\limits_{b_0}^\infty \left(|\partial_\rho R_n^{k}(\rho)|^2 +V(\rho)|R_n^{k}(\rho)|^2\right)\,\rho\mathrm{d} \rho.
\end{equation}
%
%
%\begin{align}\label{eq: estimate integral h psink d=2 N=3}
%\begin{split}
%    J[\psi_n^k,b_0]&=\langle h\psi_n^k,\psi_n^k\rangle-\int_{\left\{|x|> b_0\right\}}\left( |\nabla \psi_n^k(x)|^2+V(x)|\psi_n^k(x)|^2 \right)\, \mathrm{d}x\\
%    &\ge -\int_{\left\{|x|> b_0\right\}}\left( |\nabla \psi_n^k(x)|^2+V(x)|\psi_n^k(x)|^2 \right)\, \mathrm{d}x.
%\end{split}
%\end{align}
%
Easy computation shows that
\begin{equation}
    \partial_\rho R_n^{k}(\rho)=\begin{cases}
    -R_n(b_0)\left(\ln(k)\right)^{-1}\rho^{-1} & \text{if } b_0<\rho<kb_0,\\
    0 &\text{if } \rho >kb_0.
    \end{cases}
\end{equation}
This implies
\begin{equation}\label{eq: estimate integral of d rho R d=2 N=3}
     \int\limits_{b_0}^\infty |\partial_\rho R_n^{k}(\rho) |^2 \,\rho\mathrm{d} \rho \le  |R_n(b_0)|^2(\ln(k))^{-2}\int_{b_0}^{kb_0}\rho^{-1}\,\mathrm{d}\rho =  |R_n(b_0)|^2(\ln(k))^{-1}.
\end{equation}
 Since $\vert V(\rho)\vert \le C\left(1+\rho\right)^{-2-\nu}$ for $\rho>b_0$, we get
\begin{align}\label{eq: estimate vrnk}
\notag
\int\limits_{b_0}^\infty| V(\rho)||R_n^k(\rho)|^2 \, \rho\mathrm{d}\rho&\le C |R_n(b_0)|^2(\ln(k))^{-2} \int\limits_{b_0}^{kb_0}(1+\rho)^{-2-\nu} \left(\ln(kb_0\rho^{-1})\right)^{2} \, \rho\mathrm{d}\rho\\
& \le C|R_n(b_0)|^2\int\limits_{b_0}^{kb_0} (1+\rho)^{-2-\nu}\,\rho\mathrm{d}\rho,
\end{align}
where for the last inequality we used that $(\ln(k))^{-2}\left(\ln(kb_0\rho^{-1})\right)^{2}\le 1$ for $\rho\in (b_0,kb_0)$. By inserting
\begin{equation}
    \int\limits_{b_0}^{kb_0} (1+\rho)^{-2-\nu}\,\rho\mathrm{d}\rho \le \int\limits_{b_0}^{\infty} \rho^{-1-\nu}\,\mathrm{d}\rho = b_0^{-\nu}
\end{equation}
in inequality \eqref{eq: estimate vrnk} we find
\begin{equation}\label{eq: estimate integral of V Lemma d=2 N=3}
\int\limits_{b_0}^\infty |V(\rho)||R_n^k(\rho)|^2 \, \rho\mathrm{d}\rho  \le C|R_n(b_0)|^2b_0^{-\nu}.
\end{equation}
Combining \eqref{eq: estimate Jpsink} with  \eqref{eq: estimate integral of d rho R d=2 N=3} and \eqref{eq: estimate integral of V Lemma d=2 N=3}  we obtain
\begin{equation}\label{eq: Estimate Jpsink}
    J[\psi_n^k,b_0]\ge -2\pi C|R_n(b_0)|^2b_0^{-\nu}-2\pi |R_n(b_0)| (\ln(k))^{-1}.
\end{equation}
Recall that the left hand side of \eqref{eq: Estimate Jpsink} coincides with $J[\psi_n,b_0]$ and in particular does not depend on $k$. Therefore, sending $k$ to infinity and using
\begin{equation}
2\pi \sum_{n=-\infty}^\infty \left|R_n(b_0)\right|^2 = \int_0^{2\pi} |\psi(b_0,\omega)|^2\, \mathrm{d}\omega
\end{equation}
completes the proof of Lemma \ref{lem: lemma for main result 1 ind dim 2}.
\end{proof}
The following lemma is analogous to Lemma \ref{lem: integration in xi}
\begin{lem}\label{lem: integration in xi d=2}
Let $C_0>0$. Then for any sufficiently large $b>0$ and for any $\psi\in H^1(\mathbb{R}^2)$ we have
\begin{equation}\label{eq: no efimov Hardy-type n=2}
\int_{\{|x|\ge b\}}\left(|\nabla\psi(x)|^2-C_0 |x|^{-2-\nu}|\psi(x)|^2\right)\,\mathrm{d}x\ge -{2C_0b^{-\nu}} (2\pi)^{-1}\int_{0}^{2\pi}|\psi(b,\theta)|^2\,\mathrm{d}\theta.
\end{equation}
\end{lem}
\begin{proof}[Proof of Lemma \ref{lem: integration in xi d=2}] Let $\psi\in H^1(\mathbb{R}^2)$. We write $\psi=\psi_0+\psi_1$ with $\psi_0=\mathcal{P}^{m=0}\psi$ and $\psi_1=\mathcal{P}^{|m|\ge 1}\psi$, where $P^{m}$ is the projection onto the space of functions with angular momentum $m$. Then for $\psi_1$ we have 
\begin{equation}
	|\nabla \psi_1|^2 =  |\partial_\rho \psi_1|^2 +\frac{1}{\rho^2}|\partial_\theta \psi_1|^2  \ge  \frac{1}{\rho^2}|\psi_1|^2
\end{equation}
 and therefore 
\begin{equation}
\int_{\{|x|\ge b\}}\left(|\nabla\psi_1(x)|^2-C_0 |x|^{-2-\nu}|\psi_1(x)|^2\right)\,\mathrm{d}x\ge 0
\end{equation}
if $b>0$ is sufficiently large. Hence, it suffices to prove inequality \eqref{eq: no efimov Hardy-type n=2} for the spherically symmetric function $\psi_0$. For $|x|\ge b$ let $\tilde{\psi}(|x|)=\psi_0(|x|)-\psi_0(b)$, such that $\tilde{\psi}(b)=0$ and we extend $\tilde{\psi}$ with zero to the region $\{|x|<b\}$. Then, similarly to the one-dimensional case we obtain
\begin{align}\label{eq: integral decomposition radial part kuenstlich null}
\begin{split}
	&\int_{\{|x|\ge b\}}\left(|\nabla\psi_0(x)|^2-C_0 |x|^{-2-\nu}|\psi_0(x)|^2\right)\,\mathrm{d}x \\
	&\ge \int_{\{|x|\ge b\}}\!\left(|\nabla\tilde \psi(x)|^2-2C_0 |x|^{-2-\nu}|\tilde\psi(x)|^2\right)\,\mathrm{d}x - \int_{\{|x|\ge b\}} 2C_0 |x|^{-2-\nu}|\psi_0(b)|^2\,\mathrm{d}x.
\end{split}
\end{align}
Since $\tilde{\psi}(|x|)=0$ for $|x|\le b$, we can apply the two-dimensional Hardy inequality to the function $\tilde{\psi}$, which implies that the first integral on the r.h.s of \eqref{eq: integral decomposition radial part kuenstlich null} is non-negative.
Hence, we arrive at
\begin{equation}\label{eq: inequality integration in xi d=2}
	\int_{\{|x|\ge b\}}\left(|\nabla\psi_0(x)|^2-C_0 |x|^{-2-\nu}|\psi_0(x)|^2\right)\,\mathrm{d}x \ge  - 2C_0 \int_{\{|x|\ge b\}} |x|^{-2-\nu}|\psi_0(b)|^2\,\mathrm{d}x.
\end{equation}
Computing the integral on the r.h.s. of \eqref{eq: inequality integration in xi d=2} completes the proof of Lemma \ref{lem: integration in xi d=2}.
\end{proof}
\begin{proof}[Proof of Theorem \ref{thm: No Efimov n=2 N=3}]
In the proof we follow the same strategy as in the proof of Theorem \ref{thm: No Efimov N=3}. Let
\begin{equation}
    L[\varphi]:= \int\left(|\nabla_0\varphi|^2+V|\varphi|^2-\varepsilon |x|_m^{-4}|\varphi|^2\right)\,\mathrm{d}x.
\end{equation}
We show that $L[\varphi]\ge 0$ for all functions $\varphi\in H^1(X_0)$ with $\supp(\varphi)\subset \{\vert x\vert_m\ge R\}$ if $\varepsilon>0$ is small enough and $R>0$ is sufficiently large. First, we estimate the part of $L[\varphi]$ corresponding to the cone $K(Z,\kappa)$ for an arbitrary partition $Z$ into two clusters. Assume that $Z=(C_1,C_2)$ with $C_1=\{i,j\}$ and $C_2=\{k\}$. Note that the spaces $X_0(Z)$ and $X_c(Z)$ are both two-dimensional. We choose orthonormal bases of $X_0(Z)$ and $X_c(Z)$ and denote by $\tilde{q}_1,\, \tilde{q}_2, \, \tilde{\xi}_1, \, \tilde{\xi}_2$ the corresponding coordinates.
We write $\tilde{q}=(\tilde{q}_1, \tilde{q}_2)$, $\tilde{\xi}=(\tilde{\xi}_1,\tilde{\xi}_2)$ and $\varphi=\varphi(\tilde{q},\tilde{\xi})$.
Similarly to \eqref{eq: quadratic form no Efimov N=3} we write
 \begin{align}\label{eq: quadratic form no Efimov n=2 N=3}
 \begin{split}
 	&\int_{K_R(Z,\kappa)}\left(|\nabla_0\varphi|^2+V|\varphi|^2-\varepsilon |x|_m^{-4}|\varphi|^2\right)\,\mathrm{d}x  =\int_{K_R(Z,\kappa)}\left(|\nabla_{\tilde{q}}\varphi|^2+V_{ij}|\varphi|^2\right)\,\mathrm{d}x \\
 	&\qquad\qquad+\int_{K_R(Z,\kappa)}\left(|\nabla_{\tilde{\xi}}\varphi|^2+(V_{ik}+V_{jk})|\varphi|^2-\varepsilon|x|_m^{-4}|\varphi|^2\right)\,\mathrm{d}x.
 	\end{split}
 \end{align}
To estimate the integrals on the r.h.s of \eqref{eq: quadratic form no Efimov n=2 N=3} we introduce polar coordinates $\tilde q=(\rho_1,\beta_1)$ and $\tilde \xi=(\rho_2,\beta_2)$ in the planar spaces $X_0(Z)$ and $X_c(Z)$. For the first integral on the r.h.s. of \eqref{eq: quadratic form no Efimov n=2 N=3} we use Lemma \ref{lem: lemma for main result 1 ind dim 2} for fixed $\tilde \xi$ with $b_0=\kappa |\tilde \xi|$.
Then similarly to \eqref{eq: no efimov N=3 estimate q part} we get
\begin{align}\label{eq: no efimov n=2 N=3 estimate q}
 \begin{split}
 	\int_{K_R(Z,\kappa)}\left(|\nabla_{\tilde{q}}\varphi|^2+V_{ij}|\varphi|^2\right)\,\mathrm{d}x&= \int_{\{|\tilde{\xi}| \ge \frac{R}{2}\}}\int_{\{|\tilde{q}|\le \kappa|\tilde{\xi}|\}}\left(|\nabla_{\tilde{q}}\varphi|^2+V_{ij}|\varphi|^2\right)\,\mathrm{d}x\\
 	& \ge -C \int_{\{|\tilde{\xi}|\ge\frac{R}{2}\}} |\tilde{\xi}|^{-\nu} \int_0^{2\pi} |\varphi(\kappa|\tilde{\xi}|,\beta_1,\tilde \xi)|^2\,\mathrm{d}\beta_1\,\mathrm{d}\tilde \xi
 	\end{split}
 \end{align}
for some $C>0$. For the second integral on the r.h.s. of \eqref{eq: quadratic form no Efimov n=2 N=3} we use Lemma \ref{lem: integration in xi d=2} for fixed $\tilde q$ and with $b=\kappa^{-1}|\tilde{q}|$, which similarly to \eqref{eq: no efimov N=3 estimate xi part} yields
\begin{align}\label{eq: no efimov n=2 N=3 estimate xi}
\begin{split}
&\int_{K_R(Z,\kappa)}\left(|\nabla_{\tilde{\xi}}\varphi|^2+(V_{ik}+V_{jk})|\varphi|^2-\varepsilon|x|_m^{-4}|\varphi|^2\right)\,\mathrm{d}x\\
&\qquad\qquad  \qquad\qquad\qquad\ge- C \int_{\{|\tilde{q}|\ge\eta\}} |\tilde{q}|^{-\nu}\int_{0}^{2\pi} |\varphi(\tilde q,\kappa^{-1}|\tilde{q}|,\beta_2)|^2\,\mathrm{d}\beta_2\,\mathrm{d}\tilde q,
\end{split}
\end{align}
where $\eta=\left(1+\kappa^{-2}\right)^{-1}R$ is analogous to the proof of Theorem \ref{thm: No Efimov N=3}. Combining \eqref{eq: no efimov n=2 N=3 estimate q} and \eqref{eq: no efimov n=2 N=3 estimate xi} with \eqref{eq: quadratic form no Efimov n=2 N=3} implies
\begin{align}\label{eq: no efimov n=2 N=3 estimate q and xi combined}
	\begin{split}
	&\int_{K_R(Z,\kappa)}\left(|\nabla_0\varphi|^2+V|\varphi|^2-\varepsilon |x|_m^{-4}|\varphi|^2\right)\,\mathrm{d}x \\
	& \qquad \ge  -C \int_{\{|\tilde{\xi}|\ge\frac{R}{2}\}} |\tilde{\xi}|^{-\nu} \int_0^{2\pi} |\varphi(\kappa|\tilde{\xi}|,\beta_1,\tilde \xi)|^2\,\mathrm{d}\beta_1\,\mathrm{d}\tilde \xi\\
	& \qquad \qquad \qquad - C\int_{\{|\tilde{q}|\ge\eta\}} |\tilde{q}|^{-\nu}\int_{0}^{2\pi} |\varphi(\tilde q,\kappa^{-1}|\tilde{q}|,\beta_2)|^2\,\mathrm{d}\beta_2\,\mathrm{d}\tilde q.
	\end{split}
\end{align}
In the set $\{(|\tilde{q}|,|\tilde{\xi}|)\in \mathbb{R}_+\times \mathbb{R}_+\}$ we introduce the polar coordinates $(\rho,\theta)$, where $\rho^2 = |\tilde{q}|^2+|\tilde{\xi}|^2=|x|_m^2$ and $\theta= \arctan\left(\frac{|\tilde{q}|}{|\tilde{\xi}|}\right)\in [0,\frac{\pi}{2})$. Then $\rho_1=\rho \sin(\theta)$ and $\rho_2=\rho\cos(\theta).$ We represent the function $\varphi(x)$ as a function $\tilde{\varphi}(\rho,\theta,\beta_1,\beta_2)$.  Note that the integrals on the r.h.s of \eqref{eq: no efimov n=2 N=3 estimate q and xi combined} are integrals of the function $|\varphi(x)|^2$ over the set where $|\tilde{q}|=\kappa|\tilde{\xi}|$, i.e. where $\theta_0=\arctan(\kappa)$. Therefore, for the first integral on the r.h.s of \eqref{eq: no efimov n=2 N=3 estimate q and xi combined} we get
\begin{align}\label{eq: no Efimov integral rand xi}
\begin{split}
	&\int_{\{|\tilde{\xi}|\ge\frac{R}{2}\}} |\tilde{\xi}|^{-\nu} \int_0^{2\pi} |\varphi(\kappa|\tilde{\xi}|,\beta_1,\tilde \xi)|^2\,\mathrm{d}\beta_1\,\mathrm{d}\tilde \xi\\
	&=\int_{{\frac{R}{2}}}^\infty\int_0^{2\pi}\int_0^{2\pi}\rho_2^{-\nu} |\varphi(\kappa \rho_2,\beta_1,\rho_2,\beta_2)|^2\,\mathrm{d}\beta_1\,\mathrm{d}\beta_2\, \rho_2\mathrm{d}\rho_2,
	 \\&=c \int_R^\infty \int_0^{2\pi}\int_0^{2\pi}\rho^{-\nu} |\tilde{\varphi}(\rho,\theta_0,\beta_1,\beta_2)|^2\,\mathrm{d}\beta_1\,\mathrm{d}\beta_2\, \rho\mathrm{d}\rho,
\end{split}
\end{align}
where $c>0$ is a constant which comes from the transformation of variables if we represent the function $\rho_2\mapsto \varphi(\kappa\rho_2,\beta_1,\rho_2,\beta_2)$ as function $\rho\mapsto\tilde{\varphi}(\rho,\theta_0,\beta_1,\beta_2)$, where $\theta_0=\arctan(\kappa)$. In the first equality in \eqref{eq: no Efimov integral rand xi} we used that $\dim \left(X_c(Z)\right)=2$, which implies that the Jacobian of the transformation to polar coordinates in $X_c (Z)$ gives a factor $\rho_2$. In the last equality of \eqref{eq: no Efimov integral rand xi} we used that the function $\tilde\varphi$ is zero for $\rho<R$. Similarly we get
\begin{align}\label{eq: no Efimov integral rand q}
\begin{split}
&\int_{\{|\tilde{q}|\ge\eta\}} |\tilde{q}|^{-\nu}\int_{0}^{2\pi} |\varphi(\tilde q,\kappa^{-1}|\tilde{q}|,\beta_2)|^2\,\mathrm{d}\beta_2\,\mathrm{d}\tilde q\\
&=c'\int_R^\infty \int_0^{2\pi}\int_0^{2\pi}\rho^{-\nu} |\tilde{\varphi}(\rho,\theta_0,\beta_1,\beta_2)|^2\,\mathrm{d}\beta_1\,\mathrm{d}\beta_2\, \rho\mathrm{d}\rho
\end{split}
\end{align}
for some $c'>0$ . Therefore, by combining \eqref{eq: no Efimov integral rand xi} and \eqref{eq: no Efimov integral rand q} with \eqref{eq: no efimov n=2 N=3 estimate q and xi combined} we obtain 
\begin{align}\label{eq: estimate integral inside cone by scalar integral}
\begin{split}
&\int_{K_R(Z,\kappa)}\left(|\nabla_0\varphi|^2+V|\varphi|^2-\varepsilon |x|_m^{-4}|\varphi|^2\right)\,\mathrm{d}x\\
&\qquad \ge -C\int_R^\infty \int_0^{2\pi}\int_0^{2\pi}\rho^{1-\nu} |\tilde{\varphi}(\rho,\theta_0,\beta_1,\beta_2)|^2\,\mathrm{d}\beta_1\,\mathrm{d}\beta_2\, \mathrm{d}\rho
\end{split}
\end{align}
for some $C>0$. 
Now as in the proof of Theorem \ref{thm: No Efimov N=3} we estimate the integral on the r.h.s. of \eqref{eq: estimate integral inside cone by scalar integral}, which is an integral over the edge of $K(Z,\kappa)$ given by $\{|\tilde{q}|=\kappa|\tilde{\xi}|\}$, by an integral over the set $K(Z,\kappa,\kappa')$ for some $\kappa'$ which is slightly larger than $\kappa$. For this purpose let $\kappa'>\kappa$ be so small that the cones $K_R(Z,\kappa')$ and $K_R(Z',\kappa')$ do not overlap for partitions $Z\neq Z'$ with $|Z|=|Z'|=2$ and let $\theta_1=\arctan(\kappa')$. We apply Lemma \ref{lem: Lemma 2 in diemension 1} to the function $\varphi(\rho,\cdot, \theta_1,\theta_2)$ for fixed $\rho,\theta_1,\theta_2$ and with $b_1=\theta_0, \ b_2=\theta_1$. Then we get
\begin{align}\label{eq: no efimov n=2 Lemma winkel}
&\int_R^\infty \int_0^{2\pi}\int_0^{2\pi}\rho^{1-\nu} |\tilde{\varphi}(\rho,\theta_0,\beta_1,\beta_2)|^2\,\mathrm{d}\beta_1\,\mathrm{d}\beta_2\, \mathrm{d}\rho\\ \notag
&\le C(\theta_0,\theta_1) \int_R^\infty \!\int_0^{2\pi}\!\int_0^{2\pi}\!\int_{\theta_0}^{\theta_1}\rho^{1-\nu}\! \left(|\tilde{\varphi}(\rho,\theta,\beta_1,\beta_2)|^2+|\partial_\theta\tilde{\varphi}(\rho,\theta,\beta_1,\beta_2)|^2\right)\,\mathrm{d}\theta\,\mathrm{d}\beta_1\,\mathrm{d}\beta_2\, \mathrm{d}\rho,
\end{align}
where $C(\theta_0,\theta_1)$ depends on $\theta_0$ and $\theta_1$ only.
Using the scalar form of the four-dimensional Hardy inequality \cite[eq. (2.15)]{birman} we obtain
\begin{equation}
	\int_R^\infty\rho^{1-\nu} |\tilde{\varphi}(\rho,\theta,\beta_1,\beta_2)|^2\,\mathrm{d}\rho \le  \int_R^\infty\rho^{3-\nu} |\partial_\rho\tilde{\varphi}(\rho,\theta,\beta_1,\beta_2)|^2\,\mathrm{d}\rho. 
\end{equation}
Therefore, we get
\begin{align}\label{eq: ungleichung winkel und radiale komponente}
\begin{split}
&\int_R^\infty \rho^{1-\nu} \left(|\tilde{\varphi}(\rho,\theta,\beta_1,\beta_2)|^2+|\partial_\theta\tilde{\varphi}(\rho,\theta,\beta_1,\beta_2)|^2\right)\, \mathrm{d}\rho\\
&\qquad \le R^{-\nu}\int_R^\infty \rho^{3} \left(|\partial_\rho\tilde{\varphi}(\rho,\theta,\beta_1,\beta_2)|^2+\frac{|\partial_\theta\tilde{\varphi}(\rho,\theta,\beta_1,\beta_2)|^2}{\rho^2}\right)\, \mathrm{d}\rho.
\end{split}
\end{align}
Recall that $(\rho,\theta)$ are the polar coordinates corresponding to $(|\tilde{q}|,|\tilde{\xi}|)$, which implies
\begin{equation}
	\left(|\partial_\rho\tilde{\varphi}(\rho,\theta,\beta_1,\beta_2)|^2+\frac{|\partial_\theta\tilde{\varphi}(\rho,\theta,\beta_1,\beta_2)|^2}{\rho^2}\right) = |\partial_{|\tilde{q}|}\varphi|^2 + |\partial_{|\tilde{\xi}|}\varphi|^2\le |\nabla_0\varphi|^2.
\end{equation}
This yields
\begin{align}\label{eq: no efimov d=2 N=3 estimate integral < eps gradient}
& R^{-\nu}\int_R^\infty\! \int_0^{2\pi}\!\int_0^{2\pi}\!\int_{\theta_0}^{\theta_1}\rho^{3} \left(|\partial_\rho\tilde{\varphi}(\rho,\theta,\beta_1,\beta_2)|^2+\frac{|\partial_\theta\tilde{\varphi}(\rho,\theta,\beta_1,\beta_2)|^2}{\rho^2}\right)\,\mathrm{d}\theta\,\mathrm{d}\beta_1\,\mathrm{d}\beta_2\, \mathrm{d}\rho\notag\\
& \le \varepsilon\int_{K_R(Z,\kappa,\kappa')} |\nabla_0\varphi|^2\,\mathrm{d}x,
\end{align}
where $\varepsilon>0$ can be chosen arbitrarily small if $R>0$ is sufficiently large. Here we used that the Jacobian of the transformation from the coordinates $x=(\tilde{q}_1,\tilde{q}_2,\tilde{\xi}_1,\tilde{\xi}_2)$ to the variables $(\rho,\theta,\beta_1,\beta_2)$ is given by $\rho^3\sin(\theta)\cos(\theta)$ and we can estimate 
\begin{equation}
0<\sin(\theta_0)\cos(\theta_0)\le \sin(\theta)\cos(\theta)
\end{equation}
 for any $\theta\in (\theta_0,\theta_1)$ if $0<\kappa<\kappa'<1$. Combining \eqref{eq: no efimov d=2 N=3 estimate integral < eps gradient} with \eqref{eq: ungleichung winkel und radiale komponente}, \eqref{eq: no efimov n=2 Lemma winkel} and \eqref{eq: estimate integral inside cone by scalar integral} we get 
\begin{align}
\begin{split}
&\int_{K_R(Z,\kappa)}\left(|\nabla_0\varphi|^2+V|\varphi|^2-\varepsilon |x|_m^{-4}|\varphi|^2\right)\,\mathrm{d}x \ge - \varepsilon\int_{K_R(Z,\kappa,\kappa')} |\nabla_0\varphi|^2\,\mathrm{d}x.
\end{split}
\end{align}
This inequality is an analogue to \eqref{eq: estimate integral over rho by integral over x d=1 N=3} in the proof of Theorem \ref{thm: No Efimov N=3}. Now we can complete the proof of Theorem \ref{thm: No Efimov n=2 N=3} by repeating the same steps as in the proof of Theorem \ref{thm: No Efimov N=3} if we replace the scalar form of the two-dimensional Hardy inequality by the scalar form of the four-dimensional one. 
\end{proof}

\section*{Acknowledgements}
The authors are deeply grateful to Timo Weidl for his support, in particular for providing them an insight into the unpublished manuscript \cite{Weidlbook}. The work of Simon Barth and Andreas Bitter was supported by the Deutsche Forschungsgemeinschaft (DFG) through the Research Training Group 1838: Spectral Theory and Quantum Systems. Semjon Vugalter gratefully acknowledges the funding by the Deutsche Forschungsgemeinschaft (DFG), German Research Foundation Project ID 258734477 - SFB1173. Semjon Vugalter is grateful to the University of Toulon for the hospitality during his stay there. The authors thank the Mittag-Leffler Institute, where a part of the work was done during the semester program \textit{Spectral Methods in Mathematical Physics.}

\appendix
\section{Properties of the space $\tilde{H}^1(\mathbb{R}^d)$}\label{Appendix A}
Here we collect some properties of the space $\tilde{H}^1(\mathbb{R}^d)$ for dimensions $d=1$ and $d=2$. These spaces were introduced M. Birman in \cite[Section 2]{birman} and are intensively discussed in the book \cite{Weidlbook} by R. Frank, A. Laptev and T. Weidl. For convenience we give the statements and some of these properties below.
 \begin{prop}[Properties of $\tilde{H}^1(\mathbb{R}^d), \ d=1,2$]\label{prop: Homogenous Sobolev}  The following assertions hold.
\begin{enumerate}
\item[\textbf{(i)}](Hardy's inequality for the half-line, inequality $(2.17)$ in \cite{birman}) Let $d=1$ and $u\in \tilde{H}^1(\mathbb{R})$, such that $\liminf_{t\rightarrow 0} |u(t)|=0$. Then
\begin{equation}\label{eq: one dimensional Hardy Birman}
	\int\limits_0^\infty \frac{|u(t)|^2}{t^2}\,\mathrm{d}t \le 4 \int\limits_0^\infty |u'(t)|^2\,\mathrm{d}t.
\end{equation}
\item[\textbf{(ii)}](Two-dimensional Hardy inequality, inequality $(2)$ in \cite{Solomyak}) Let $d=2$ and assume that $u\in \tilde{H}^1(\mathbb{R}^2)$, represented in polar coordinates $(r,\theta)$, satisfies
\begin{equation}\label{int in trace sense}
	\int_{0}^{2\pi} u(1,\theta)\,\mathrm{d}\theta = 0,
\end{equation}
where $u(1,\theta)$ is understood in the trace sense. Then
\begin{equation}\label{eq: two dimensional Hardy Solomyak}
	\int_{\mathbb{R}^2} \frac{|u|^2}{|x|^2(1+\ln^2(|x|))} \,\mathrm{d}x \le 4 \int_{\mathbb{R}^2} |\nabla u|^2 \,\mathrm{d}x.
\end{equation}
	\item[\textbf{(iii)}] Let $d=1$. Then there exists a constant $C>0$, such that for all functions $u\in \tilde{H}^1(\mathbb{R})$ 
\begin{equation}\label{eq: weidl inequality in 1}
\int_{-\infty}^\infty \frac{|u(x)|^2}{1+x^2}\, \mathrm{d}x \leq C \Vert u \Vert_{\tilde{H}^1}^2.
\end{equation}
\item[\textbf{(iv)}] Let $d=2$. Then there exists a constant $C>0$, such that for all functions $u\in \tilde{H}^1(\mathbb{R}^2)$ 
\begin{equation}\label{eq: weidl inequality in 2}
\int_{\mathbb{R}^2} \frac{|u(x)|^2}{1+x^2(\ln|x|)^2}\, \mathrm{d}x \leq C \Vert u \Vert_{\tilde{H}^1}^2.
\end{equation}
\item[\textbf{(v)}] Let $u\in \tilde{H}^1(\mathbb{R}^d)$ and let $(u_n)_{n\in \mathbb{N}}$ be a sequence in $\tilde{H}^1(\mathbb{R}^d)$, such that $u_n \rightharpoonup u$ weakly in $\tilde{H}^1(\mathbb{R}^d)$. Then for every measurable bounded set $B\subset \mathbb{R}^d$ we have $\chi_B u_n \rightarrow \chi_B u$ in $L^2(\mathbb{R}^d).$
\end{enumerate}
\end{prop}
\begin{proof}
\textbf{(i)} We borrow the proof of this inequality from the book \cite{Weidlbook}, which is currently in preparation. We are grateful to R. Frank, A. Laptev and T. Weidl for sharing it with us before publishing.
\\
By applying twice the product rule for weakly differentiable functions we get
\begin{align}
\begin{split}
	\left| u'-\frac{1}{2x}u\right|^2&=|u'|^2+\frac{1}{4x^2}|u|^2-\frac{1}{2x}\left(|u|^2\right)^{'}=|u'|^2-\frac{1}{4x^2}|u|^2-\left(\frac{1}{2x}|u|^2\right)^{'}.
	\end{split}
\end{align}
Hence, for fixed $0<\varepsilon<M<\infty$ we have
\begin{align}\label{eq: Start Hardy}
\begin{split}
0\le \int_\varepsilon^M  \left| u'-\frac{1}{2x}u\right|^2\,\mathrm{d}x= &\int_\varepsilon^M |u'|^2 \,\mathrm{d}x-\int_\varepsilon^M \frac{1}{4x^2}|u|^2\,\mathrm{d}x\\
&-\frac{1}{2M}\left(|u(M)|^2\right)+\frac{1}{2\varepsilon}\left(|u(\varepsilon)|^2\right).
\end{split}
\end{align}
Note that 
\begin{equation}\label{eq: ineq u by int u'}
	|u(x)|^2\le x\int_0^x|u'(y)|^2\,\mathrm{d}y.
\end{equation}
Indeed, we have
\begin{align}
\begin{split}
|u(x)-u(x')|&=\left|\int_{x'}^x u'(y)\,\mathrm{d}y\right| \\ &\le \left(\int_{x'}^x |u'(y)|^2\,\mathrm{d}y\right)^{\frac{1}{2}} \left(\int_{x'}^x 1\,\mathrm{d}y\right)^{\frac{1}{2}}\le \left(\int_{0}^x |u'(y)|^2\,\mathrm{d}y\right)^{\frac{1}{2}}x^{\frac{1}{2}}.
\end{split}
\end{align}
Now \eqref{eq: ineq u by int u'} follows from the assumption that $\liminf_{x'\rightarrow 0} |u(x')|=0.$\\
 Now we take inequality \eqref{eq: Start Hardy} and let $M\rightarrow \infty$. Note that the first integral on the r.h.s. of \eqref{eq: Start Hardy} converges as $M\rightarrow \infty$, because $u'\in L^2(0,\infty).$ By \eqref{eq: ineq u by int u'} we get $\sup_M |u(M)|M^{-1}<\infty$. Therefore,
\begin{equation}
	\int_\varepsilon^\infty |u|^2x^{-2}\,\mathrm{d}x<\infty.
\end{equation}
The finiteness of this integral and the fact that $x\mapsto x^{-1}$ is not integrable implies
\begin{equation}
	\liminf_{x\rightarrow \infty} |u(x)|^2x^{-1} =0.
\end{equation}
Choosing $M$ in \eqref{eq: Start Hardy} along a sequence where this $\liminf$ is realized we obtain 
\begin{align}\label{eq: Start Hardy 2}
\begin{split}
0\le &\int_\varepsilon^\infty |u'|^2 \,\mathrm{d}x-\int_\varepsilon^\infty \frac{1}{4x^2}|u|^2\,\mathrm{d}x+\frac{1}{2\varepsilon}\left(|u(\varepsilon)|^2\right).
\end{split}
\end{align}
Now we let $\varepsilon\rightarrow 0$. Since $u'\in L^2(0,\infty)$, the first integral on the r.h.s. of \eqref{eq: Start Hardy 2} converges. By the same argument, together with \eqref{eq: ineq u by int u'} we get $\lim_{\varepsilon\rightarrow 0} |u(\varepsilon)|^2 \varepsilon^{-1}=0.$ This shows that the second integral on the r.h.s. of \eqref{eq: Start Hardy 2} also converges and we have
\begin{equation}
	0\le \int_0^\infty |u'|^2 \,\mathrm{d}x-\int_0^\infty \frac{1}{4x^2}|u|^2\,\mathrm{d}x.
\end{equation}
\textbf{(ii)} The proof can be found in \cite{Solomyak}.
\\ \textbf{(iii)} We take a smooth function $\xi$ with $0\le \xi\le 1$ on $\mathbb{R}_+$, $\xi=0$ on $(0,1/2]$ and $\xi=1$ on $[1,\infty)$. Then
\begin{equation}\label{eq: start hardy local}
	\int_0^\infty \frac{|u|^2}{1+x^2}\,\mathrm{d}x \le 2\int_0^\infty (1-\xi)^2\frac{|u|^2}{1+x^2}\,\mathrm{d}x + 2 \int_0^\infty \xi^2 \frac{|u|^2}{1+x^2}\,\mathrm{d}x. 
\end{equation}
The first term on the r.h.s. of \eqref{eq: start hardy local} can be controlled by 
\begin{equation}
2\int_0^\infty (1-\xi)^2\frac{|u|^2}{1+x^2}\,\mathrm{d}x \le 2\sup\limits_{0\le s\le 1} \frac{1}{1+s^2}\int_0^1 |u|^2\,\mathrm{d}x=2\int_0^1 |u|^2\,\mathrm{d}x.
\end{equation}
The second term on the r.h.s. of \eqref{eq: start hardy local} can be estimated by using the Hardy inequality as
\begin{align}
\begin{split}
	2 \int_0^\infty \xi^2 \frac{|u|^2}{1+x^2}\,\mathrm{d}x &\le 8 \int_0^\infty |(\xi u)'|^2\,\mathrm{d}x \le 16 \int_0^\infty \left(\xi^2|u'|^2+\left(\xi'\right)^2|u|^2\right)\,\mathrm{d}x
\\	&\le C \left(\int_0^\infty |u'|^2 \,\mathrm{d}x + \int_0^1 |u|^2\,\mathrm{d}x\right),
\end{split}
\end{align}
where $C$ depends on $\xi$, but not on $u$. 
\\ \textbf{(iv)} The proof is a simple modification of the proof of assertion \textbf{(iii)}.
\\
\textbf{(v)} The proof follows immediately from the Rellich-Kondrachov theorem, see \cite[Theorem 6.3]{adams2003sobolev}
\end{proof}

\begin{rem}
 Inequality \eqref{eq: two dimensional Hardy Solomyak} is equivalent to the scalar inequality 
\begin{equation}\label{eq: two dimensional hardy - scala variant}
	\int\limits_0^\infty \frac{|u(t)|^2}{t\left(1+\ln^2(t)\right)}\,\mathrm{d}t \le 4\int\limits_0^\infty t(u'(t))^2\,\mathrm{d}t, \qquad  u(1)=0.
\end{equation}
\end{rem}

\begin{cor}\label{cor: Estimate Vu^2}
It follows from Proposition \ref{prop: Homogenous Sobolev} \textbf{(iv)} and \textbf{(v)} that if $V$ satisfies \eqref{cond: relatively form bounded} and \eqref{cond: decay at infinity},
then there exists a constant $C>0$, such that for all $u\in \tilde{H}^1(\mathbb{R}^d)$ we have
\begin{equation}
	\int_{\mathbb{R}^d}|V(x)||u(x)|^2\,\mathrm{d}x \le C\Vert u\Vert_{\tilde{H}^1}^2.
\end{equation}
\end{cor}

\begin{cor}\label{lem: two dimensional Hardy radially symmetric functions}
For any function $u\in \tilde H^1(\mathbb{R}^2)$ with $\supp(u)\subset \{x\in \mathbb{R}^2:|x|<1\}$ and any constant $\nu\in(0,1)$ we have
\begin{equation}\label{eq: hardy inequality scheibe}
		\int_{\{|x|\ge \nu\}} \frac{|u|^2}{|x|^2(1+\ln^2(|x|))} \,\mathrm{d}x \le 4 \int_{\{|x|\ge \nu\}} |\nabla u|^2 \,\mathrm{d}x.
\end{equation}
\end{cor}
\begin{proof}[Proof of Corollary \ref{lem: two dimensional Hardy radially symmetric functions}]
	Since the function $v(x)=\vert x\vert^{-2}\left(1+(\ln\left(|x|\right))^2\right)^{-1}$ is spherically symmetric, due to the rearrangement inequality it suffices to show that the inequality holds for spherically symmetric functions. Let $u\in \tilde H^1(\mathbb{R}^2)$ be spherically symmetric with $\supp(u)\subset \{x:|x|< 1\}$. Then the function $\tilde u$ given by
	\begin{equation}
		\tilde{u}(x) = \begin{cases}
			u(x) & \text{if} \ |x|\ge \nu,\\
			u(\nu) & \text{if} \ |x|< \nu
		\end{cases}
	\end{equation}
	is also an element of $ \tilde H^1(\mathbb{R}^2) $.
Applying the two-dimensional Hardy inequality \eqref{eq: two dimensional Hardy Solomyak} to the function $\tilde u$, using that $\tilde{u}$ is constant for $|x|\le \nu$ and that $\tilde{u}$ and $u$ coincide for $|x|\ge \nu$ proves \eqref{eq: hardy inequality scheibe}.
\end{proof}
\begin{rem}
Analogously to the proof of Corollary \ref{lem: two dimensional Hardy radially symmetric functions} one can see that if $d\ge 3$, then for any function $u\in \tilde{H}^1(\mathbb{R}^d)$ we have
\begin{equation}\label{eq: Hardy exterior domain d>=3}
		\int_{\{|x|\ge \nu\}} \frac{|u|^2}{|x|^2} \,\mathrm{d}x \le \frac{4}{(d-2)^2} \int_{\{|x|\ge \nu\}} |\nabla u|^2 \,\mathrm{d}x.
\end{equation}
\end{rem}

\section{Necessary and sufficient conditions for virtual levels}\label{Appendix Virtual Levels}
In this section we prove Theorem \ref{lem no virtual level subtract d=1}, which is stated for one-particle Schr\"odinger operators in dimension one or two. Afterwards, we give an analogue of this theorem for the multi-particle case.
\begin{proof}[Proof of Theorem \ref{lem no virtual level subtract d=1}]
We only need to prove that the absence of a virtual level of $h$ implies that \eqref{eq: no virtual level subtract d=1} does not hold. The proof of the other direction follows from Theorem \ref{thm: abstraktes Theorem} and the variational principle. 
\\Let $d=1$. Note that we can assume that $\mathcal{U}(x)=-(1+|x|)^{-2}$.
For $\psi \in H^1(\mathbb{R})$ we write
\begin{equation}
\psi_0(x)= \psi(x)- \psi(0).
\end{equation}
 Then, $\psi_0(0)=0$ and we can apply Hardy's inequality on the half line $\mathbb{R}_+$ to obtain
\begin{equation}
\int_0^\infty \frac{|\psi_0(x)|^2}{|x|^2}\, \mathrm{d}x \leq 4 \int_0^\infty |\psi_0'(x)|^2 \, \mathrm{d}x =  4 \int_0^\infty |\psi'(x)|^2 \, \mathrm{d}x.
\end{equation}
Furthermore, we have
\begin{align}\label{1: abstract virtual level}
\begin{split}
\langle V \psi,\psi \rangle &= \int\limits_{-\infty}^\infty V(x)|\psi(0)|^2\, \mathrm{d}x + \int\limits_{-\infty}^\infty V(x)|\psi_0(x)|^2\, \mathrm{d}x + 2\mathrm{Re}\int\limits_{-\infty}^\infty V(x)\psi(0)\psi_0(x)\, \mathrm{d}x 
\\ &\geq \int\limits_{-\infty}^\infty V(x)|\psi(0)|^2\, \mathrm{d}x + \int\limits_{-\infty}^\infty V(x)|\psi_0(x)|^2\, \mathrm{d}x - 2\int\limits_{-\infty}^\infty |V(x)| |\psi(0)\psi_0(x)|\, \mathrm{d}x.
\end{split}
\end{align}
Note that for any $\delta>0$
\begin{equation}
2|\psi(0)\psi_0(x)| \leq \delta |\psi(0)|^2 + \delta^{-1}|\psi_0(x)|^2,
\end{equation}
which together with \eqref{1: abstract virtual level} implies
\begin{align}\label{eq: equivalent definition virtual level estimate Vpsi}
\begin{split}
\langle V \psi,\psi \rangle &\geq |\psi(0)|^2\int\limits_{-\infty}^\infty \left( V(x) - \delta|V(x)|\right)\, \mathrm{d}x + \int\limits_{-\infty}^\infty |\psi_0(x)|^2 \left( V(x)-\delta^{-1} |V(x)| \right) \, \mathrm{d}x
\\ &\geq |\psi(0)|^2 \int\limits_{-\infty}^\infty (V(x)-\delta|V(x)|)\, \mathrm{d}x -(1+\delta^{-1}) \int\limits_{-\infty}^\infty |V(x)||\psi_0(x)|^2\, \mathrm{d}x
\\ &\geq |\psi(0)|^2 \int\limits_{-\infty}^\infty (V(x)-\delta|V(x)|)\, \mathrm{d}x - C(1+\delta^{-1}) \Vert \psi_0\Vert^2_{\tilde{H}^1},
\end{split}
\end{align}
where in the last estimate we used Corollary \ref{cor: Estimate Vu^2}. Since $\psi_0(0)=0$, we have $\Vert \psi_0\Vert^2_{\tilde{H}^1}\le C\Vert \psi_0'\Vert^2$. This, together with \eqref{eq: equivalent definition virtual level estimate Vpsi} yields
\begin{equation}
	\langle V\psi,\psi\rangle\ge  |\psi(0)|^2 \int\limits_{-\infty}^\infty (V(x)-\delta|V(x)|)\, \mathrm{d}x-C(\delta)\int_{-\infty}^\infty |\psi_0'(x)|^2\,\mathrm{d}x.
\end{equation}
 Since $\int_{-\infty}^\infty V(x)\, \mathrm{d}x>0$, we can choose the constant $\delta>0$ sufficiently small, such that $\int_{-\infty}^\infty(V(x)-\delta|V(x)|)\, \mathrm{d}x \geq \frac{1}{2}\int_{-\infty}^\infty V(x)\, \mathrm{d}x=:C_0$. This, together with $\psi'(x)=\psi_0'(x)$ implies
\begin{equation}\label{1: |phi(0)|<}
|\psi(0)|^2 \leq C_0^{-1}\langle V\psi,\psi \rangle + C_1(\delta)\Vert \psi' \Vert^2
\end{equation}
for some constant $C_1(\delta)>0$ which depends on $V$ and $\delta$ only. This yields
\begin{align}\label{1: eps_1}
\begin{split}
\varepsilon_1 \int\limits_{-\infty}^\infty \frac{|\psi(x)|^2}{1+x^2}\, \mathrm{d}x &\leq 2\varepsilon_1 |\psi(0)|^2 \int\limits_{-\infty}^\infty \frac{1}{1+x^2}\, \mathrm{d}x + 2 \varepsilon_1 \int\limits_{-\infty}^\infty \frac{|\psi_0(x)|^2}{1+x^2}\, \mathrm{d}x
\\ &\leq \varepsilon_1 \left( 2\pi |\psi(0)|^2 + 8 \Vert \psi'\Vert^2 \right)\leq \varepsilon_1 \left( 2\pi C_0^{-1} \langle V\psi,\psi\rangle+C_2(\delta)\Vert \psi' \Vert^2 \right),
\end{split}
\end{align}
where $C_2(\delta)=C_1(\delta)+8$. We distinguish between two cases:
\\ \textbf{(i)} If $2\pi C_0^{-1} \langle V\psi,\psi \rangle < C_2(\delta)\Vert \psi' \Vert^2$, then \eqref{1: eps_1} yields
\begin{equation}
\varepsilon_1 \int\limits_{-\infty}^\infty \frac{|\psi(x)|^2}{1+x^2}\, \mathrm{d}x \leq 2\varepsilon_1C_2(\delta)\Vert \psi' \Vert^2.
\end{equation}
Now since $h$ does not have a virtual level, we can choose $\varepsilon_1>0$ sufficiently small
 to conclude that
\begin{equation}\label{eq: no virtual level subtract}
	\langle h\psi,\psi\rangle -\varepsilon_1 \int_{-\infty}^\infty \frac{|\psi(x)|^2}{1+x^2}\,\mathrm{d}x\ge 0.
\end{equation}
\\ \textbf{(ii)} If $2\pi C_0^{-1} \langle V\psi,\psi \rangle \geq C_2(\delta)\Vert \psi' \Vert^2$, we have $\langle V\psi,\psi \rangle > 0$ and
\begin{equation}
\varepsilon_1 \int\limits_{-\infty}^\infty \frac{|\psi(x)|^2}{1+x^2}\, \mathrm{d}x \leq 4\varepsilon_1 \pi C_0^{-1} \langle V\psi,\psi \rangle.
\end{equation}
By choosing $0<\varepsilon_1< (4\pi)^{-1}C_0$ we obtain
\begin{equation}
	\langle h\psi,\psi\rangle - \varepsilon_1 \int\limits_{-\infty}^\infty \frac{|\psi(x)|^2}{1+x^2}\, \mathrm{d}x\ge \Vert \psi'\Vert^2\ge 0.
\end{equation} 
This implies \eqref{eq: no virtual level subtract} and therefore the statement of Theorem \ref{lem no virtual level subtract d=1} for the case $d=1$.\\
Now we assume that $d=2$. For $\psi\in H^1(\mathbb{R}^2)$ we write $ \psi_0(x)=\psi(x) - a_0$, where 
\begin{equation}
	a_0 = \frac{1}{2\pi}\int_{0}^{2\pi} \psi(1,\theta)\,\mathrm{d}\theta.
\end{equation}
Then $\int_{0}^{2\pi} \psi_0(1,\theta)\,\mathrm{d}\theta=0$  and thus we can apply the two-dimensional Hardy inequality \eqref{eq: two dimensional Hardy Solomyak} to the function $\psi_0$. Proceeding as in the proof of the one-dimensional case yields the statement for $d=2$ and therefore completes the proof of Theorem \ref{lem no virtual level subtract d=1}.
\end{proof}
Now we extend Theorem \ref{lem no virtual level subtract d=1} to the case of multi-particle Schr\"odinger operators. 
\begin{thm}\label{thm: equivalent definition virtual level multiparticle}
Let $H$ be the Schr\"odinger operator corresponding to a system of $N\ge 2$ one- or two-dimensional particles, where the potentials $V_{ij}\neq 0$ satisfy \eqref{cond: relatively form bounded} and \eqref{cond: decay at infinity} and let $H\ge 0$. 
 Then $H$ has a virtual level at zero if and only if the following two assertions hold.
\begin{enumerate}
\item[\textbf{(i)}] There exists an $\varepsilon_0>0$, such that for any cluster $C$ with $1<|C|<N$ we have
\begin{equation}\label{eq: subsystems non-negative}
	H[C]-\varepsilon_0 \left(1+|q[C]|^2_m(\ln(|q[C]|_m)^2\right)^{-1} \ge 0.
\end{equation}
\item[\textbf{(ii)}] For any $\varepsilon>0$ we have
\begin{equation}\label{eq: spectrum perturbation Hardy term}
	\inf\mathcal{S}\left(H-\varepsilon\left(1+|x|_m^2\ln^2(|x|_m)\right)^{-1}\right)<0.
\end{equation} 
\end{enumerate}
\end{thm}
\begin{proof}[Proof of Theorem \ref{thm: equivalent definition virtual level multiparticle}]
For $N=2$ the statement was proved in Theorem \ref{lem no virtual level subtract d=1}. Now assume that $N\ge 3$. First we prove that if $H$ has a virtual level at zero, then \eqref{eq: spectrum perturbation Hardy term} is true. According to remark \textbf{(ii)} after Theorem \ref{thm d=1 N=3} we know that if $d=1, \, N\ge 3 $ or $d=2,\,N\ge 4$ zero is an eigenvalue of $H$. Taking the corresponding eigenfunction as a trial function shows that \eqref{eq: spectrum perturbation Hardy term} holds for any $\varepsilon>0$. For $d=2,\, N=3$ we do not know if zero is an eigenvalue. However, by Theorem \ref{thm d=2 N=3} we know that there is a function $\varphi_0\in \tilde{H}^1(X_0)$ with $\Vert \nabla_0\varphi_0\Vert^2 +\langle V\varphi_0,\varphi_0\rangle =0$, which yields \eqref{eq: spectrum perturbation Hardy term}.
\\ \ \\
In the rest of the proof we will use induction in the number of particles. Assume that the system has $N\geq 3$ particles and that the theorem holds for all system with the number of particles less or equal to $N-1$. 
\\If $H$ has a virtual level and condition \eqref{eq: subsystems non-negative} does not hold, then there exists at least one cluster $C$ in this system with $1<|C|<N$, such that for this cluster $C$ and all $\varepsilon>0$ condition \eqref{eq: subsystems non-negative} does not hold as well. Among such clusters we choose one with the smallest number of particles and denote it by $C_0$. Then we have
\begin{equation}\label{eq: B16}
\inf\mathcal{S}\left(H[C_0]-\varepsilon \left(1+|q[C_0]|^2_m(\ln(|q[C_0]|_m)^2\right)^{-1}\right)<0 \qquad \text{for any} \ \varepsilon>0.
\end{equation}
If $C_0$ consists of only two particles, then by Theorem \ref{lem no virtual level subtract d=1} this condition implies that $H[C_0]$ has a virtual level at zero. However, by the remark after Definition \ref{def: Vitrtual levels multiparticle} Hamiltonians of non-trivial clusters can not have virtual levels at zero if the Hamiltonian of the whole system has a virtual level. Therefore, $C_0$ must consist of at least three particles. 
Now since $C_0$ is the smallest cluster for which \eqref{eq: subsystems non-negative} does not hold for any small $\varepsilon>0$, for each cluster $\tilde{C}\subsetneq C_0$ with $|\tilde{C}|>1$ we have 
\begin{equation}\label{eq: B18}
H[\tilde C]-\varepsilon_0 \left(1+|q[\tilde C]|^2_m(\ln(|q[\tilde C]|_m))^2\right)^{-1} \ge 0 
\end{equation}
for some $\varepsilon_0>0$. Since $|C_0|<N-1$, \eqref{eq: B16}, \eqref{eq: B18} and the induction assumption yield that $H[C_0]$ has a virtual level, which according to the remark after the Definition \ref{def: Vitrtual levels multiparticle} contradicts the assumption that $H$ has a virtual level. 
\\To complete the proof of the theorem we have to show that if condition \eqref{eq: subsystems non-negative} and \eqref{eq: spectrum perturbation Hardy term} are fulfilled, then $H$ has a virtual level.
\\
At first, we prove that condition \textbf{(i)} of Definition \ref{def: Vitrtual levels multiparticle} is fulfilled, namely that there exists a constant $\varepsilon_0\in(0,1)$, such that 
\begin{equation}\label{eq: B19}
\inf\mathcal{S}_{\mathrm{ess}} \left(H+\varepsilon_0\Delta_0\right)=0.
\end{equation} 
Recall that due to the remark after Definition \ref{def: Vitrtual levels multiparticle} to prove \eqref{eq: B19} it suffices to show that for any $C$ with $1<|C|<N$ the operator $H[C]$ does not have a virtual level. Assume for contradiction that there exists a cluster $C_1$ with $1<|C_1|<N$, such that $H[C_1]$ has a virtual level at zero. Then, due to the induction assumption we have
\begin{equation}
	\inf \mathcal{S}\left(H[C_1]-\varepsilon\left(1+|q[ C_1]|^2_m(\ln(|q[ C_1]|_m))^2\right)^{-1} \right)<0.
\end{equation}
This is a contradiction to \eqref{eq: subsystems non-negative}. Hence, condition \textbf{(i)} of Definition \ref{def: Vitrtual levels multiparticle} is fulfilled.
\\
It remains to prove that if conditions \eqref{eq: subsystems non-negative} and \eqref{eq: spectrum perturbation Hardy term} of Theorem \ref{thm: equivalent definition virtual level multiparticle} hold, then condition \textbf{(ii)} of Definition \ref{def: Vitrtual levels multiparticle} is fulfilled, namely
\begin{equation}\label{eq_ def vl (ii)}
\inf\mathcal{S}\left(H+\varepsilon\Delta_0\right)<0 \qquad  \text{for any} \ \varepsilon \in (0,1).
\end{equation}
If $\dim(X_0)\ge 3$, we can use Hardy's inequality to conclude that \eqref{eq_ def vl (ii)} holds. If $\dim(X_0)<3$, i.e. the system consists of three one-dimensional particles, \eqref{eq: spectrum perturbation Hardy term} implies that for any $n\in \mathbb{N}$ the operator $H-n^{-1}\left(1+|x|_m^2(\ln(|x|_m))^2\right)^{-1}$ has a negative eigenvalue. We take a sequence of eigenfunctions $\psi_n$ corresponding to these eigenvalues, normalized by $\Vert \psi_n\Vert_{\tilde{H}^1} =1$. Applying the same arguments as in the proof of Theorem \ref{thm d=1 N=3} we see that $\psi_n$ converges in $L^2(X_0)$ to a function $\psi_0$ which is an eigenfunction of the operator $H$ corresponding to the eigenvalue zero. For this function we have
\begin{equation}
	(1-\varepsilon) \Vert \nabla_0 \psi_0\Vert^2+\langle V\psi_0,\psi_0\rangle=-\varepsilon\Vert \nabla_0\psi_0\Vert^2 <0
\end{equation}
for any $\varepsilon>0$. This proves that condition \textbf{(ii)} of Definition \ref{def: Vitrtual levels multiparticle} is fulfilled and completes the proof of Theorem \ref{thm: equivalent definition virtual level multiparticle}.
\end{proof}

\section{A sufficient condition for finiteness of the discrete spectrum}\label{App: Technical Lemmas}
In this section we give a criterion for the finiteness of the number of negative eigenvalues, which we used in the proofs of Theorem \ref{thm: No Efimov}, Theorem \ref{thm: No Efimov N=3} and Theorem \ref{thm: No Efimov n=2 N=3}. This criterion, in a slightly different form, is due to G. Zhislin and is a part of the proof of the main result in \cite{Zhislin1}. For the convenience of the reader we give it here.
\begin{lem}\label{lem: Abstraktes Lemma endlich viele Eigenwerte}
Let $h=-\Delta+V$ in $L^2(\mathbb{R}^k),\ k\in \mathbb{N},$ where $V$ satisfies \eqref{cond: relatively form bounded}. Assume there exist constants $\beta, \varepsilon,\, b>0$, such that
\begin{equation}\label{A22}
\langle h\psi,\psi\rangle - \varepsilon \langle |x|^{-\beta}\psi,\psi \rangle \geq 0
\end{equation}
holds for any $\psi \in H^1(\mathbb{R}^k)$ with $\supp \psi \subset \{x\in \mathbb{R}^k, \ |x|\geq b\}$. Then the following assertions hold.
\begin{enumerate}
\item[$\textbf{(i)}$] $\inf \mathcal{S}_{\mathrm{ess}}(h) \geq 0$.
\item[$\textbf{(ii)}$] The operator $h$ has at most a finite number of negative eigenvalues.
\item[$\textbf{(iii)}$] Zero is not an infinitely degenerate eigenvalue of $h$.
\end{enumerate}
\end{lem}
To prove Lemma \ref{lem: Abstraktes Lemma endlich viele Eigenwerte} we use the following
\begin{lem}\label{Cutoff ball}
Assume that $V$ satisfies \eqref{cond: relatively form bounded}.	Let $\beta>0$, $\varepsilon>0$ and $\tilde b>b>0$. Then there exist a constant $C(\varepsilon,\beta)$ and a function $\chi_1\in C^1(\mathbb{R}^k), \ 0\le \chi_1\le 1,$ with
	\begin{equation}
		\chi_1(x) = \begin{cases}
			1, & \vert x\vert \le b,\\
			0, & \vert x\vert \ge \tilde b,
		\end{cases}
	\end{equation}
such that for all $\psi \in H^1(\mathbb{R}^k)$ we have
	\begin{equation}\label{eq: Cutoff ball}
		\langle h\psi, \psi\rangle\ge  \langle h\psi \chi_1, \psi \chi_1\rangle-C(\varepsilon,\beta)\Vert \psi \chi_1\Vert^2 + \langle h \psi \chi_2, \psi \chi_2\rangle -\varepsilon \Vert \vert x\vert^{-\beta} \psi \chi_2\Vert^2_{\{b\le |x|\le\tilde b\}},
	\end{equation}
	where $\chi_2=\sqrt{1-\chi_1^2}$.
\end{lem}
\begin{proof}[Proof of Lemma \ref{Cutoff ball}]
 Let $\beta,\varepsilon>0$ and $b,\tilde b>0$ with $\tilde b>b$ be fixed. Furthermore, let $u:\mathbb{R}_+\rightarrow [0,1]$ be a $C^1$-function, such that $u(t)=1, \, t\le b$ and $u(t)=0,\, t\ge \tilde{b}$. We assume that $u$ is strictly monotonically decreasing on $(b,\tilde b)$. Let $v=\sqrt{1-u^2}$. We choose $u$ in such a way that $v'(t)(1-v^2(t))^{-\frac{1}{2}}\rightarrow 0$ as $t\rightarrow \tilde b_-$. For $x\in \mathbb{R}^k$ let
\begin{equation}	
	 \chi_1(x)=u(|x|), \qquad \chi_2(x)=v(|x|).
\end{equation}	 
Then we have
	\begin{equation}\label{eq: compute sum of squared gradients}
		|\nabla \chi_1|^2+|\nabla\chi_2|^2 = \frac{|\nabla \chi_2|^2}{1-\chi_2^2}=\frac{(v'(|x|))^2}{1-v^2(|x|)}.
\end{equation}
Since $v'(|x|)(1-v^2(|x|))^{-\frac{1}{2}}\rightarrow 0$ as $|x|\rightarrow \tilde b_-$ and $v(|x|)$ is close to one in a vicinity of $|x|=\tilde b$, we can choose $b'$ so close to $\tilde b$ that
\begin{equation}
	\frac{(v'(|x|))^2}{1-v^2(|x|)} \le \varepsilon v^2(|x|)|x|^{-\beta},\qquad b'\le |x|\le \tilde b.
\end{equation}
This, together with \eqref{eq: compute sum of squared gradients} implies \begin{equation}\label{eq: estimate b'<x<tilde b}
|\nabla \chi_1|^2+|\nabla\chi_2|^2 \le \varepsilon \chi_2^2(x)|x|^{-\beta}, \qquad  b'\le |x|\le \tilde b.
\end{equation}
Now we estimate $|\nabla \chi_1|^2+|\nabla\chi_2|^2$ for $b\le |x|\le b'$. Recall that $u(t)> u(b')>0$ for $b<t<b'$. Hence, we get  
\begin{equation}\label{eq: estimate b<x<b'}
	\frac{(v'(|x|))^2}{1-v^2(|x|)}\le C u^2(|x|)|x|^{-\beta},\qquad b\le |x|\le b'
\end{equation}	 
for some $C>0$ which depends on $b'$ (which itself depends on $\varepsilon$ and $\beta$). Due to the IMS formula we have
\begin{equation}
\langle h\psi,\psi\rangle =\langle h\psi\chi_1,\psi\chi_1\rangle +\langle h\psi\chi_2,\psi\chi_2\rangle -\int \left(|\nabla \chi_1|^2+|\nabla \chi_2|^2\right)|\psi|^2\,\mathrm{d}x.
\end{equation}
This, together with \eqref{eq: estimate b'<x<tilde b} and \eqref{eq: estimate b<x<b'} completes the proof of Lemma \ref{Cutoff ball}.
\end{proof}
Now we turn to the
\begin{proof}[Proof of Lemma \ref{lem: Abstraktes Lemma endlich viele Eigenwerte}]
We construct a finite-dimensional subspace $M\subset L^2(\mathbb{R}^k)$, such that $\langle h \psi,\psi\rangle > 0$ holds for any $\psi \in H^1(\mathbb{R}^k)$, $\psi\neq 0$ which is orthogonal to $M$. Let $\varepsilon,\beta, b>0$, such that \eqref{A22} is fulfilled. Let $\chi_1$ and $\chi_2$ be functions according to Lemma \ref{Cutoff ball}. Then by assumption of the lemma for any function $\psi\in H^1(\mathbb{R}^k)$
\begin{align}
\begin{split}
	\langle h\psi,\psi\rangle &\ge  \langle h\psi \chi_1, \psi \chi_1\rangle-C(\varepsilon,\beta)\Vert \psi \chi_1\Vert^2 + \langle h \psi \chi_2, \psi \chi_2\rangle -\varepsilon \Vert \vert x\vert^{-\beta} \psi \chi_2\Vert^2_{\{b\le |x|\le\tilde b\}}\\
	&\ge  \langle h\psi \chi_1, \psi \chi_1\rangle-C(\varepsilon,\beta)\Vert \psi \chi_1\Vert^2 ,
	\end{split}
\end{align}
because $\supp(\chi_2)\subset \{x\in \mathbb{R}^k:|x|\ge b\}$. Thus, to prove statements \textbf{(i)}-\textbf{(iii)} it suffices to show that
\begin{equation}
\langle h\psi \chi_1, \psi \chi_1\rangle-C(\varepsilon,\beta)\Vert \psi \chi_1\Vert^2\ge 0
\end{equation}
 holds for any function $\psi\in H^1(\mathbb{R}^k)$ with $\psi\perp M$ (in $L^2(\mathbb{R}^k)$) for some finite-dimensional space $M\subset H^1(\mathbb{R}^k)$. By condition \eqref{cond: relatively form bounded} we get
\begin{equation}
	 \langle h\psi \chi_1, \psi \chi_1\rangle-C(\varepsilon,\beta)\Vert \psi \chi_1\Vert^2  \ge (1-\varepsilon)\Vert \nabla \left(\psi\chi_1\right)\Vert^2-C'(\varepsilon,\beta)\Vert \psi\chi_1\Vert^2
\end{equation}
for some $C'(\varepsilon,\beta)>0$. For $l \in \mathbb{N}$ let
\begin{equation}
	M_l := \left\{\varphi_1 \chi_1, \dots, \varphi_l \chi_1\right\},
\end{equation}
where $\{\varphi_1, \dots, \varphi_l\}$ is an orthonormal set of eigenfunctions corresponding to the $l$ lowest eigenvalues of the Laplacian, acting on $L^2\left(\{|x|\le \tilde b\}\right)$ with Dirichlet boundary conditions. For $\psi \perp M_l$ we have $\psi \chi_1 \perp \varphi_1, \dots \varphi_l$, which for sufficiently large $l$ implies
\begin{equation}
	\Vert \nabla (\psi \chi_1) \Vert^2 \ge  \left(1-\varepsilon\right)^{-1}C'(\varepsilon,\beta) \Vert \psi \chi_1\Vert^2.
\end{equation}
Therefore, we conclude $L[\psi \chi_1] > 0$. This proves statements \textbf{(i)}-\textbf{(iii)} of Lemma \ref{lem: Abstraktes Lemma endlich viele Eigenwerte}. 
\end{proof}

\section*{data availability statement}
Data sharing is not applicable to this article as no new data were created or analyzed in this study.

\bibliography{Bib}{}
\bibliographystyle{abbrv}
\end{document}